\documentclass[11pt,letterpaper,english]{article}

\usepackage{babel}
\usepackage{authblk}

\usepackage{algorithm}
\usepackage{algorithmicx}
\usepackage{algpseudocode}
\usepackage{url}
\usepackage{caption}

\usepackage{hyperref}

\newcommand{\Input}{\item[\textbf{Input:}]}

\usepackage[toc,page]{appendix}

\usepackage{amssymb,amsmath,amsfonts,amsthm}

\usepackage[margin=2cm]{geometry}
\usepackage{amsfonts,amssymb,tabularx}
\usepackage{fancyhdr}
\usepackage{graphicx}

\usepackage{float}
\usepackage[nottoc,notlot,notlof]{tocbibind}
\usepackage{color}
\usepackage{enumerate}

\usepackage{cleveref}
\usepackage{algorithm,algpseudocode}
\usepackage{lscape}
\usepackage{multirow}
\usepackage{booktabs}
\usepackage{threeparttable}
\usepackage[rightcaption]{sidecap}
\usepackage{wrapfig,lipsum}
\usepackage{array}

\theoremstyle{remark}

\newtheorem{lemma}{Lemma}

\usepackage[
    backend=biber,
    style=apa,
    uniquename=false,
    uniquelist=false
]{biblatex}

\DeclareLanguageMapping{english}{english-apa}
\renewbibmacro{in:}{}
\DeclareFieldFormat{pages}{#1}

\usepackage{pict2e}

\usepackage{pgfkeys}
\usepackage{xcolor}

\usepackage{subcaption}
\usepackage{pgfgantt}

\usepackage{verbatim}
\usepackage{pgfplots}
\usepackage{etoolbox}

\usepackage{pstricks}
\usepackage{longtable}

\date{\normalsize \today}

\pgfplotsset{compat=1.18}
\begin{document}

\title{Waiting time analysis in a finite‑capacity multi‑server systems with dynamic priorities, dynamically evolving customer types, and abandonment}

\author[a,b]{M. Abdullah Khokhar\footnote{Corresponding author: Email address: MuhammadAbdullah.Khokhar@utas.edu.au (M. Abdullah Khokhar)}}
\author[a,c]{Ma{\l}gorzata M. O'Reilly}
\author[d]{Richard Turner}
\affil[a]{School of Natural Sciences, Discipline of Mathematics, University of Tasmania, Australia}
\affil[b]{email: MuhammadAbdullah.Khokhar@utas.edu.au}
\affil[c]{email: Malgorzata.OReilly@utas.edu.au}
\affil[d]{School of Medicine, University of Tasmania, Australia, email: Richard.Turner@utas.edu.au}
\maketitle

\begin{abstract} 
In many service systems, an estimation of customers' waiting times for the service can assist in decision making focused on enhancing the operational efficiency, improving the customers' experience, and ensuring efficient resource allocation. In this paper, we study the customers' waiting times in a finite-capacity service system with a finite number of parallel servers and a shared waiting area. We consider two customer types, Type~$1$ and Type~$2$, with dynamic admission priorities, dynamically evolving customer type, and abandonment.

We model the system under such assumptions using a continuous-time Markov chain (CTMC) and develop a methodology based on Krylov subspace {\em approximation} methods to evaluate the conditional waiting time distributions of Type~$1$ and Type~$2$ customers in the system. This methodology (CTMC-Krylov) offers a scalable computational approach that is well suited for analysing large complex systems. 

Next, we model this system using a quasi-birth-and-death (QBD) process and derive analytical expressions building on matrix-analytic methods to evaluate the conditional and long-run waiting time distributions using {\em recursion}.

We illustrate the practical applicability of our models in a hospital system through a suite of numerical examples based on a large dataset obtained from a tertiary referral hospital in Australia, considering two types of patients, complex (Type~$1$) and other (Type~$2$). We compare the conditional waiting time distributions of complex and other patients obtained from the CTMC-Krylov method with the QBD-based approach in the scenarios when complex patients have absolute priority for admission. The comparison shows that CTMC model together with Krylov subspace methods gives results similar to those obtained from the QBD model, with an error tolerance less that $10^{-7}$.

Further, in our examples, we evaluate the long-run waiting time distributions using the CTMC model and the stationary distribution obtained from the QBD model. We analyse the impact of admission priorities given to the patients in the waiting area and give valuable insights to support operational decision-making in large complex service systems, such as hospital.
\end{abstract}



\section{Introduction}\label{sec:introduction}

Waiting time is an important performance measure in many service systems and a useful indicator of system efficiency and quality of service. In many service systems, customers arrive randomly and compete for limited service resources. When all servers are occupied, customers wait in a queue, and their waiting time may depend on their position in the queue, the service discipline, and the overall system dynamics. Queues may grow excessively large, leading to unnecessarily long waiting times. According to~\textcite{komashie2015integrated}, long waiting times are directly associated with lower customers' satisfaction. In service systems such as healthcare, prolonged waiting times may lead to deterioration in patients’ conditions and adverse clinical outcomes.~\textcite{masoumi2022m} noted that even a marginal increase in waiting times can lead to disproportionately severe consequences, including patients' mortality. Therefore, efforts are made to estimate customers' waiting times when making decisions with the aim of optimising resource allocation, improving throughput, and enhancing customer experience, see~\textcite{worlitz2019waiting}.

The analysis of waiting times in multi-class, finite-capacity service systems with customer priorities has received considerable attention in the queueing theory literature. Many researchers focus on reducing or eliminating waiting times for customers requiring urgent service by modeling systems using preemptive priority queueing models. Under preemptive priority, a high‑priority customer receives service immediately upon arrival by interrupting the service of a low‑priority customer, if the system is full. This mechanism ensures that waiting times for high‑priority customers are minimally affected by the congestion.

Several studies have employed such preemptive priority multi-class queueing models to analyse waiting times in service systems. For example,~\textcite{Wang2015733} studied a two-class preemptive queueing system $M/M/c$, and developed analytical and numerical methods to compute waiting times distributions by transforming the original two-dimensional Markov process to one-dimensional through priority‑based state aggregation.~\textcite{Fajardo2017255} developed a Laplace–Stieltjes transform (LST) based approach using the maximal priority process to derive steady-state waiting time distributions in preemptive accumulating priority queues.~\textcite{ghanbari2022novel} modeled an emergency department as a preemptive priority queue $M/M/c$ and derived analytical expressions for waiting times of different priority classes of patients.

On the other hand, some authors argue that preemptive priority may be unreasonable in service systems where service interruption is costly, impractical, or undesirable, see~\textcite{hou2020using}. In such systems, non‑preemptive priority queueing models are often adopted, where priority may influence the order of service initiation but ongoing services are allowed to complete. For example,~\textcite{hou2020using} modeled a healthcare system as a non-preemptive priority queue in which patients are classified according to urgency, allowing high-priority patients to receive earlier access to service without interrupting patients already in treatment, and derived analytical expressions for waiting times for different priority classes.

While priority-based service disciplines have been widely adopted in waiting time analysis, they are not sufficient to fully capture the dynamics of real-world systems. In practice, customers may become impatient and abandon the queue if their waiting time is too long. The abandonment introduces additional complexity in the analysis of waiting time distributions. To address such complexities, many researchers have developed analytical and numerical approaches for evaluating waiting time distributions in service systems with abandonment~\parencite{Yoshiaki2019524, Evdokimova201833, Mandelbaum2002149}.

We note that a powerful approach to model multi‑class service systems with admission priorities and abandonment is the class of models referred to as the quasi‑birth‑and‑death (QBD) processes developed in the theory of matrix‑analytic methods, see~\textcite{Neuts_1981},~\textcite{he2014fundamentals}, and~\textcite{latouche1999introduction}. Matrix‑analytic methods are a collection of applied probability techniques for analysing structured Markov chains, particularly QBD processes, by expressing stationary and transient performance measures in terms of matrix equations and their numerical solutions. These methods are widely employed as computational tools in complex queueing models, see~\textcite{WU2026102555},~\textcite{2016JF},~\textcite{phung2010simpleNEW}, and~\textcite{aksamit2024sensitivities}.

QBD processes allow the system to be described in terms of both its size and its composition and for this reason have been used to analyse priority queueing systems and derive key performance measures . For instance,~\textcite{Rastpour20221693} modeled two-class $M/M/s+M(\text{Pr})$ queueing system with preemptive priorities and reneging, using level-dependent QBD process. They developed an algorithm to evaluate waiting time distributions for different priority classes, and demonstrate how priority rules influence waiting time.~\textcite{Drekic201587} modeled a deceased‑donor transplant waiting list using a QBD process and explicitly analysed patient waiting time distributions. In their analysis, they applied matrix‑analytic methods and provided insights into how system parameters and allocation rules affect waiting times of patients on the transplant waiting list.

Many studies that incorporate priority mechanisms, typically impose absolute priority rules, where one class of customers is always served before another. The queueing models with absolute priority rule for one class of customers have been studied by~\textcite{almehdawe2019optimization},~\textcite{Drekic201587},~\textcite{Rastpour20221693},~\textcite{d2025strategic},~\textcite{KhademiLiu2024}, and~\textcite{CaoXie2016}. Some of these studies also consider abandonment and customers' type change.

The studies on two-class service systems have shortcomings in the following main areas. 
\begin{itemize}
    \item Most studies that consider preemptive or non-preemptive priorities for the priority customers typically impose absolute priority in the sense that priority customers must be served before others. This assumption may be unreasonable. First, absolute priority rules can result in over-servicing priority customers while under-serving other, leading to inefficiencies and inequities in access, as argued by~\textcite{Stanford2014297}. Second, absolute priority rules may be appropriate for minimising overall mean waiting times, but they can lead to very long waiting times for the customers with low priority, as argued by~\textcite{Haviv2016505}. The long waiting times may be inappropriate for service systems such as healthcare, where delays in service for low priority patients can lead to deterioration in their health and can increases the number of patients leaving without being seen. 
    \item Many studies, such as~\textcite{CaoXie2016} and~\textcite{KhademiLiu2024} that consider customer class change during the waiting period, also assume fixed or absolute priority for high-priority customers. In such models, low-priority customers in the queue may be unable to receive service until they transition into a high-priority class. In a service system such as healthcare, this means that less critical patients may not receive timely care unless their health deteriorates to a critical level. Consequently, this can lead to increased resource consumption and longer service times, as patients may require more intensive treatment once their health condition worsens.
    \item In service systems such as healthcare, preemptive priority for critical patients may have unintended consequences for other patients. Although this approach helps ensure that critical patients receive timely care, it may leave other patients vulnerable to health deterioration, thereby increasing the likelihood of revisits and placing additional demand on the system.
\end{itemize}

In this paper, we study waiting times of customers in a two-class multi-server system with a finite waiting area. Unlike models that assume absolute service priority for a single customer class, here we consider general non‑preemptive admission priorities as well as dynamically evolving customer type and abandonment. We model this system using a continuous-time Markov chain (CTMC) and apply the Krylov subspace approximation methods (see~\textcite{van1996matrix} and~\textcite{saad2011numerical}) to evaluate the distributions of the customers' waiting times. Our approach offers a scalable computational framework that is well suited for analysing large-scale complex systems.

Next, we construct a level-dependent quasi-birth-and-death (QBD) model for this system and derive expressions for the LSTs of waiting time distributions. Then the numerical inversion algorithms by~\textcite{DenIseger_2006} and~\textcite{horvath2020numerical} can be applied to obtain probability density functions. This approach to evaluate the waiting time distributions has been adopted by many researchers, see~\textcite{Kim2013286} and~\textcite{Fajardo2017255}. Our objective is to compare the distributions of waiting times obtained from the CTMC-Krylov approximation methods with the respective QBD-based distributions. This helps validate the accuracy of the distributions obtained from CTMC–Krylov approach. Further, we derive the expressions for computing the long-run waiting time distributions. Those expressions involve the long-run probabilities of observing specific number of customers in the system, which we evaluate as the stationary distribution of the QBD model. 

We note that the QBD model is not always computationally feasible for large-scale systems. The state space of the QBD grows rapidly with the system size, leading to a substantial increase in the dimensionality of the generator matrix. Consequently, the associated computations become increasingly demanding in terms of both memory and computational time. This limits the applicability of the QBD approach when analysing realistically sized service systems. In contrast, the CTMC–Krylov based approach is computationally efficient and significantly reduces both memory requirements and computational complexity. Moreover, the computational cost of the Krylov-based approach scales favourably with the size of the state space, making it feasible to analyse systems of realistic size.

We illustrate the application potential of our models to multi-server systems by analysing the conditional and long-run waiting times of patients in an example of a hospital system. The model parameters we use in our analysis are based on a large dataset from a tertiary referral hospital in Australia. We analyse the impact of customer admission priorities on their waiting times and compare the results of the QBD and CTMC models. The comparison shows that CTMC model together with Krylov subspace method produces results similar to those obtained from the QBD model, with an error tolerance less that $10^{-7}$. We give useful insights to support decision-making in hospital management.

In summary, the following are the key contributions of this paper.
\begin{itemize}
\item We consider a two-class finite-capacity multi-server system with dynamic customers' priorities, dynamically evolving customer type, and abandonment. 

\item We define the LSTs of the distributions of the conditional and long-run waiting times.

\item We apply a CTMC to model this system, and describe a methodology based on Krylov subspace approach to evaluate the distributions of the conditional waiting times.

\item We also construct suitable level-dependent QBDs useful for the derivation of the waiting times. 

\item We illustrate the application potential of our methodology to multi-server systems through an example of a healthcare system, based on data obtained from a tertiary referral hospital in Australia.

\item We assess the accuracy of the CTMC-Krylov method by comparing with the conditional waiting time distributions obtained from the QBD-based approach.

\item We give valuable insights from the analysis to support the decision-making in hospital management.
\end{itemize}
We note that we refer to entities in our general models as \textit{customers} to keep the exposition general. When discussing the application to a healthcare system, we use the term \textit{patients}.

\subsection{Structure of the paper}

The rest of the paper is structured as follows. In section~\ref{sec:Problem}, we describe the problem and notations. In Section~\ref{sec:CTMCoriginal}, we give the CTMC model and the Krylov subspace method to compute the waiting time distributions. In Section~\ref{sec:QBDmodelIII}, we give the QBD model and the expressions to compute the waiting time distribution. Section~\ref{Numerical_examples} describe the numerical example, model parameters, and the stationary analysis performed via the QBD model. Section~\ref{Conditional_waiting_times_beta_0_r1_1_r2_0} gives the conditional waiting times analysis under both the CTMC and the QBD models. Section~\ref{longrun_waiting_times_beta_r1_r2} gives long-run waiting time analysis. We give a discussion of our analysis in Section~\ref{discussion}, a conclusion in Section~\ref{sec:Conclusion}, and limitations and future directions in Section\ref{limitations}.

\section{Problem description}\label{sec:Problem}

We consider a finite-capacity queueing system with total capacity $N$, comprising $B<N$ parallel servers and a shared waiting area of size $N-B$. The system serves $i$ types of customers, $i=1,2$. We assume that the Type~$i$ customers arrive at Poisson rates $\lambda_i$ and their service times follow exponential distribution with rate $\mu_i$, for $i=1,2$. An arriving customer is promptly assigned to a server if there is one available on their arrival. Otherwise, they enter the waiting area and queue until a server becomes available. The queue in the waiting area is divided into $i=1,2$ queues, each for the respective Type~$i$ customers, which we refer to as Queue~$i$, respectively. When a customer arrives and all servers are occupied, they join their respective queue from the end. When the system is full, the arriving customers are redirected to the nearest facility.

Furthermore, we consider the possibility that waiting customers may abandon the system before receiving service. Specifically, a Type~$i$ customer in the waiting area independently leaves the system after an exponentially distributed patience time, at rate $\chi_i \geq 0$. Figure~\ref{QBDsystem} provides a schematic representation of the system.
\begin{figure}[!htbp]
    \centering    
\includegraphics[scale=0.5]{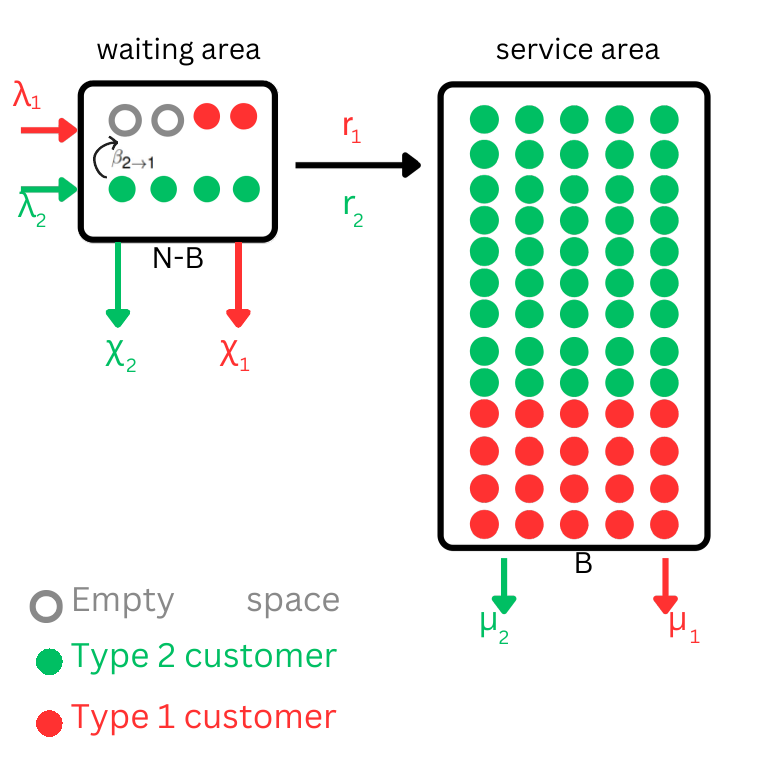}
    \caption{Structure of the system under consideration.}
    \label{QBDsystem}
\end{figure}
We assume that if a server becomes available (due to a customer's departure) at the time when there are some customers in both queues, such server is assigned to a customer in a first-come-first-served (FCFS) discipline, giving priority to Type~$i$ customers with probability $r_i$, for $i=1,2$, where $r_1+r_2=1$. For instance,  if $r_1=1$ (and so $r_2=0$), then Type~$2$ customers are allocated to servers based on their arrival order only when the Queue~$1$ is empty. If one queue is empty while the other contains customers at the moment a server becomes available, the server is assigned to the customer at the front of the non-empty queue. We assume non‑preemptive priority, that is, customers already in service continue uninterrupted until completion, and priority rules apply only when a server becomes available. Further, we assume that customers in the queues may change their type from $2$ to $1$ at a Poisson rate $\beta_{2\to 1}\geq 0$. In the event that a Type~$2$ customer in the waiting area transitions to Type~$1$, they will be positioned at the back of Queue~$1$.

We consider a problem of estimating waiting times of Type~$i$ customers at any position in their respective queue in the above system. Let $T_1(n, b_1, m_1)$ denote the random variable that records the waiting time of a Type~$1$ customer, given that they are the $m_1$-th ($m_1\leq w_1$) Type~$1$ customer in the waiting area and there are $n$ customers in the system with $b_1 \geq 0$ servers occupied by Type~$1$ customers, for some $m_1 > 0$. It follows that all servers are occupied, implying $n > B$ and $b_2 = B - b_1$. Let $f_1(n, b_1, m_1)(t)$ denotes the corresponding (conditional) probability density function for all $t > 0$, and define its Laplace–Stieltjes transform (LST) as
\begin{equation}
\widetilde{F}_1(n, b_1, m_1)(s) = \int_0^{\infty} e^{-st} f_1(n, b_1, m_1)(t) \, dt,
\end{equation}
and let $ \mathbb{E}(T_1(n, b_1, m_1))$ be the corresponding mean.

Further, let $T_2(n, b_1, w_1, m_2)$ denotes the random variable that records the waiting time of a Type~$2$ customer, given that they are the $m_2$-th ($m_2\leq w_2$) Type~$2$ customer in the waiting area, and there are $n$ customers in the system in total with $w_1 \geq 0$ Type~$1$ customers in the waiting area and $b_1 \geq 0$ beds occupied by Type~$1$ customers, for some $m_2 > 0$. It follows that $n > B$ and the total number of Type~$2$ customers in the waiting area is $n - B - w_1 \geq m_2$. Let $f_2(n, b_1, w_1, m_2)(t)$ denotes the corresponding (conditional) probability density function for all $t > 0$, and define its Laplace–Stieltjes transform (LST) as
\begin{equation}
\widetilde{F}_2(n, b_1, w_1, m_2)(s) = \int_0^{\infty} e^{-st} f_2(n, b_1, w_1, m_2)(t) \, dt,
\label{Type2withNoChange}
\end{equation}
and let $ \mathbb{E}(T_2(n, b_1, w_1,m_2))$ be the corresponding mean.

Furthermore, let $T_1$ denotes the long-run waiting time of Type~$1$ customers, and let $f_1(t)$ be the corresponding (unconditional) probability density function for all $t > 0$, defined by
\begin{equation}
f_1(t) = \sum_{(n, b_1, w_1)} \pi_{(n-1, b_1, w_1-1)} \cdot f_1(n, b_1, w_1)(t),
\label{longrun_density_T1}
\end{equation}
where $\pi_{(n-1, b_1, w_1-1)}$ is the long-run proportion of time we observe $n-1$ customers in the system with $b_1$ Type 1 customers at the servers and $w_1-1$ Type 1 customers in the waiting area, interpreted as the proportion of time an arriving Type~$1$ customer is the last in Queue~$1$ and just before they enter they observe $n-1$ total customers in the system, and $b_1$ and $w_1-1$ Type~$1$ customers in the servers and in the waiting area, respectively. It follows that the corresponding mean is
\begin{equation}
\mathbb{E}(T_1) = \sum_{(n, b_1, w_1)} \pi_{(n-1, b_1, w_1-1)} \cdot \mathbb{E}(T_1(n, b_1, w_1)).
\label{E(T1)}
\end{equation}

Moreover, let $T_2$ denotes the long-run waiting time of Type~$2$ customers, and let $f_2(t)$ be the corresponding (unconditional) probability density function for all $t > 0$, defined by
\begin{equation}
f_2(t) = \sum_{(n, b_1, w_1)} \pi_{(n-1, b_1, w_1)} \cdot f_2(n, b_1, w_1, w_2)(t),
\label{longrun_density_T2}
\end{equation}
with $n=B+w_1+w_2$, where $\pi_{(n-1, b_1, w_1)}$ is the long-run proportion of time we observe $n-1$ customers in the system with $b_1$ Type 1 customers at the servers and $w_1$ Type 1 customers in the waiting area, interpreted as the proportion of time an arriving Type~$2$ customer is the last in Queue~$2$ and just before they enter they observe $n-1$ total customers in the system, and $b_1$ and $w_1$ Type~$1$ customers in the servers and in the waiting area, respectively. It follows that the corresponding mean is
\begin{equation}
\mathbb{E}(T_2) = \sum_{(n, b_1, w_1)} \pi_{(n-1, b_1, w_1)} \cdot \mathbb{E}(T_2(n, b_1, w_1, w_2)).
\label{E(T2)}
\end{equation}

\section{CTMC priority model ($\beta_{2\to1}\geq0$, $r_1, r_2\in[0,1]$)}\label{sec:CTMCoriginal}

We construct a CTMC model for the system described in Section~\ref{sec:Problem} with the aim of estimating the waiting times of Type~$i$, $i=1,2$, customers. We tag one such customer whose waiting time distribution is of interest. Let $R(t)$ denote the \textit{rank} of the tagged Type~$i$, $i=1,2$, customer at time $t$, where the \textit{rank} is defined as the total number of Type~$1$ and Type~$2$ (if any) customers ahead of the tagged customer in the waiting area, plus the tagged customer itself. As an example, if there are $w_1$ Type~$1$ customers in the waiting area and the tagged Type~$1$ customer of interest is the $m_1$-th customer in Queue~$1$ at time $t$, then $R(t) = m_1 \leq w_1$. Similarly, if there are $w_1$ Type~$1$ customers and $w_2$ Type~$2$ customers in the waiting area and the tagged Type~$2$ customer of interest is the $m_2$-th customer in Queue~$2$  at time $t$, then $R(t) = w_1+m_2$, with $m_2 \leq w_2$. Note that if $R(t) = r \leq w_1$, the rank $r$ corresponds to a Type‑$1$ customer, whereas if $R(t) = r > w_1$, the rank refers to a Type~$2$ customer in the waiting area.

Further, let $B_1(t)$ and $W_1(t)$ be the random variables that record the total number of Type~$1$ incustomers (customers in the servers) and Type~$1$ customers in the waiting area at time $t$, respectively. We denote by $B_2(t)$ and $W_2(t)$ the number of Type~$2$ incustomers and Type~$2$ customers in the waiting area, respectively, at time $t$. Then, we have $X(t)=B_1(t)+B_2(t)+W_1(t)+W_2(t)$, with $B_1(t)+B_2(t)\leq B$ and $W_1(t)+W_2(t)=\max\{0,X(t)-B\}$, where $X(t)$ denotes the total number of customers in the system. Note that given $(X(t),B_1(t),W_1(t))=(n,b_1,w_1)$, it follows that the number of Type~$2$ customers ($w_2$) in the waiting area, and the number of Type~$2$ incustomers ($b_2$) currently in the system are given by
\begin{align}
w_2 &= \max\{0, n - B - w_1\}, \label{w2} \\
b_2 &= \min\{n, B\} - b_1. \label{b2}
\end{align}

We consider a continuous-time Markov process $\{(R(t), B_1(t), W_1(t), W_2(t)): t \geq 0\}$ with a state space 
\begin{eqnarray*}\label{eq:statespace_Rt_Type2}
\mathcal{S}&=&
\{(r,b_1,w_1,w_2):r=0,1,\ldots,N-B;b_1=0,1,\ldots,B;
\\
&&
w_1,w_2=0,1,\ldots,N-B; r\leq w_1+w_2 \leq N-B\},
\end{eqnarray*}
where, $R(t) = r=0$ represents an absorbing state, observed when the tagged customer is allocated to a server. The generator matrix of the process is 
\begin{eqnarray*}
{\bf Q} = 
\begin{bmatrix}
{\bf T} & {\bf T_0} \\
{\bf O} & {\bf O}
\end{bmatrix},
\end{eqnarray*}
where ${\bf T}$ denotes the transient rates submatrix and ${\bf T_0}$ is the exit rate vector. The transition rates of the process are given in Table~\ref{tab:QBDgenerator_Rt} for $r_1=1$, $r_2=0$, whereas Table~\ref{tab:QBDgenerator_Rt_CTMCs} presents the transition rates for $r_1,r_2\in[0,1]$.

\begin{table}[!htbp]
		\centering
        
		\setlength{\tabcolsep}{1pt}	
		\begin{tabular} {l@{\hspace{1cm}} l@{\hspace{1cm}}l@{\hspace{1cm}} l}

  \hline
   	$\text{Conditions}$&	\text{Current state} &	$\text{New state}$ & $\text{Rate}$\\

\hline
$w_1+w_2<N-B$ &  &	 &	\\

\quad $r>w_1$ & $(r,b_1,w_1,w_2)$  &	$(r+1,b_1,w_1+1,w_2)$  &	 $\lambda_1$\\

\quad $1\leq r\leq w_1$ & $(r,b_1,w_1,w_2)$  &	$(r,b_1,w_1+1,w_2)$  &	 $\lambda_1$\\

\quad $1\leq r<N-B$ & $(r,b_1,w_1,w_2)$  &	$(r,b_1,w_1,w_2+1)$  &	 $\lambda_2$\\

\hline
$w_2>0$ & & & \\

\quad $1\leq r\leq w_1$  & $(r,b_1,w_1,w_2)$  &	$(r,b_1,w_1+1,w_2-1)$  &	 $w_2\beta_{2\to1}$\\

\quad $r >w_1$ & $(r,b_1,w_1,w_2)$  &	$(r+1,b_1,w_1+1,w_2-1)$  &	 $(w_1+w_2-r) \beta_{2\to1}$\\

\quad $r >w_1$ & $(r,b_1,w_1,w_2)$  &	$(r,b_1,w_1+1,w_2-1)$  &	 $(r-w_1-1) \beta_{2\to1}$\\

\quad $r>w_1\geq0$ & $(r,b_1,w_1,w_2)$  &	$(w_1+1,b_1,w_1+1,w_2-1)$  &	 $\beta_{2\to1}$\\

\hline

$b_1>0$ & & & \\

\quad $r>w_1\geq1$ & $(r,b_1,w_1,w_2)$  &	$(r-1,b_1,w_1-1,w_2)$  &	 $b_1\mu_{1}$\\

\quad  $r>w_1=0$ & $(r,b_1,w_1,w_2)$  &	$(r-1,b_1-1,w_1,w_2-1)$  &	 $b_1\mu_{1}$\\

\quad  $1\leq r \leq w_1$ & $(r,b_1,w_1,w_2)$  &	$(r-1,b_1,w_1-1,w_2)$  &	 $b_1\mu_{1}$\\

\hline
$b_2>0$ & & & \\

\quad $r>w_1\geq1$ & $(r,b_1,w_1,w_2)$  &	$(r-1,b_1+1,w_1-1,w_2)$  &	 $b_2\mu_{2}$\\

\quad $r>w_1=0$ & $(r,b_1,w_1,w_2)$  &	$(r-1,b_1,w_1,w_2-1)$  &	 $b_2\mu_{2}$\\

\quad  $1\leq r \leq w_1$ & $(r,b_1,w_1,w_2)$  &	$(r-1,b_1+1,w_1-1,w_2)$  &	 $b_2\mu_{2}$\\

\hline

$1< r\leq w_1$  & $(r,b_1,w_1,w_2)$  &	$(r-1,b_1,w_1-1,w_2)$  & $(r-1)\chi_{1}$  \\

$1\leq r < w_1$  & $(r,b_1,w_1,w_2)$  &	$(r,b_1,w_1-1,w_2)$  & $(w_1-r)\chi_{1}$  \\

$r > w_1 \geq 1$ & $(r,b_1,w_1,w_2)$  &	$(r-1,b_1,w_1-1,w_2)$  & $w_1\chi_{1}$  \\

$w_1< r < w_1+w_2$  & $(r,b_1,w_1,w_2)$  &	$(r,b_1,w_1,w_2-1)$  & $(w_1+w_2-r)\chi_{2}$  \\

$w_1+1< r \leq w_1+w_2$  &$(r,b_1,w_1,w_2)$  &	$(r-1,b_1,w_1,w_2-1)$  & $(r-w_1-1)\chi_{2}$  \\

\hline
  \end{tabular}		
  \caption{Transition rates of the CTMC process under the assumptions: Abandonment with rates $\chi_i$, $i=1,2$, and change from Type~$2$ to Type~$1$ may be observed in the waiting area with $\beta_{2\to1}\geq 0$ ($\beta_{2\to 1}\neq0$ is possible), and $r_1=1$, $r_2=0$.}
\label{tab:QBDgenerator_Rt}
\end{table}

\begin{table}[!htbp]
		\centering
        
		\setlength{\tabcolsep}{1pt}	
		\begin{tabular} {l@{\hspace{1cm}} l@{\hspace{1cm}}l@{\hspace{1cm}} l}

  \hline
   	$\text{Conditions}$&	\text{Current state} &	$\text{New state}$ & $\text{Rate}$\\

\hline
$ w_1+w_2<N-B$ & & & \\

\quad $r>w_1$ & $(r,b_1,w_1,w_2)$  &	$(r+1,b_1,w_1+1,w_2)$  &	 $\lambda_1$\\

\quad  $1\leq r\leq w_1$ & $(r,b_1,w_1,w_2)$  &	$(r,b_1,w_1+1,w_2)$  &	 $\lambda_1$\\

\quad $1\leq r$  & $(r,b_1,w_1,w_2)$  &	$(r,b_1,w_1,w_2+1)$  &	 $\lambda_2$\\
\hline

$w_2>0$ & & & \\

\quad $1\leq r\leq w_1$ & $(r,b_1,w_1,w_2)$  &	$(r,b_1,w_1+1,w_2-1)$  &	 $w_2\beta_{2\to1}$\\

\quad $r >w_1$ & $(r,b_1,w_1,w_2)$  &	$(r+1,b_1,w_1+1,w_2-1)$  &	 $(w_1+w_2-r) \beta_{2\to1}$\\

\quad $r >w_1$ & $(r,b_1,w_1,w_2)$  &	$(r,b_1,w_1+1,w_2-1)$  &	 $(r-w_1-1) \beta_{2\to1}$\\

\quad $r>w_1$ & $(r,b_1,w_1,w_2)$  &	$(w_1+1,b_1,w_1+1,w_2-1)$  &	 $\beta_{2\to1}$\\

\hline

$b_1>0$, $w_1\geq1$ & & & \\

\quad $r>w_1$ & $(r,b_1,w_1,w_2)$  &	$(r-1,b_1,w_1-1,w_2)$  &	 $r_1 b_1\mu_{1}$\\

\quad $r>w_1+1$, $w_2>1$ & $(r,b_1,w_1,w_2)$  &	$(r-1,b_1-1,w_1,w_2-1)$  &	 $r_2 b_1\mu_{1}$\\

\quad $r=w_1+1$, $w_2>0$& $(r,b_1,w_1,w_2)$  &	$(0,b_1-1,w_1,w_2-1)$  &	 $r_2 b_1\mu_{1}$\\

\hline
$b_1>0$, $r>w_1=0$ & $(r,b_1,w_1,w_2)$  &	$(r-1,b_1-1,w_1,w_2-1)$  &	 $b_1\mu_{1}$\\

\hline
$b_1>0$, $1\leq r \leq w_1$ & & & \\

\quad $w_2=0$ & $(r,b_1,w_1,w_2)$  &	$(r-1,b_1,w_1-1,w_2)$  &	 $b_1\mu_{1}$\\

\quad  $w_2>0$ & $(r,b_1,w_1,w_2)$  &	$(r-1,b_1,w_1-1,w_2)$  &	 $r_1 b_1\mu_{1}$\\

\quad  $w_2>0$ & $(r,b_1,w_1,w_2)$  &	$(r,b_1-1,w_1,w_2-1)$  &	 $r_2 b_1\mu_{1}$\\

\hline
$b_2>0$, $r>w_1\geq1$ & & & \\

\quad $r>w_1$ & $(r,b_1,w_1,w_2)$  &	$(r-1,b_1+1,w_1-1,w_2)$  &	 $r_1 b_2\mu_{2}$\\

\quad $r>w_1+1$, $w_2>1$& $(r,b_1,w_1,w_2)$   &	$(r-1,b_1,w_1,w_2-1)$  &	 $r_2 b_2\mu_{2}$\\

\quad $r=w_1+1$, $w_2>0$ & $(r,b_1,w_1,w_2)$   &	$(0,b_1,w_1,w_2-1)$  &	 $ r_2 b_2\mu_{2}$\\

\hline
$b_2>0$, $r>w_1=0$ & $(r,b_1,w_1,w_2)$  &	$(r-1,b_1,w_1,w_2-1)$  &	 $b_2\mu_{2}$\\
\hline

$b_2>0$, $1\leq r \leq w_1$ & & & \\

\quad $w_2=0$ & $(r,b_1,w_1,w_2)$  &	$(r-1,b_1+1,w_1-1,w_2)$  &	 $b_2\mu_{2}$\\

\quad $w_2>0$ & $(r,b_1,w_1,w_2)$  &	$(r-1,b_1+1,w_1-1,w_2)$  &	 $r_1 b_2\mu_{2}$\\

\quad $w_2>0$ & $(r,b_1,w_1,w_2)$  &	$(r,b_1,w_1,w_2-1)$  &	 $r_2 b_2\mu_{2}$\\

\hline

$1< r\leq w_1$  & $(r,b_1,w_1,w_2)$  &	$(r-1,b_1,w_1-1,w_2)$  & $(r-1)\chi_{1}$  \\

$1\leq r < w_1$  & $(r,b_1,w_1,w_2)$  &	$(r,b_1,w_1-1,w_2)$  & $(w_1-r)\chi_{1}$  \\

$r > w_1 \geq 1$  &$(r,b_1,w_1,w_2)$  &	$(r-1,b_1,w_1-1,w_2)$  & $w_1\chi_{1}$  \\

$w_1< r < w_1+w_2$  & $(r,b_1,w_1,w_2)$  &	$(r,b_1,w_1,w_2-1)$  & $(w_1+w_2-r)\chi_{2}$  \\

$w_1+1< r \leq w_1+w_2$  & $(r,b_1,w_1,w_2)$  &	$(r-1,b_1,w_1,w_2-1)$  & $(r-w_1-1)\chi_{2}$  \\

\hline

\end{tabular}
		\caption{Transition rates of the CTMC process under the assumptions: Abandonment with rates $\chi_i$, $i=1,2$, and change from Type~$2$ to Type~$1$ may be observed in the waiting area with $\beta_{2\to 1}\geq 0$ ($\beta_{2\to 1}\neq0$ is possible), $r_1,r_2\in[0,1]$.}
\label{tab:QBDgenerator_Rt_CTMCs}
\end{table}

\subsection{Computational and memory complexity of the CTMC model} \label{comemcomplexity}

The CTMC model developed to evaluate the waiting time of a tagged Type~$i$ customer is structurally rich. However, the state space of the model is four-dimensional, and so the number of possible transient states grows rapidly with the total capacity $N$ and the number of beds $B$. The total number of states can be computed using the formula
\begin{equation}
|\mathcal{S}| = \sum_{r=0}^{N-B} \sum_{b_1=0}^{B} \sum_{w_1=0}^{N-B} \sum_{w_2=0}^{N-B} \mathbf{1}_{\{r \leq w_1 + w_2 \leq N - B\}},
\label{ctmc_system_size}
\end{equation}
where $\mathbf{1}$ represents an indicator function that gives a value $1$ (or $0$) when the condition inside is true (or false), respectively. Note that if we skip $r=0$ in the above formula, it gives the total number of possible transient states. We denote this by $|\mathcal{S}^*|$. The infinitesimal generator matrix ${\bf T}$ is then of size $\mathcal{S}^* \times \mathcal{S}^*$, that is, ${\bf T} \in \mathbb{R}^{|\mathcal{S}^*| \times |\mathcal{S}^*|}$. Although ${\bf T}$ is mostly sparse, storing and manipulating it in dense format requires memory of order $\mathcal{O}(|\mathcal{S}^{*}|^{2})$. Specifically, storing ${\bf T}$ in double-precision floating-point format requires $8 \times |\mathcal{S}^{*}|^{2}$ bytes. Additionally, storing the initial distribution vector $\boldsymbol{\alpha}$ and the column vector ${\bf 1}$ requires $8 \times |\mathcal{S}^{*}|$ bytes each.

For instance, when $N = 100$ and $B = 80$, the total number of states are $|\mathcal{S}| = 268,191$, of which $|\mathcal{S}^*| = 249,480$ are transient. This implies that storing ${\bf T}$ in dense format would require $8 \times (249,480)^2 \approx 463.73$ GB of memory. This clearly demonstrates that dense storage and direct computation are infeasible for realistically sized systems, motivating the need for memory-efficient computational techniques, described next.

\subsection{Sparse matrix representation}
\label{sparsematrixrep}

As a first step towards memory-efficient computation, we store and utilise the transition rate matrix ${\bf T}$ and the vector ${\bf t}=-{\bf T}{\bf 1}$ in  compressed sparse row (CSR) format. CSR significantly reduces memory usage by storing only the non-zero entries and their positions (see~\textcite{saad2011numerical}).

Let $n$ denotes the number of rows (or columns, since ${\bf T}$ is square), and $z$ denote the number of non-zero entries in ${\bf T}$. Each row in ${\bf T}$ corresponds to a transient state of the CTMC and contains only a small number of non-zero entries, denoted by $k$, where $1\leq k\leq 8$ in our model. For example, the state $(2,1,0,3)$ exhibits the maximum number of transitions in ${\bf T}$. From this state, the CTMC can transition to $(3,1,1,3)$, $(2,1,0,4)$, $(3,1,1,2)$, $(2,1,1,2)$, $(1,1,1,2)$, $(1,0,0,2)$, and $(1,1,0,2)$. Consequently, the row of ${\bf T}$ corresponding to $(2,1,0,3)$ contains eight non-zero entries: seven off-diagonal elements representing the transition rates to these states, and one diagonal entry that equals the negative of the total rate of leaving $(2,1,0,3)$. On the other hand, the state $(1,0,20,0)$ exhibits the minimum number of transitions in ${\bf T}$. It can only transition to $(0,1,19,0)$, which is an absorbing state and thus not included in ${\bf T}$. Therefore, within ${\bf T}$, the only non-zero entry in a row corresponding to $(1,0,20,0)$ is its diagonal element.

Therefore, $z=nk$. Assuming double-precision floating-point storage (8 bytes per value) and 32-bit integers (4 bytes per index), the total memory required is given by
\begin{eqnarray*}
    \text{Memory} &=& 8z \ (\text{values}) + 4z \ (\text{column indices}) + 4(n + 1) \ (\text{row pointers})\\
    &=& 12z + 4(n + 1) \quad \text{bytes}
\end{eqnarray*}
For example, take $N = 100$ capacity and $B = 80$ servers in our model. The number of transient states is $n = 249,480$. Assuming $k=6$, the above formula gives the total memory required
\begin{eqnarray*}
    \text{Memory} = 18,960,480 \text{ bytes} \approx 18.08 \text{ MB}.
\end{eqnarray*}
This shows that the sparse matrix representation techniques, such as those available in matlab (sparse) or Python (scipy.sparse), are significantly more efficient than storing the matrix in dense format, especially when the matrix is very large. Although storing ${\bf T}$ in a sparse matrix format can reduce memory usage for storage, it does not significantly lower the overall computational cost. This is because computing the phase-type distribution of the waiting time involves evaluating expressions such as $f(t)=-\boldsymbol{\alpha} e^{{\bf T} t} {\bf T} {\bf 1} = \boldsymbol{\alpha} e^{{\bf T}t} {\bf t}$, and $\mathbb{E}[T] = -\boldsymbol{\alpha} {\bf T}^{-1} {\bf 1}$. 
Both of these require either matrix exponentiation or inversion, which are computationally expensive operations with complexity $\mathcal{O}(n^3)$ for dense matrices, and still costly for sparse matrices due to fill-in during factorisation.

\subsection{Krylov subspace method for efficient computation}
\label{krylovsubspacemethod}

To avoid costly computations, we use the Krylov subspace projection methods (the well-known Arnoldi and Lanczos process (see~\textcite{van1996matrix} and~\textcite{saad2011numerical})) to compute the distribution of the time until the tagged customer gets to a server. Krylov methods are particularly well-suited for large sparse transition rate matrices arising from CTMCs. They avoid the need to compute or store the full matrix exponential $e^{{\bf T} t}$. Instead, they approximate $e^{{\bf T} t} {\bf t}$ by projecting it onto a low-dimensional Krylov subspace $\mathcal{K}_m({\bf T}, {\bf t}) = \text{span}\{{\bf t}, {\bf T}{\bf t}, {\bf T}^2{\bf t}, \ldots, {\bf T}^{m-1}{\bf t}\}$ generated by ${\bf T}$ and ${\bf t}$.

The Arnoldi process constructs an orthonormal basis ${\bf V}_m$ for this subspace  and reduces ${\bf T}$ to a small dense Hessenberg matrix ${\bf H}_m = {\bf V}_m^T {\bf T} {\bf V}_m$. The dimension $m$ is typically much smaller than $|\mathcal{S}|$, often in the range of 20–50. The exponential action is then approximated as 
\begin{eqnarray}
    e^{{\bf T} t} {\bf t} \approx {\bf V}_m e^{{\bf H}_m t} {\bf V}_m^T {\bf t},
\end{eqnarray}
which can be computed efficiently and with controlled accuracy. This approach reduces memory usage from $\mathcal{O}(n^2)$ to $\mathcal{O}(nm)$, and computational complexity from $\mathcal{O}(n^3)$ to approximately $\mathcal{O}(nm^2)$, assuming sparse matrix-vector multiplications. Furthermore, Krylov methods allow for adaptive control of approximation error, enabling high accuracy with minimal computational cost. By leveraging the Arnoldi iteration algorithm, we can efficiently compute the phase-type distributions, even when ${\bf T}$ is very large, making the methodology scalable and practical.

To compute $e^{{\bf T} t} {\bf t}$ numerically, we use the function \texttt{expv} from the \texttt{EXPOKIT} package, given by~\textcite{Sidje1998130}. The function \texttt{expv} implements an Arnoldi-based Krylov subspace projection with adaptive error control. In all computations, we employ the default \texttt{EXPOKIT} settings: a maximum Krylov subspace dimension of $m=30$ and a tolerance of $1\times 10^{-7}$. This tolerance represents the relative error threshold that \texttt{EXPOKIT} uses to determine when the Krylov approximation of the matrix exponential action has reached sufficient accuracy. During the Arnoldi iteration, \texttt{EXPOKIT} automatically monitors the local truncation error and, if needed, increases the subspace dimension or subdivides the time interval to ensure that the estimated relative error remains below the specified tolerance. Consequently, the computations achieve an accuracy of approximately seven correct digits, without requiring any additional error estimation procedures.

Although Krylov methods are significantly more efficient than direct matrix exponentiation, it is important to highlight their algorithmic complexity. The Arnoldi iteration used in \texttt{EXPOKIT} has a computational cost of approximately $\mathcal{O}(nm^2)$, where $n$ is the number of transient states and $m$ is the Krylov subspace dimension. Since $m$ is typically small (e.g., $m=30$), the method scales linearly with $n$ and quadratically with $m$, making it suitable for large-scale CTMCs. The adaptive time-stepping and error control mechanisms in \texttt{EXPOKIT} further ensure that the computational cost remains manageable without sacrificing accuracy.

\subsection{Conditional waiting times ($\beta_{2\to1}\geq0$, $r_1, r_2 \in [0,1]$)}\label{sec:ConditionalCMTCoriginal}

We recall that the random variables $T_1(n, b_1, w_1)$ and $T_2(n, b_1, w_1, w_2)$ record the waiting times of a $w_1$-th Type~$1$ customer and $w_2$-th Type~$2$ customer in the waiting area, respectively. Then, both $T_1(n, b_1, w_1)$ and $T_2(n, b_1, w_1, w_2)$ follow phase-type distribution with representation $\mbox{PH}(\boldsymbol{\alpha},{\bf T})$, where $\boldsymbol{\alpha}=[\alpha(r,b_1,w_1,w_2)]$ denotes the initial distribution vector such that $\alpha(R(0), B_1(0), W_1(0), W_2(0))=1$, and $R(0)=w_1$ for Type~$1$ and $R(0)=w_1+w_2$ for Type~$2$ customer, respectively. We immediately have the following result.
\begin{lemma}
The conditional density of the waiting time of  Type~$1$ customer from the time they arrive, is given by
\begin{eqnarray}
f_1(n,b_1,w_1)(t) = -\boldsymbol{\alpha} e^{{\bf T}t} {\bf T}{\bf 1}
\label{conditionaldensityPH_T1}
\end{eqnarray}
with $(R(0), B_1(0), W_1(0), W_2(0))=(w_1,b_1,w_1,n-B-w_1)$ and $\alpha(w_1,b_1,w_1,n-B-w_1)=1$. The corresponding  mean waiting time of the $w_1$-th Type~$1$ customer is given by 
\begin{align}
\mathbb{E}[T_1(n, b_1, w_1)] 
&= -\boldsymbol{\alpha} {\bf T}^{-1} {\bf 1}.
\label{CTMCmeanformulaT1}
\end{align}

Similarly, the conditional density of the waiting time of Type~$2$ customer from the time they arrive, is given by
\begin{eqnarray}
f_2(n,b_1,w_1,w_2)(t) = -\boldsymbol{\alpha} e^{{\bf T}t} {\bf T}{\bf 1}
\label{conditionaldensityPH_T2}
\end{eqnarray}
with $n=B+w_1+w_2$ and $(R(0), B_1(0), W_1(0), W_2(0))=(w_1+w_2,b_1,w_1,w_2)$ and $\alpha(w_1+w_2,b_1,w_1,w_2)=1$, since $n-B-w_1=w_2$. The corresponding  mean waiting time of the $w_2$-th Type~$2$ customer is given by 
\begin{align}
\mathbb{E}[T_2(n, b_1, w_1, w_2)] 
&= -\boldsymbol{\alpha} {\bf T}^{-1} {\bf 1}.
\label{CTMCmeanformulaT2}
\end{align}

Here, $t$ (non-boldface) denotes the (scalar) time variable in the density function $f(t)$, and ${\bf t}$ (boldface) refers to the vector ${\bf t}=-{\bf T}{\bf 1}$.

\end{lemma}

\subsection{Algorithm for the CTMC-based computation of the waiting times}

Here, we outline the key steps to compute the conditional waiting time distribution of a tagged customer using the CTMC approach and present a structured pseudocode in Algorithm~\ref{alg:CTMC_computation}.

\begin{enumerate}[{\bf Step~1.}]
    \item Construct the transition rate matrix ${\bf Q}$ for the CTMC model in Table~\ref{tab:QBDgenerator_Rt_CTMCs} and store it in sparse format to ensure memory efficiency.

\item Extract the transient matrix $\mathbf{T}$ from ${\bf Q}$, ensuring that the sparse structure is preserved.

\item Determine the initial state $(r,b_1,w_1,w_2)$ of the process from the time the tagged customer enters the system. Then construct the initial distribution vector $\boldsymbol{\alpha}$ with $\boldsymbol{\alpha}(r,b_1,w_1,w_2)=1$ and store it in sparse format. 

\item Compute the mean waiting time of the tagged customer using~\eqref{CTMCmeanformulaT1} or~\eqref{CTMCmeanformulaT2}, ensuring that all computations preserve the sparsity of matrices. 

\item Evaluate the probability density function (pdf) of the customer's waiting time using~\eqref{conditionaldensityPH_T1} or~\eqref{conditionaldensityPH_T2}, implemented via the function \texttt{expv} from the matlab package \texttt{EXPOKIT}, given by~\textcite{Sidje1998130}.
\end{enumerate}

\begin{algorithm}[H]
\caption{Pseudocode for CTMC computation (conditional waiting times)}
\label{alg:CTMC_computation}

\begin{algorithmic}[1]

\Input Model parameters $N$, $B$, $\lambda_1$, $\lambda_2$, $\mu_1$, $\mu_2$, $\beta_{2\to1}$, $r_1$, $r_2$, initial state $(r,b_1,w_1,w_2)$

\State Initialise $\texttt{states}\gets\emptyset$

\For{$r=0$ to $N-B$} ($r=0$ is absorbing)
    \For{$b_1=0$ to $B$}
        \For{$w_1=0$ to $N-B$}
            \For{$w_2=0$ to $N-B$}
                \If{$r \leq w_1+w_2$ \textbf{and} $w_1+w_2\leq N-B$}
                    \State Append $(r,b_1,w_1,w_2)$ to $\texttt{states}$
                \EndIf
            \EndFor
        \EndFor
    \EndFor
\EndFor

\State $\texttt{numStates}\gets \texttt{length(states)}$

\State Define $\texttt{key}(r,b_1,w_1,w_2) = \texttt{``r\_b1\_w1\_w2''}$
for fast state lookup.

\State Initialise $\texttt{stateMap}$ as an empty map.

\For{$i=1$ to $\texttt{numStates}$}
    \State $\texttt{stateMap[key(states(i))]} \gets i$
\EndFor

\Function{safeLookup}{$r,b_1,w_1,w_2$}
    \State $k\gets \texttt{key}(r,b_1,w_1,w_2)$
    \If{$k$ exists in $\texttt{stateMap}$}
        \State \Return $\texttt{stateMap}[k]$
    \Else
        \State \Return $0$
    \EndIf
\EndFunction

\State Initialise
$\texttt{row\_idx}\gets [\,]$,
$\texttt{col\_idx}\gets [\,]$,
$\texttt{rate\_val}\gets [\,]$

\State Identify all possible transitions.

\State Use \textsc{safeLookup} to find destination state indices.

\State Append transitions to
    $\texttt{row\_idx}$,
    $\texttt{col\_idx}$,
    $\texttt{rate\_val}$.

\State Set each diagonal entry equal to the negative row sum.

\State Construct sparse generator matrix $\texttt{Q}
=
\texttt{sparse}
(
\texttt{row\_idx},
\texttt{col\_idx},
\texttt{rate\_val},
\texttt{numStates},
\texttt{numStates}
)$.

\State Extract transient matrix $\mathbf{T}$ as $\texttt{T\_transient}$, from $\mathbf{Q}$.

\State Let $n\gets \texttt{length}(\texttt{T\_transient})$

\State Initialise sparse row vector
$\boldsymbol{\alpha}$ of size $n$

\State $\texttt{full\_idx} \gets \textsc{safeLookup}(r,b_1,w_1,w_2)$ (Get the index of the initial state in states)

\State $\texttt{transient\_pos} \gets$ position of $(r,b_1,w_1,w_2)$ in
$\texttt{T\_transient}$

\State $\boldsymbol{\alpha}(\texttt{transient\_pos})\gets 1$ (set the initial probability)

\State Compute exit rate vector, $\texttt{exit\_vector} = -\texttt{T\_transient}\times\mathbf{1}$

\State Compute mean waiting time, $\mathbb{E}[T] = \boldsymbol{\alpha} \left( -\texttt{T\_transient} \backslash \mathbf{1} \right)$

\State Determine a vector
$\texttt{t\_values}$ of suitable time points.

\For{$i=1$ to $\texttt{length(t\_values)}$}
    \State
    $(v_t,\texttt{err})
    \gets
    \texttt{expv}
    (\texttt{t\_values}(i),
    \texttt{T\_transient},
    \texttt{exit\_vector})$
    \State
    $\texttt{pdf}(i)
    \gets
    \boldsymbol{\alpha} v_t$
\EndFor

\State \Return Conditional pdf $f(t)$ and expected waiting time $\mathbb{E}[T]$.

\end{algorithmic}
\end{algorithm}

Next, we apply a quasi-birth-and-death (QBD) process to model the system defined in Section~\ref{sec:Problem}, motivated by the following two reasons. First, we are interested in comparing the conditional waiting time distributions obtained from the CTMC-Krylov approximation methods with the respective distributions obtained from QBD based matrix-analytic methods in the scenarios when $r_1=1$ and $r_2=0$. This approach helps validate the accuracy of the conditional waiting times obtained from CTMC-Krylov approach. Second, we aim to compute the probability densities of the long-run waiting times of Type~$1$ and Type~$2$ customers, as defined in \eqref{longrun_density_T1} and \eqref{longrun_density_T2}, respectively, and the corresponding mean waiting times given in \eqref{E(T1)} and \eqref{E(T2)}. These computations require the long-run probabilities $\pi_{(n,b_1,w_1)}$ for observing $n$ number of total customers with $b_1$ and $w_1$ Type~$1$ customers in the servers and the waiting area, respectively.

We construct a QBD model and use matrix-analytic methods to evaluate the stationary distribution $\pi_{(n,b_1,w_1)}$ and the conditional waiting time distributions. We note that unlike the CTMC model, the QBD model is not always computationally feasible for analysing the waiting time distributions in large-scale systems, as discussed in Section~\ref{sec:introduction}.

\section{QBD model}\label{sec:QBDmodelIII}

We model the system using a continuous-time level-dependent QBD $\{(X(t),\varphi(t)):t\geq 0\}$, where the level variable $X(t)$ records the total number of customers in the system and the phase variable $\varphi(t)=(B_1(t), W_1(t))$ records the total number of Type~$1$ incustomers (customers in the servers) and Type~$1$ customers in the waiting area at time $t$, respectively. The state space of such QBD is
\begin{eqnarray}\label{eq:statespace}
\mathcal{S}&=&
\{(n,b_1,w_1):n=0,\ldots,N;b_1=0,\ldots,\min\{n,B\};w_1=0,\ldots,\max\{0,n-B\}\},
\end{eqnarray}
where $n$ is the total number of customers in the system, $b_1$ is the number of Type~$1$ incustomers, $w_1$ is the number of Type~$1$ customers in the waiting area. We note that the number of Type~$2$ customers ($w_2$) in the waiting area and the number of Type~$2$ incustomers ($b_2$) currently in the system can be obtained using~\eqref{w2} and~\eqref{b2}.

Let $\alpha_{(n,b_1,w_1)}=\mathbb{P}(X(0)=n,\varphi(0)=(b_1,w_1))$ be the probability that the process is in the state $(n,b_1,w_1)$ at the time the process starts. The initial distribution vector of the QBD process is then represented by
\begin{eqnarray}
\boldsymbol{\alpha}=[\boldsymbol{\alpha}_n]_{n=0,1,2,\ldots,N},  \boldsymbol{\alpha}_n=[\boldsymbol{\alpha}_{(n,b_1,w_1)}]_{b_1=0,\ldots,\min\{n,B\};w_1=0,\ldots,\max\{0,n-B\}},
\end{eqnarray}
and the transition rate from state $(n,b_1,w_1)$ to $(n',b_1',w_1')$ is
\begin{eqnarray}
    q_{(n,b_1,w_1)\to(n',b_1',w_1')}=\left.\frac{d}{dt}
\mathbb{P}\big(X(t)=n',\varphi(t)=(b_1',w_1') \ | \ X(0)=n,\varphi(0)=(b_1,w_1)\big)\right\vert_{t=0}.
\end{eqnarray}

These transition rates are recorded in block matrices ${\bf Q} = {\bf Q}^{[n,n^{'}]}=\left[q_{(n,b_1,w_1)\to(n',b_1',w_1')}\right]$, where $b_1=0,\ldots,\min\{n,B\}$ and $w_1=0,\ldots,\max\{0,n-B\}$, that constitute a tri-diagonal generator matrix

\begin{eqnarray}
	{\bf Q}
	=
	\begin{bmatrix}
		{\bf Q}^{[0,0]} & {\bf Q}^{[0,1]} & {\bf 0} & \cdots & \cdots & \cdots & \cdots & {\bf 0}\\
		{\bf Q}^{[1,0]} & {\bf Q}^{[1,1]} & {\bf Q}^{[1,2]} & {\bf 0} & \cdots & \cdots & \cdots & {\bf 0}\\
		{\bf 0} & {\bf Q}^{[2,1]} & {\bf Q}^{[2,2]} & {\bf Q}^{[2,3]} & \cdots &  \cdots & \cdots & {\bf 0}\\
		\vdots & \vdots & \vdots & \vdots & \cdots & \cdots & \cdots & \vdots\\
		{\bf 0} & {\bf 0} & {\bf 0} & {\bf 0} & \cdots & {\bf Q}^{[N-1,N-2]} & {\bf Q}^{[N-1,N-1]} & {\bf Q}^{[N-1,N]}\\
		{\bf 0} & {\bf 0} & {\bf 0} & {\bf 0} & \cdots & {\bf 0} & {\bf Q}^{[N,N-1]} & {\bf Q}^{[N,N]}
	\end{bmatrix}.
\end{eqnarray}
The off-diagonal transition rates $q_{(n,b_1,w_1)\to(n',b_1',w_1')}$ of this QBD are given in Table~\ref{tab:parametersQBDI-2ISYM}. The on-diagonal transition rates can be calculated using the standard formula
\begin{eqnarray}\label{eq:on-diagonals}
q_{(n,b_1,w_1)\to(n,b_1,w_1)} = -\sum_{(n^{'},b_1^{'},w_1^{'})\not= (n,b_1,w_1)  } q_{(n,b_1,w_1)\to(n^{'},b_1^{'},w_1^{'})}.
\end{eqnarray}

\begin{table}[!htbp]
		\centering
        
		\setlength{\tabcolsep}{1pt}	
		\begin{tabular} {l@{\hspace{0.5cm}} l@{\hspace{0.5cm}}l l}

  \hline
   	$\text{Conditions}$&	\text{Current state} &	$\text{New state}$ &	$\text{Rate}$\\

\hline
 	$0 \leq n < B$ & \text{$(n,b_1,0)$ \quad } &	\text{$(n+1,b_1+1,0)$ \quad } &	$\lambda_1$\\

$B\leq n<N$ & \text{$(n,b_1,w_1)$ \quad } &	\text{$(n+1,b_1,w_1+1)$ \quad } &	$\lambda_1$ \\

\hline

$n<N$  &\text{$(n,b_1,w_1)$ \quad } &	\text{$(n+1,b_1,w_1)$ \quad } & $\lambda_2$\\

\hline

$w_1 \geq 1$  &  \text{$(n,b_1,w_1)$ \quad }  & 	\text{$(n-1,b_1,w_1-1)$ \quad }  & $w_1\chi_1$ \\

$w_2 \geq 1$   &  \text{$(n,b_1,w_1)$ \quad }  & 	\text{$(n-1,b_1,w_1)$ \quad }  &  $w_2\chi_2$  \\

\hline

$0 < n \leq B, b_1\geq 1 $ & \text{$(n,b_1,0)$ \quad } &	\text{$(n-1,b_1-1,0)$ \quad } &	$b_1\mu_1$ \\

$n>B, b_1 \geq 1, w_1=0$ & \text{$(n,b_1,0)$ \quad } &	\text{$(n-1,b_1-1,0)$ \quad } &	$b_1 \mu_1$ \\

$n>B, b_1 \geq 1, w_1=n-B$ & \text{$(n,b_1,w_1)$ \quad } &	\text{$(n-1,b_1,w_1-1)$ \quad } &	$b_1 \mu_1$ \\

$n>B, b_1 \geq 1, 1 \leq w_1 < n-B$ & \text{$(n,b_1,w_1)$ \quad } &	\text{$(n-1,b_1,w_1-1)$ \quad } &	$r_1 b_1 \mu_1$ \\

$n>B, b_1 \geq 1, 1 \leq w_1 < n-B$ & \text{$(n,b_1,w_1)$ \quad } &	\text{$(n-1,b_1-1,w_1)$ \quad } &	$r_2 b_1 \mu_1$ \\

\hline

$w_2 \geq 1$ &$(n,b_1,w_1)$ &	$(n,b_1,w_1+1)$&	$w_2 \beta_{2\to 1}$\\

\hline

$0 < n \leq B, b_2\geq 1$ & \text{$(n,b_1,0)$ \quad } &	\text{$(n-1,b_1,0)$ \quad } &	$b_2\mu_2$ \\

$n>B, b_2 \geq 1, w_1=0$ & \text{$(n,b_1,0)$ \quad } &	\text{$(n-1,b_1,0)$ \quad } &	$b_2 \mu_2$\\

$n>B, b_2 \geq 1, w_1 = n-B$ & \text{$(n,b_1,w_1)$ \quad } &	\text{$(n-1,b_1+1,w_1-1)$ \quad } &	$b_2 \mu_2$\\

$n>B, b_2 \geq 1, 1 \leq w_1 < n-B$ & \text{$(n,b_1,w_1)$ \quad } &	\text{$(n-1,b_1+1,w_1-1)$ \quad } &	$r_1 b_2 \mu_2$ \\

$n>B, b_2 \geq 1, 1 \leq w_1 < n-B$ & \text{$(n,b_1,w_1)$ \quad } &	\text{$(n-1,b_1,w_1)$ \quad } &	$r_2 b_2 \mu_2$ \\

   \hline
\end{tabular}
		\caption{Transition rates of the QBD process $\{(X(t),\varphi(t)):t\geq 0\}$.}
\label{tab:parametersQBDI-2ISYM}
\end{table}

\subsection{Size of block matrices ${\bf Q}^{[n,n']}$ in the generator}\label{sec:sizeQnn}

To assist the computation, we derive a formula for the number of rows (and columns) in each block $\mathbf{Q}^{[n,n']}$. This provides a reproducible mapping from states $(n,b_1,w_1)$ to matrix positions and forms the foundation for the computations carried out in later sections.

For a fixed level $n$, the total number of customers, we arrange states $(n,b_1,w_1)$ for all possible values of $b_1$ and $w_1$ in the lexicographical order according to $(n,0,0)$, $(n,0,1)$, $\ldots$ , $(n,0,max\{w_1\})$, $(n,1,0)$, $\ldots$ , $(n,1,max\{w_1\})$, $\ldots$ , $(n,max\{b_1\},0)$, $\ldots$ , $(n,max\{b_1\},max\{w_1\})$. Denote by $i(n,b_1,w_1)$, the position of state $(n,b_1,w_1)$ in this sequence.

If $n\leq B$, there would be no customers in the waiting area. In this case, $max\{w_1\}=0$ and
\begin{eqnarray}
    i(n,b_1,0)=b_1+1.
    \label{RowNo_nlessB}
\end{eqnarray}
If $n>B$, then $n-B=max\{w_1\}\geq 1$ customers would be in the waiting area. In this case, first note that the total number of states in the sequence of states $(n,0,0),\ldots,(n,b_1-1,\max\{w_1\})$ is $b_1(n-B+1)$. Indeed, there are $b_1$ entries in the sequence $0,\ldots ,b_1-1$ and for each of these, there are $(n-B+1)$ entries in the sequence $0,\ldots,\max\{w_1\}$. Second, note that there are $w_1+1$ entries in the sequence $(n,b_1,0),\ldots,(n,b_1,w_1)$. Therefore,  
\begin{eqnarray}
    i(n,b_1,w_1)=b_1(n-B+1)+w_1+1.
    \label{RowNo_ngreaterB}
\end{eqnarray}
Thus, the conversion from state $(n,b_1,w_1)$ to row number $i(n,b_1,w_1)$ can be written as
\begin{eqnarray}
    i(n,b_1,w_1)&=& \left\{
\begin{array}{ll}
b_1(n-B+1)+w_1+1&\mbox{when }n> B 
\\
b_1+1&\mbox{when }n\leq B.\nonumber
\end{array}
\label{eq:rows_no}
\right.\\
&=&
    b_1\left(\max\{n-B+1,1\} \right)+w_1+1.
    \label{RowNo_ngreaterB}
\end{eqnarray}
Given $n$ and $B$, the minimum row number corresponds to $b_1=0$ and $w_1=0$. Putting these values in~\eqref{RowNo_ngreaterB} gives 
\begin{eqnarray*}
i(n,b_1,w_1)&=&\lefteqn{
b_1\left(\max\{n-B+1,1\} \right)+w_1+1
}
\\
&=&
0\left(\max\{n-B+1,1\} \right)+0+1
\\
&=&1.
\end{eqnarray*}
The maximum row number is given by
\begin{eqnarray}
i(n,b_1,w_1)&=&\lefteqn{
b_1\left(\max\{n-B+1,1\} \right)+w_1+1
}
\nonumber
\\
&=&
min\{n,B\}\left(\max\{n-B+1,1\} \right)+\max\{0,n-B\}+1
\nonumber
\\
&=&
\left\{
\begin{array}{ll}
(n-B+1)(B+1)&\mbox{when }n> B
\\
n+1&\mbox{when }n\leq B.
\end{array}
\label{eq:no_rows}
\right.
\end{eqnarray}

We note that, for a fixed $B$, $i(n,b_1,w_1)$ is a linear function of $n$ which means that the number of rows in ${\bf Q}^{[n,n]}$ increases as the number of patients in the system increases, in a linear manner. Also, for a fixed $n>B$, we have
\begin{eqnarray*}
i(n,b_1,w_1)&=&(n-B+1)(B+1)
\\
&=&
-B^2+nB+n+1,
\end{eqnarray*}
which means that the number of rows in ${\bf Q}^{[n,n]}$ as a function of $B$, increases with the maximum reached when $B=\lfloor n/2 \rfloor$, and then decreases until $B=n-1$. The number of columns in ${\bf Q}^{[n,n]}$ can be calculated using~\eqref{eq:no_rows}. Similarly, the number of columns in  ${\bf Q}^{[n,n+1]}$ and ${\bf Q}^{[n,n-1]}$ can be calculated by using appropriate substitution in~\eqref{eq:no_rows}.

By the analysis above, the total number of rows (or columns) in {\bf Q} can be obtained using 
\begin{eqnarray}
\sum_{n=0}^{N}
\ 
\left(
\min\{n,B\}\left(\max\{n-B+1,1\} \right)+\max\{0,n-B\}+1
\right).
\label{eq:NumRows}
\end{eqnarray}

We are also interested in evaluating state $(n,b_1,w_1)$ given $i$ such that $i(n,b_1,w_1)=i$. First, we evaluate $w_1$ using
\begin{eqnarray}
   w_1&=& (i-1)\mod(\max\{n-B+1,1\})
\end{eqnarray}
and then $b_1$ using
\begin{eqnarray}
 b_1&=&   \frac{i-1-w_1}{\max\{n-B+1,1\}}.
\end{eqnarray}

\subsection{Waiting times: Type~$1$ customers ($\beta_{2\to1}\geq 0$, $r_1=1$, $r_2=0$)}
\label{sec:WT_Type1nochange}

Suppose that Type~$1$ customers have priority (i.e., $r_1 = 1$, $r_2 = 0$) and Type~$2$ customers may change to Type~$1$ with a non-negative rate (i.e, $\beta_{2\to1}\geq0$). Then, to evaluate $\widetilde{F}_1(n, b_1, m_1)(s)$, we apply the following approach. We observe that the distribution of $T_1(n, b_1, m_1)$ is influenced only by departures and by the abandonment of Type~$1$ customers who are \emph{ahead} of the tagged $m_1$-th Type~$1$ customer (conditional on the tagged customer not abandoning, since their waiting time is of interest). In contrast, the distribution of $T_1(n, b_1, m_1)$ is unaffected by any future arrivals and by the abandonment of Type~$2$ customers as well as only those Type~$1$ customers positioned \emph{behind} the $m_1^{th}$ Type~$1$ customer. This is because such customers either arrive or abandon strictly \emph{behind} the $m_1^{th}$ Type~$1$ customer in the queue, and therefore do not affect the position of the $m_1^{th}$ Type~$1$ in the queue. Therefore, arrivals of both customer types, abandonment of Type~$2$ customers, and abandonment of Type~$1$ customers behind the $m_1^{th}$ Type~$1$ customer can all be omitted from the analysis. We also observe that all Type~$2$ customers in the waiting area (if any), Type~$2$ customers transitioning to Type~$1$, and any Type~$1$ customers positioned behind the $m_1$-th Type~$1$ customer, will only be served after the $m_1$-th Type~$1$ customer. Therefore, such customers can also be excluded from the analysis.

We now consider a censored QBD process $\{(Y(t), \varphi(t)) : t \geq 0\}$ with initial state $(Y(0), \varphi(0)) = (n, b_1, m_1)$, $n=B+m_1$, where only departures are observed. The generator of this process is obtained by removing transitions corresponding to arrivals and type change from the original QBD process $\{(X(t), \varphi(t)) : t \geq 0\}$. The transition rates for this censored QBD are given in Table~\ref{tab:parametersQBDI-3I}. The key observation is that the distribution of $T_1(n, b_1, m_1)$ is equivalent to the distribution of the time taken by the censored QBD to first reach the state $(n - m_1, b_1', 0)$ for some $b_1'$, starting from $(n, b_1, m_1)$. Therefore, we have the following result.

\begin{lemma}
\label{lemma1}
We have, 
\begin{eqnarray}
\widetilde F_1(n,b_1,m_1)(s) &=& 
\sum_{(b_1',m_1')}
\widetilde {G_Y}_{(b_1,m_1)(b_1^{'},m_1')}^{n,n-m_1}(s)
=
\sum_{b_1'=\max\{1,b_1\}}^{b_1+\min\{m_1,b_2\}}\widetilde {G_Y}_{(b_1,m_1)(b_1^{'},0)}^{n,n-m_1}(s),
\label{Eq:LST_Type1}
\end{eqnarray}
where $\widetilde {G_Y}_{(b_1,m_1)(b_1^{'},0)}^{n,n-m_1}(s)$ is the LST of the time for the censored QBD to first visit level $n - m_1$ and do so in phase $(b_1', 0)$, given it starts at level $n$ in phase $(b_1, m_1)$. 
\end{lemma}
\noindent{\bf Proof:} The summation bounds in Equation~\eqref{Eq:LST_Type1} arise from the following possible system evolutions until the $m_1$-th Type~$1$ customer is admitted to a server.
\begin{itemize}
    \item If $b_1 \neq 0$, and only Type~$1$ customers depart during the evolution of the system. Then each departing Type~$1$ customer is replaced by a Type~$1$ customer from the waiting area, resulting in the state $(n - m_1, b_1, 0)$. If $b_1 = 0$, and a single Type~$2$ customer departs first, followed by Type~$1$ departures only. This gives $b_1' = 1$. Thus, the lower bound is $b_1' = \max\{1, b_1\}$.
    
    \item If $b_1 \neq B$, and only one Type~$2$ customer along with some Type~$1$ customers departs, the resulting state is $(n - m_1, b_1 + 1, 0)$.
    
    \item If $b_1 \neq B$, and all Type~$2$ customers depart. In this case, the resulting state would be $(n - m_1, b_1 + m_1, 0)$ if $m_1 < b_2$, and $(n - m_1, b_1+b_2, 0)=(n - m_1, B, 0)$ if $m_1 \geq b_2$. Combining both cases, the upper bound becomes $b_1' = b_1 + \min\{m_1, b_2\}$.
\end{itemize}
To compute $\widetilde {G_Y}_{(b_1,m_1)(b_1^{'},0)}^{n,n-m_1}(s)$ for $b_1'=\max\{1,b_1\},\ldots,b_1+\min\{m_1,b_2\}$, we consider the matrix $\widetilde {\bf G}_Y^{n,n-m_1}(s)=[\widetilde {G_Y}_{(b_1,m_1)(b_1^{'},0)}^{n,n-m_1}(s)]$, and follow the matrix-analytic methods approach given in~\textcite{aksamit2024sensitivities}. Here, the blocks $\mathbf Q^{[n,n]}$, $\mathbf Q^{[n,n\pm1]}$ are those of the censored QBD (Table~\ref{tab:parametersQBDI-3I}), and $s\geq 0$. We first compute the one-step first-visit matrices $\widetilde{\bf G}^{n,n-1}(s)$ for $n = N,\dots,n-m_1+1$ using
\begin{align*}
\widetilde{\bf G}^{N,N-1}(s) 
&= -({\bf Q}^{[N,N]} - s{\bf I})^{-1}{\bf Q}^{[N,N-1]}, \\
\widetilde{\bf G}^{n,n-1}(s) 
&= -\Big({\bf Q}^{[n,n]} - s{\bf I} 
+ {\bf Q}^{[n,n+1]}\widetilde{\bf G}^{n+1,n}(s)\Big)^{-1}
{\bf Q}^{[n,n-1]}.
\end{align*}
We then calculate $\widetilde {\bf G}_Y^{n,n-m_1}(s)$ using the following recursion
\begin{equation*}
\widetilde{\bf G}^{n,n-m_1}(s) = \widetilde{\bf G}^{n,n-1}(s)\widetilde{\bf G}^{n-1,n-2}(s)\cdots\widetilde{\bf G}^{n-m_1+1,n-m_1}(s),
\end{equation*}
and extract the required entries $\widetilde {G_Y}_{(b_1,m_1)(b_1^{'},0)}^{n,n-m_1}(s)$, $b_1'=\max\{1,b_1\},\ldots,b_1+\min\{m_1,b_2\}$, from the matrix $\widetilde{\bf G}^{n,n-m_1}(s)$.
\rule{9pt}{9pt}

\begin{table}[!htbp]
		\centering
        
		\setlength{\tabcolsep}{1pt}	
		\begin{tabular} {l@{\hspace{0.5cm}} l@{\hspace{0.5cm}}l l}

  \hline
   	$\text{Conditions}$&	\text{Current state} &	$\text{New state}$ &	$\text{Rate}$\\

\hline

$0 < n \leq B, b_1\geq 1 $ & \text{$(n,b_1,0)$ \quad } &	\text{$(n-1,b_1-1,0)$ \quad } &	$b_1\mu_1$ \\

$n>B, b_1 \geq 1, w_1=0$ & \text{$(n,b_1,0)$ \quad } &	\text{$(n-1,b_1-1,0)$ \quad } &	$b_1 \mu_1$ \\

$n>B, b_1 \geq 1, 1 \leq w_1 \leq n-B$& $(n,b_1,w_1)$ &	$(n-1,b_1,w_1-1)$  &	$b_1 \mu_1$  \\

\hline

$0 < n \leq B, b_2\geq 1$ & \text{$(n,b_1,0)$ \quad } &	\text{$(n-1,b_1,0)$ \quad } &	$b_2\mu_2$ \\

$n>B, b_2 \geq 1, w_1=0$ & \text{$(n,b_1,0)$ \quad } &	\text{$(n-1,b_1,0)$ \quad } &	$b_2 \mu_2$\\

$n>B, b_2 \geq 1, 1 \leq w_1 \leq n-B$  & \text{$(n,b_1,w_1)$ \quad }  &	\text{$(n-1,b_1+1,w_1-1)$ \quad }  &	$b_2 \mu_2$  \\

   \hline

$n>B, w_1 \geq 2$ & \text{$(n,b_1,w_1)$ \quad }  &	\text{$(n-1,b_1,w_1-1)$ \quad }  &	$(w_1-1) \chi_1$  \\

\hline
  \end{tabular}
		\caption{Off-diagonals of the generator in the censored QBD process $\{(Y(t),\varphi(t)):t\geq 0\}$ discussed in Section~\ref{sec:WT_Type1nochange}. The on-diagonals can be calculated using~\eqref{eq:on-diagonals}.}
\label{tab:parametersQBDI-3I}
\end{table}

\subsection{Waiting times: Type~$2$ customers ($\beta_{2\to1}=0$, $r_1=1$, $r_2=0$)
} \label{sec:WT_Type2nochange}

We suppose that Type~$1$ customers have priority (i.e., $r_1 = 1$, $r_2 = 0$) and Type~$2$
customers in the waiting area do not change their type, that is, $\beta_{2\to1} = 0$. To evaluate $\widetilde{F}_2(n, b_1, w_1, m_2)(s)$, we proceed as follows.

We observe that the distribution of $T_2(n, b_1, w_1, m_2)$ is influenced only by the departures,  arrivals of Type~$1$ customers, abandonment of Type~$1$ customers, and abandonment of those Type~$2$ customers positioned \emph{ahead} of the $m_2$-th Type~$2$ customer. In contrast, the distribution of $T_2(n, b_1, w_1, m_2)$ is not affected by the arrivals of Type~$2$ customers and abandonment those Type~$2$ customers positioned \emph{behind} the $m_2$-th Type~$2$ customer. This is because newly arriving Type~$2$ customers join Queue~$2$ at the end and are served after the $m_2$-th Type~$2$ customer. On the other hand, Type~$1$ customers join Queue~$1$, are positioned ahead of Queue~$2$, and are served before the $m_2$-th Type~$2$ customer. Additionally, Type~$2$ customers positioned behind the $m_2$-th customer in the queue are either served later or get abandoned, both events do not affect the waiting time of $m_2$-th Type~$2$ customer. Therefore, newly arriving Type~$2$ customers, those positioned after the $m_2$-th Type~$2$ customer in the waiting area, and those getting abandoned behind the $m_2$-th Type~$2$ customer can be excluded from the analysis, and we take $n=B+w_1+m_2$ at $t=0$.

We now consider a censored QBD process $\{(Z(t), \varphi(t)) : t \geq 0\}$ with initial state $(Z(0), \varphi(0)) = (n, b_1, w_1)$, $n=B+w_1+m_2$, in which departures of both types and arrivals of only Type~$1$ customers are observed. The generator of this censored QBD is constructed by removing transition rates corresponding to Type~$2$ arrivals and type change from the original QBD process $\{(X(t), \varphi(t)) : t \geq 0\}$. The transition rates for this censored QBD are given in Table~\ref{tab:parametersQBDI-4I}. The key observation is that the distribution of $T_2(n, b_1, w_1, m_2)$ is equivalent to the distribution of the time taken by the censored QBD to first reach the state $(B, b_1', 0)$, given the initial state $(n, b_1, w_1)$ with $n - B - w_1 = m_2$. Therefore, we have the following result.

\begin{lemma}
\label{lemma2}
We have,
\begin{equation}
\widetilde{F}_2(n, b_1, w_1, m_2)(s) = \sum_{(b_1',w_1')} 
\widetilde{G_Z}^{n, B}_{(b_1, w_1)(b_1', 0)}(s)
= 
\sum_{b_1' = b_1 - \min\{m_2, b_1\}}^{B - 1} 
\widetilde{G_Z}^{n, B}_{(b_1, w_1)(b_1', 0)}(s),
\label{formula_Type2}
\end{equation}
where $\widetilde{G_Z}^{n, B}_{(b_1, w_1)(b_1', 0)}(s)$ is the LST of the time for the censored QBD to first reach level $n = B$ in phase $(b_1', 0)$, given it starts at level $n$ in phase $(b_1, w_1)$. 

\end{lemma}
\noindent{\bf Proof:} The bounds of the summation in Equation~\eqref{formula_Type2} arise from the following possibilities of the system evolution.
\begin{itemize}
    \item First, we consider the case when only Type~$1$ customers depart, if there are any. If $m_2 \geq b_1$, then all Type~$1$ incustomers depart before the $m_2$-th Type~$2$ customer is admitted, resulting in $b_1' = 0$. If $m_2 < b_1$, then $b_1' = b_1 - m_2$. Combining both cases gives the lower bound $b_1' = b_1 - \min\{m_2, b_1\}$.
    
    \item Now, we consider the case when only Type~$2$ customers depart, and Type~$1$ customers continue to arrive until all Type~$2$ incustomers have departed just before the $m_2$-th customer goes to a server. In this case, the $m_2$-th Type~$2$ customer is admitted when $b_1' = B - 1$. In all other cases, where both Type~$1$ and Type~$2$ customers depart, the value of $b_1'$ can range between $b_1 - \min\{m_2, b_1\}$ and $B - 1$. Combining them we get the upper bound $b_1'=B-1$. 
\end{itemize}
To compute $\widetilde{G_Z}^{n, B}_{(b_1, w_1)(b_1', 0)}(s)$, we follow the approach similar to the one given in Section~\ref{sec:WT_Type1nochange}. 
\rule{9pt}{9pt}

\begin{table}[!htbp]
		\centering
        
		\setlength{\tabcolsep}{1pt}	
		\begin{tabular} {l@{\hspace{0.5cm}} l@{\hspace{0.5cm}}l l}

  \hline
   	$\text{Conditions}$&	\text{Current state} &	$\text{New state}$ &	$\text{Rate}$\\

\hline
 	$0 \leq n < B$ & \text{$(n,b_1,0)$ \quad } &	\text{$(n+1,b_1+1,0)$ \quad } &	$\lambda_1$\\

$B\leq n<N$ & \text{$(n,b_1,w_1)$ \quad } &	\text{$(n+1,b_1,w_1+1)$ \quad } &	$\lambda_1$ \\

\hline

$0 < n \leq B, b_1\geq 1 $ & \text{$(n,b_1,0)$ \quad } &	\text{$(n-1,b_1-1,0)$ \quad } &	$b_1\mu_1$ \\

$n>B, b_1 \geq 1, w_1=0$ & \text{$(n,b_1,0)$ \quad } &	\text{$(n-1,b_1-1,0)$ \quad } &	$b_1 \mu_1$ \\

$n>B, b_1 \geq 1, 1 \leq w_1 < n-B$  & \text{$(n,b_1,w_1)$ \quad }  &	\text{$(n-1,b_1,w_1-1)$ \quad }  &	$b_1 \mu_1$ \\

\hline

$0 < n \leq B, b_2\geq 1$ & \text{$(n,b_1,0)$ \quad } &	\text{$(n-1,b_1,0)$ \quad } &	$b_2\mu_2$ \\

$n>B, b_2 \geq 1, w_1=0$ & \text{$(n,b_1,0)$ \quad } &	\text{$(n-1,b_1,0)$ \quad } &	$b_2 \mu_2$\\

$n>B, b_2 \geq 1, 1 \leq w_1 < n-B$  & \text{$(n,b_1,w_1)$ \quad }  &	\text{$(n-1,b_1+1,w_1-1)$ \quad }  &	 $b_2 \mu_2$  \\

   \hline

$n>B, w_1 < n-B$  & \text{$(n,b_1,w_1)$ \quad }  &	\text{$(n-1,b_1,w_1-1)$ \quad }  &	$w_1 \chi_1$  \\

$n>B, w_1 < n-B-1$  & \text{$(n,b_1,w_1)$ \quad }  &	\text{$(n-1,b_1,w_1)$ \quad }  &	$(w_2-1) \chi_2$  \\

\hline

  \end{tabular}
		\caption{Off-diagonals in the generator of the censored QBD process $\{(Z(t),\varphi(t)):t\geq 0\}$ in Section~\ref{sec:WT_Type2nochange}. The on-diagonals can be calculated using~\eqref{eq:on-diagonals}.}

\label{tab:parametersQBDI-4I}
\end{table}

\section{Numerical examples}
\label{Numerical_examples}

We consider a hospital system with a total capacity of $N=100$ and a bed capacity of $B=80$. While these values represent a smaller-scale system, they are chosen deliberately to demonstrate the core concepts and compare the distributions obtained from the CTMC and the QBD models. The proposed methodology (CTMC-Krylov) is equally scalable and adaptable to realistically sized systems. The state space of the QBD models in our examples is of size $21,951$. This follows from the summation $\sum_{n=0}^{N} min\{n,B\}\left(\max\{n-B+1,1\} \right)+\max\{0,n-B\}+1$, given in Section~\ref{sec:sizeQnn}. The state space of the CTMC is of size $268,191$, of which $249,480$ are the transient states, follows from ~\eqref{ctmc_system_size}. We performed all the computations on Intel Core i5-1145G7 processors with 2.60 GHz CPU and 16GB RAM in Windows 11.

\subsection{Input parameters }

We classify patients into two categories: Type~$1$ (complex) and Type~$2$ (others). The QBD parameters used in our models are summarised in Table~\ref{tab:QBDparameters}. These rates are derived from a large dataset obtained from a tertiary referral hospital in Australia, with ethics approval granted by the Tasmanian Department of Health, as follows.  

First, to choose suitable values of the parameters for a moderately-sized system of capacity $N=100$, $B=80$, the departure rates $\mu_i$ and the expected lengths of stay $\mathbb{E}(\text{LoS}_i)=1/\mu_i$ for Type~$i$ patients, where $i = 1, 2$, are fitted directly from the data. We assumed that Type~$1$ (complex) patients were those whose Diagnosis Related Group (DRG) (see~\textcite{rahmawati2021cost}) category was recording a {\em major complexity}, while all other patients were assumed to be Type~$2$ (other). Then, using recorded length of stay values and standard fitting methods (command \texttt{fitdist} in Matlab) we evaluated $\mu_i$ for each patient type, assuming Poisson departure rates. Next, we chose the total arrival rate $\lambda=23.7$ and then scaled the individual arrival rates using $\lambda_i=p_i\lambda$, with $p_i$ being the proportions of Type~$i$ arrivals as obtained from the data. We note that analyses of hospital data indicate that Poisson and exponential distributions are suitable tools to model the patients' arrival and their lengths of stay in a hospital system~\parencite{whitt2017data}. These assumptions have been widely employed in queueing-based models of hospital systems and is also explicitly adopted in studies such as~\textcite{wang2025markov}.

We assume $\beta_{2\to1}=3$, so that this  corresponds to an average waiting time of $8$ hours before the change Type~$2$ $\to$ Type~$1$. It is important to note that the change Type~$2$ $\to$ Type~$1$ does not imply an instantaneous clinical transformation but rather a reclassification decision by hospital staff. Specifically, once a patient’s condition has sufficiently deteriorated, hospital managers update the patient’s label from Type~$2$ to Type~$1$ to reflect increased complexity and to include them in Queue~$1$ with higher service priority. While the models permit $\chi_1,\chi_2 \geq 0$, we take $\chi_1 = \chi_2 = 0$ in our analysis, since the data included only the details of the patients that were sufficiently sick/injured so that, once they arrived in the hospital, they would not, or maybe could not, attempt to leave and go to another hospital.

\begin{table}[H]
    \centering
    \setlength{\tabcolsep}{6pt}	
    \begin{tabular}{l l l}
        \hline
        \textbf{Parameter} & \textbf{Type~$1$ (Complex)} & \textbf{Type~$2$ (other)} \\
        \hline
        Proportion of arrivals $p_i = \lambda_i / \lambda$ & $p_1 = 0.2446$ & $p_2 = 0.7554$ \\
        Arrival rates $ \lambda_i = p_i \lambda $ (per day) & $ \lambda_1 = 5.7961 $ & $ \lambda_2 = 17.9039 $ \\
        Mean length of stay $ \mathbb{E}(\text{LoS}_i) $ (in days) & $ \mathbb{E}(\text{LoS}_1) = 6.59 $ & $ \mathbb{E}(\text{LoS}_2) = 2.43 $ \\
        Departure rate $ \mu_i = 1 / \mathbb{E}(\text{LoS}_i) $ & $ \mu_1 = 0.1517 $ & $ \mu_2 = 0.4113 $ \\

        (per patient, per day) &  &  \\

        Daily demand (bed-days) $ D_i = \lambda_i \mathbb{E}(\text{LoS}_i) $ & $ D_1 = 38.2142 $ & $ D_2 = 43.5301 $ \\
        \hline
        Transition rate from Type~$2$ to Type~$1$  &  & $ \beta_{2 \to 1} = 3 $ \\
        \hline
    \end{tabular}
    \caption{Parameters of the QBD model. }
    \label{tab:QBDparameters}
\end{table}

\subsection{Stationary analysis of the QBD model ($\beta_{2\to 1}=3$)}

We now present a stationary analysis of the QBD model described in Section~\ref{sec:QBDmodelIII}, considering various priority policies corresponding to $(r_1,r_2)$ for a system with capacity $N = 100$, and bed capacity $B = 80$. The stationary distribution vector of the QBD process is denoted by $\boldsymbol{\pi} = [\boldsymbol{\pi}_n]_{n = 0, 1, \ldots, N}$, such that
$$
\boldsymbol{\pi}_n = [\pi_{(n, b_1, w_1)}]_{b_1 = 0, 1, \ldots, B;\ w_1 = 0, 1, \ldots, N - B},
$$
where
$$
\pi_{(n, b_1, w_1)} = \lim_{t \to \infty} \mathbb{P}\left(X(t) = n,\ \varphi(t) = (b_1, w_1)\right)
$$
represents the long-run proportion of time the process spends in state $(n, b_1, w_1)$, where $n$ is the total number of patients in the system, $b_1$ is the number of Type~$1$ patients occupying beds, and $w_1$ is the number of Type~$1$ patients in the waiting area. To evaluate $\boldsymbol{\pi}$, we follow the approach given in~\textcite{aksamit2024sensitivities} and~\textcite{Gusthesis}. We first compute $\boldsymbol{\pi}_N$ using the explicit matrix equations
\begin{equation}
\boldsymbol{\pi}_N(\widehat{\bf R}^{(N-1)}{\bf Q}^{[N-1,N]}+{\bf Q}^{[N,N]}) = {\bf 0}, \\
\boldsymbol{\pi}_N\Big({\bf 1}+\sum_{n=0}^{N-1}\prod_{k=N-1}^{n}\widehat{\bf R}^{(k)}{\bf 1}\Big) = 1,
\label{eq:piN}
\end{equation}
where the matrices $\widehat{\bf R}^{(n)}=\bigl[ R^{(n)}_{(b_1,w_1)(b_1',w_1')} \bigr]$ record the expected time $R^{(n)}_{(b_1,w_1)(b_1',w_1')}$ that the process spends in states $(n,b_1',w_1')$ when it begins in state $(n+1,b_1,w_1)$, before it returns to level $n+1$, and for $n=0,1,\ldots,N-1$, $\widehat{\bf R}^{(n)}$ are calculated using the explicit recursions 
\begin{equation}
\widehat{\bf R}^{(0)} = -{\bf Q}^{[1,0]}({\bf Q}^{[0,0]})^{-1},
\widehat{\bf R}^{(n)} = -{\bf Q}^{[n+1,n]}(\widehat{\bf R}^{(n-1)}{\bf Q}^{[n-1,n]}+{\bf Q}^{[n,n]})^{-1}.
\label{eq:Rn}
\end{equation}
Finally, $\boldsymbol{\pi}_n = \boldsymbol{\pi}_{n+1}\widehat{\bf R}^{(n)} = \boldsymbol{\pi}_N\prod_{k=N-1}^{n}\widehat{\bf R}^{(k)}$ gives the probabilities $\boldsymbol{\pi}_{N-1},\dots,\boldsymbol{\pi}_0$.

The stationary marginal probabilities under different priority policies are given in Figures~\ref{stationaryQBD_r1_1_r2_0}-\ref{stationaryQBD_r1_0_r2_1}. From these figures, we observe that the system maintains an average of approximately $85$ patients, with around $36$ Type~$1$ patients occupying beds. We note though that $85\%$ is an arbitrary figure, and it would be overly simplistic to adopt it in general, as discussed in~\textcite{OReilly20251493}.  Further, we evaluate key performance metrics defined in Table~\ref{tab:keymetricsdefinition} for each priority policy. The values of these metrics are presented in Table~\ref{tab:keymetricestimations}. These results confirm that the parameters we consider (given in Table~\ref{tab:QBDparameters}), yield an average occupancy close to $85\%$. We note that as the priority for Type~$1$ (complex) patients increases, the number of complex patients in the waiting area $W_1$ and in the beds $B_1$ decreases. This is because the complex patients get priority access to service. As the priority for Type~$2$ patients increases, the number of Type~$1$ patients in the waiting area also increases, which is consistent with the system behavior.

\begin{figure}[H]
	\centering
	{\includegraphics[width=5.1 cm,height=5.1 cm]{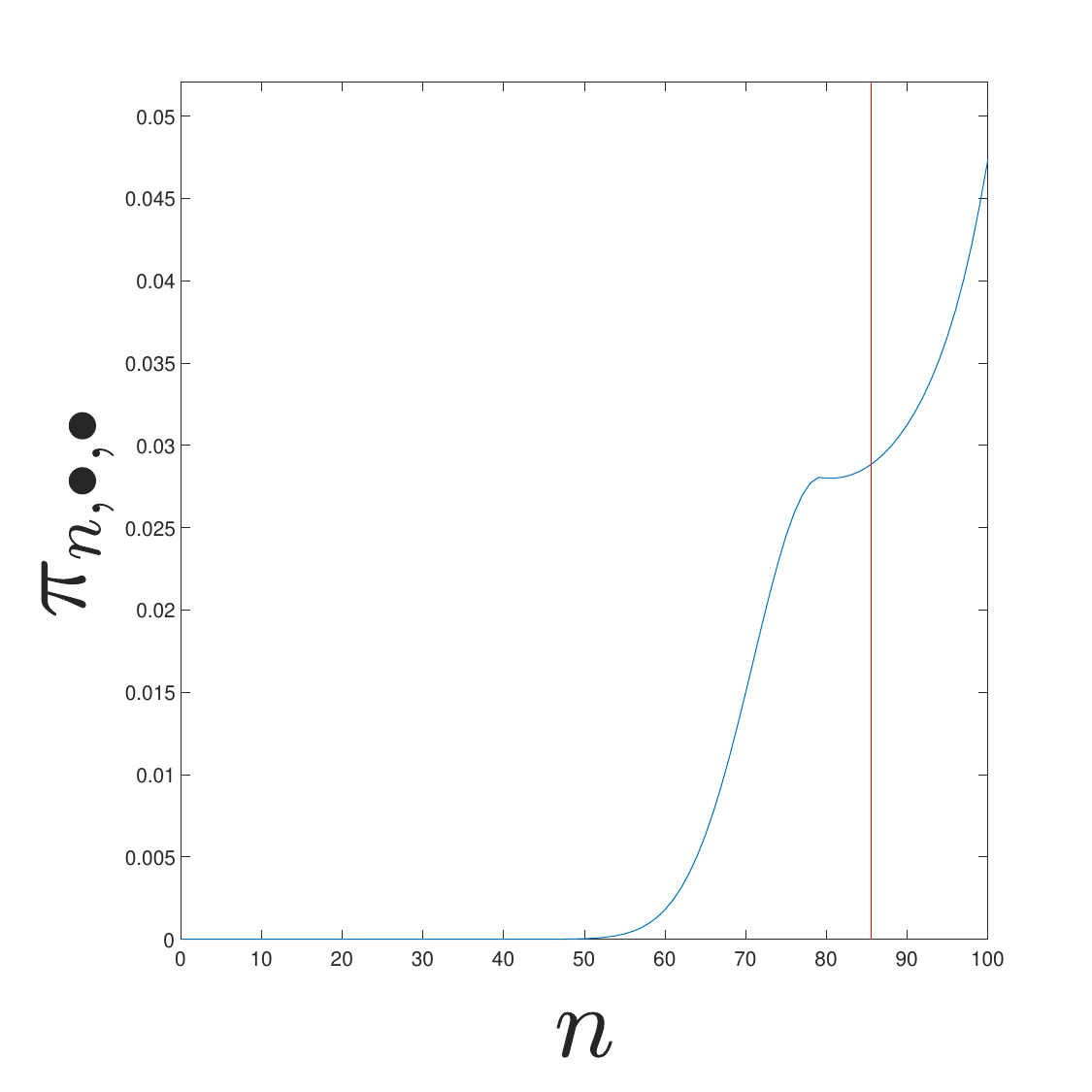}}
    \quad
	{\includegraphics[width=5.1 cm,height=5.1 cm]{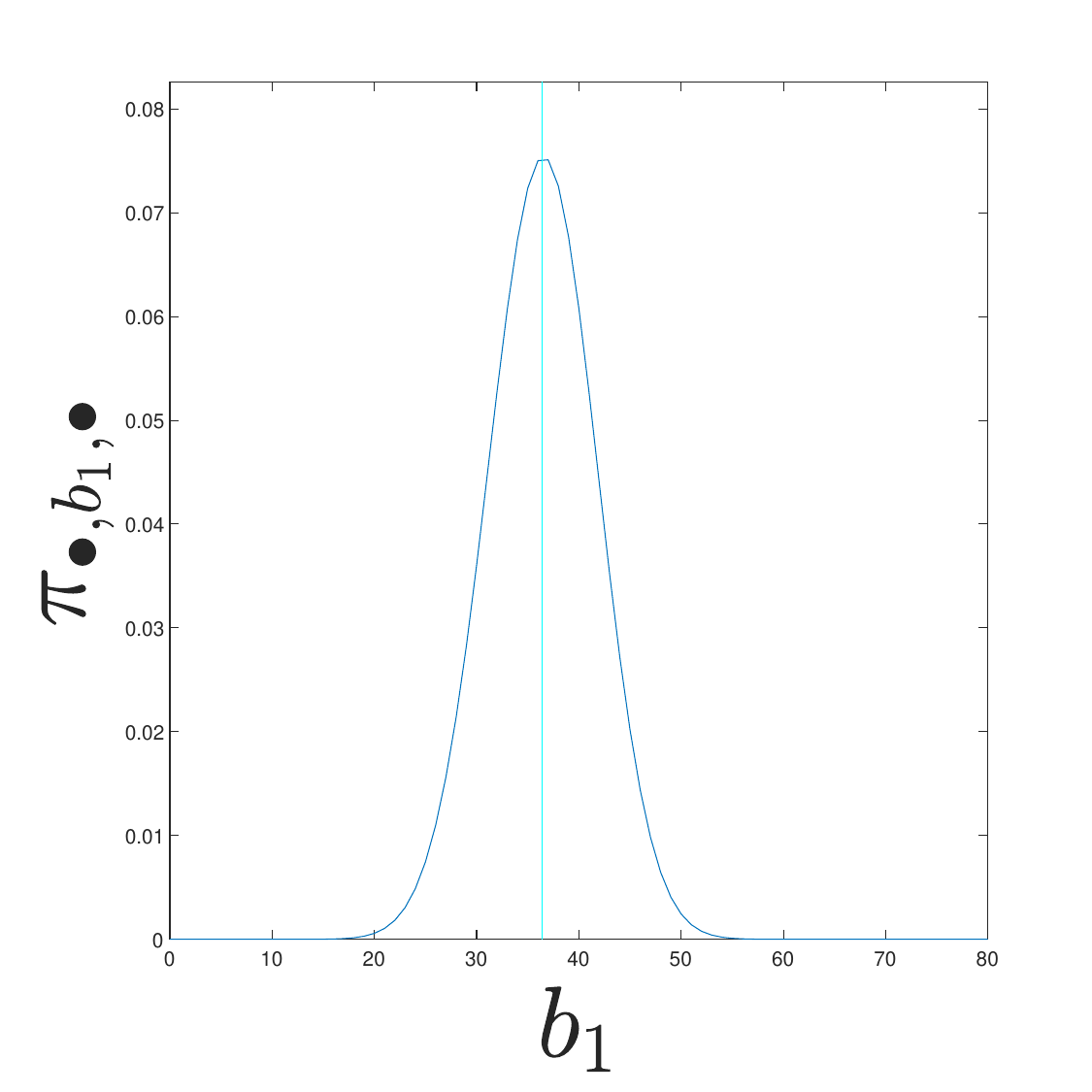}}    
    \quad
	{\includegraphics[width=5.1 cm,height=5.1 cm]{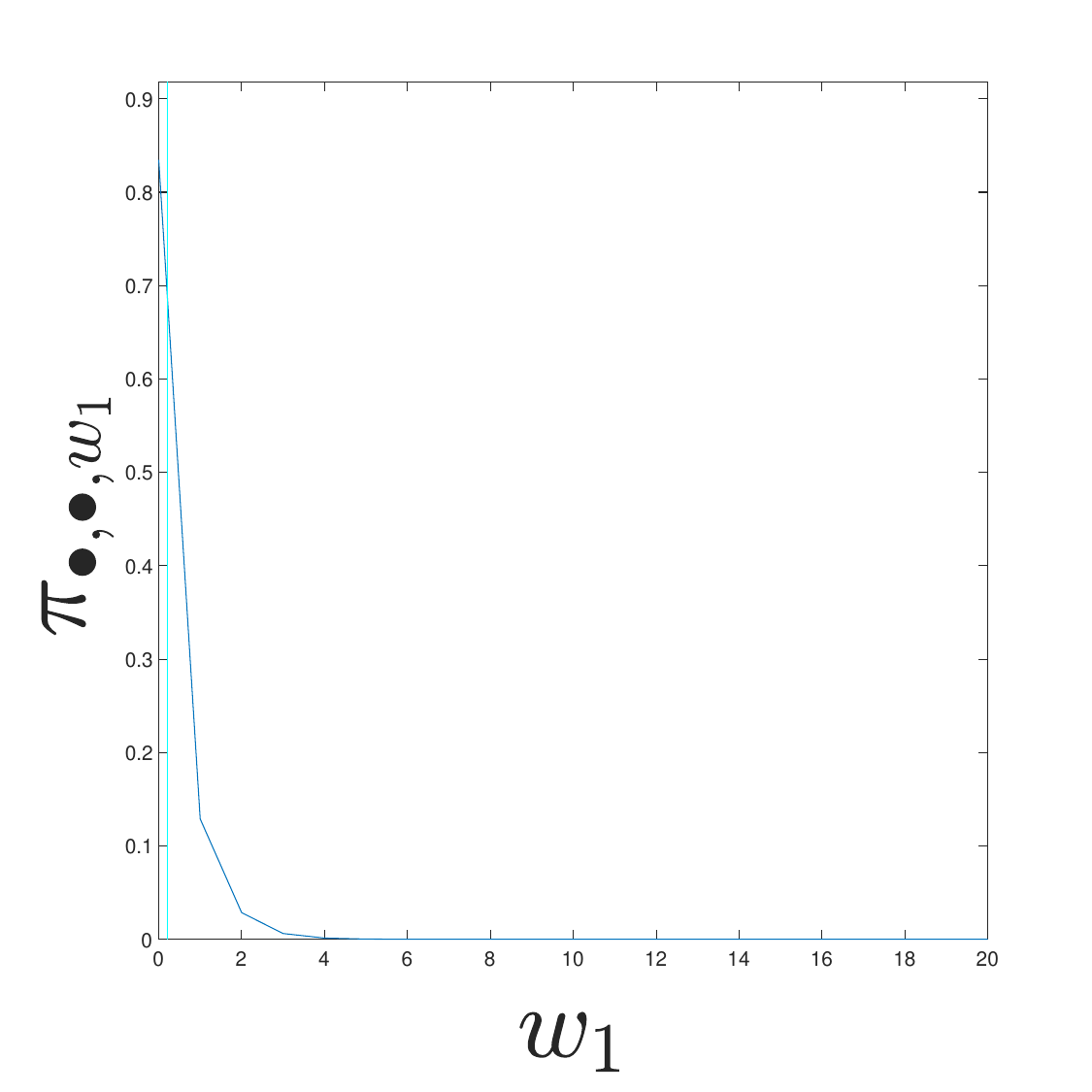}}    

    	\caption{Stationary distribution of the QBD model assuming $N=100$, $B=80$ and parameters described in Table~\ref{tab:QBDparameters}. We assume that Type~$1$ patients have priority ($r_1=1$, $r_2=0$).}
	\label{stationaryQBD_r1_1_r2_0}
\end{figure}

\begin{figure}[H]
	\centering
	{\includegraphics[width=5.1 cm,height=5.1 cm]{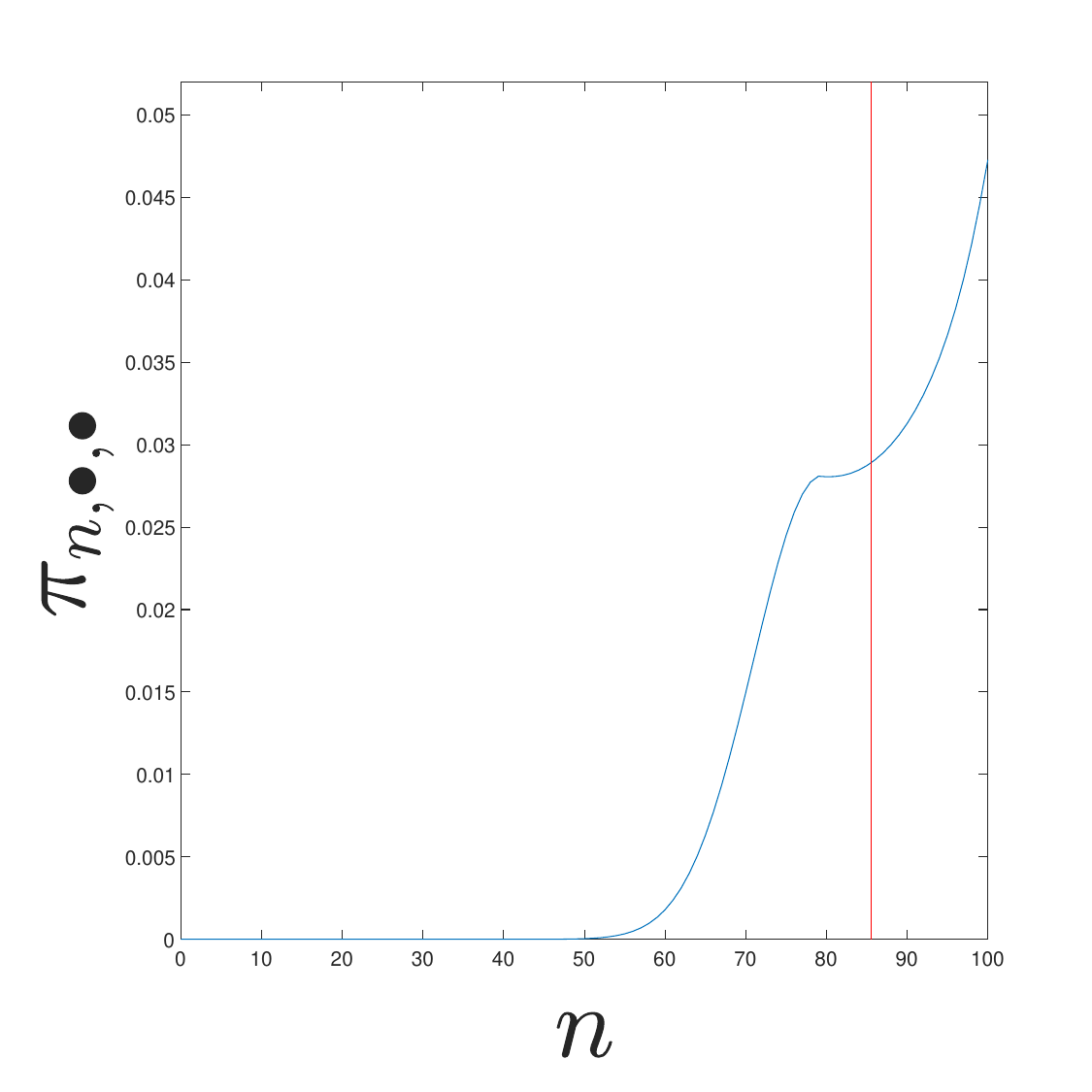}}
    \quad
	{\includegraphics[width=5.1 cm,height=5.1 cm]{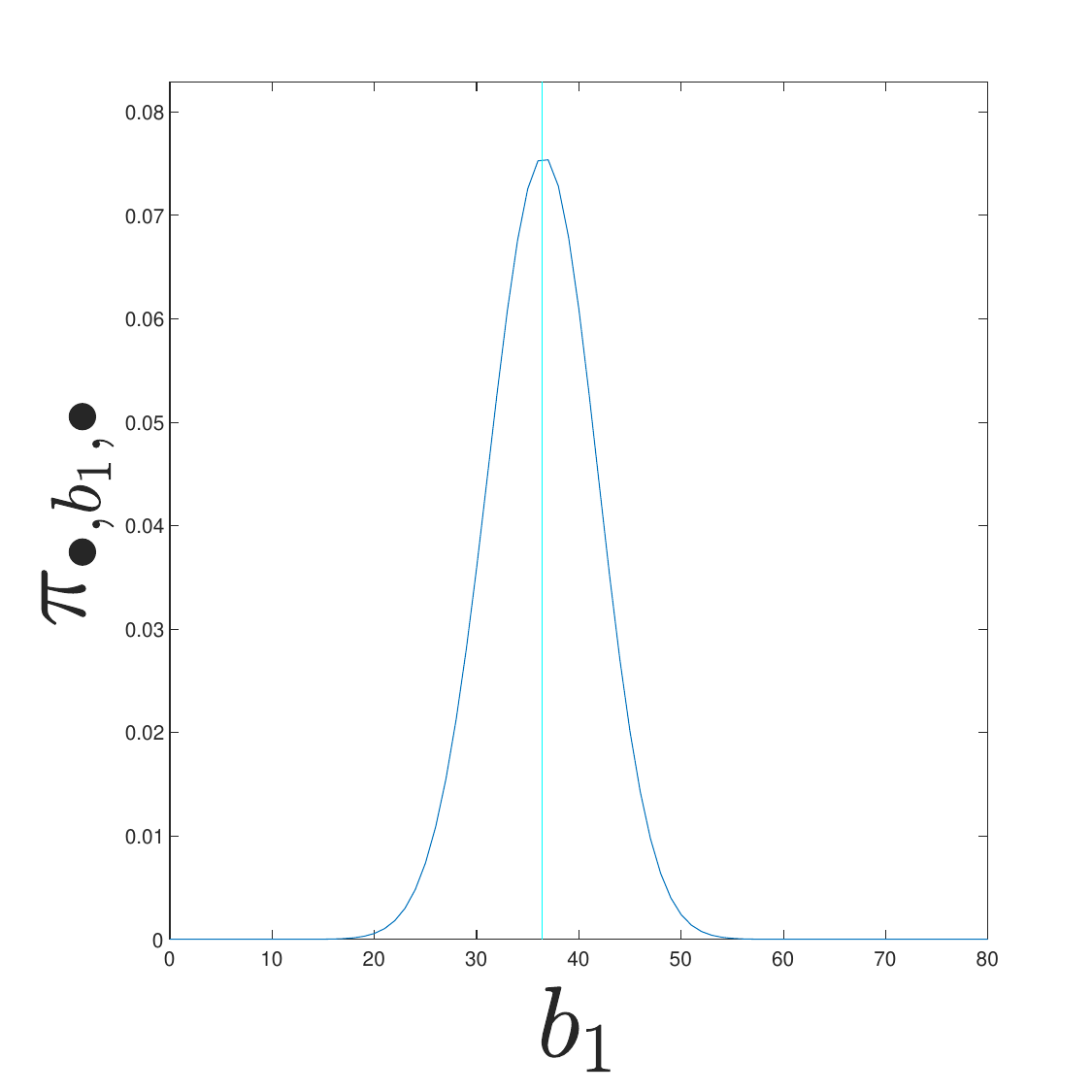}}    
    \quad
	{\includegraphics[width=5.1 cm,height=5.1 cm]{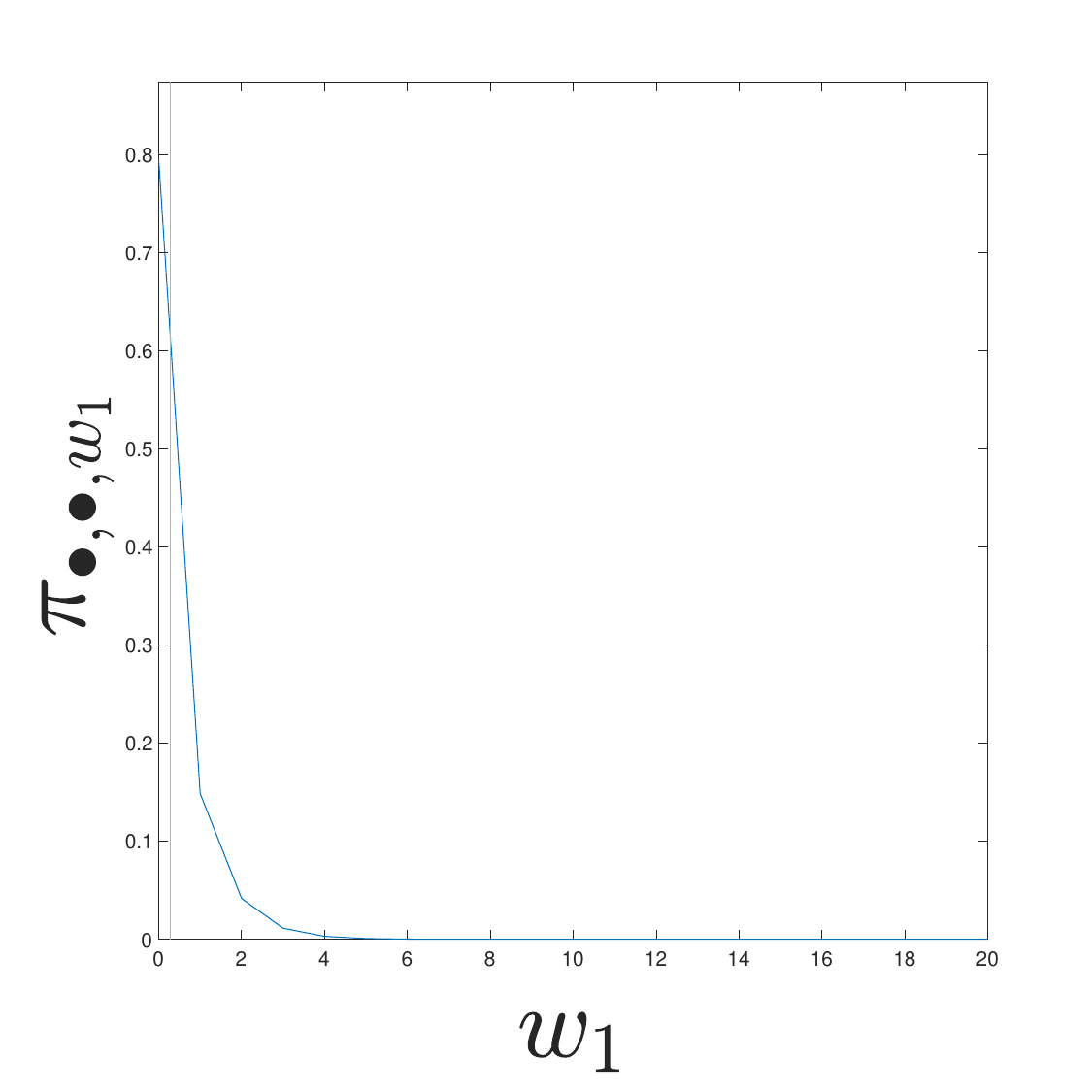}}    
	\caption{Stationary distribution of the QBD model assuming $N=100$, $B=80$ and parameters described in Table~\ref{tab:QBDparameters}. We assume that the patients have priority ($r_1=0.8$, $r_2=0.2$).}
	\label{stationaryQBD_r1_8_r2_2}
\end{figure}

\begin{figure}[H]
	\centering
	{\includegraphics[width=5.1 cm,height=5.1 cm]{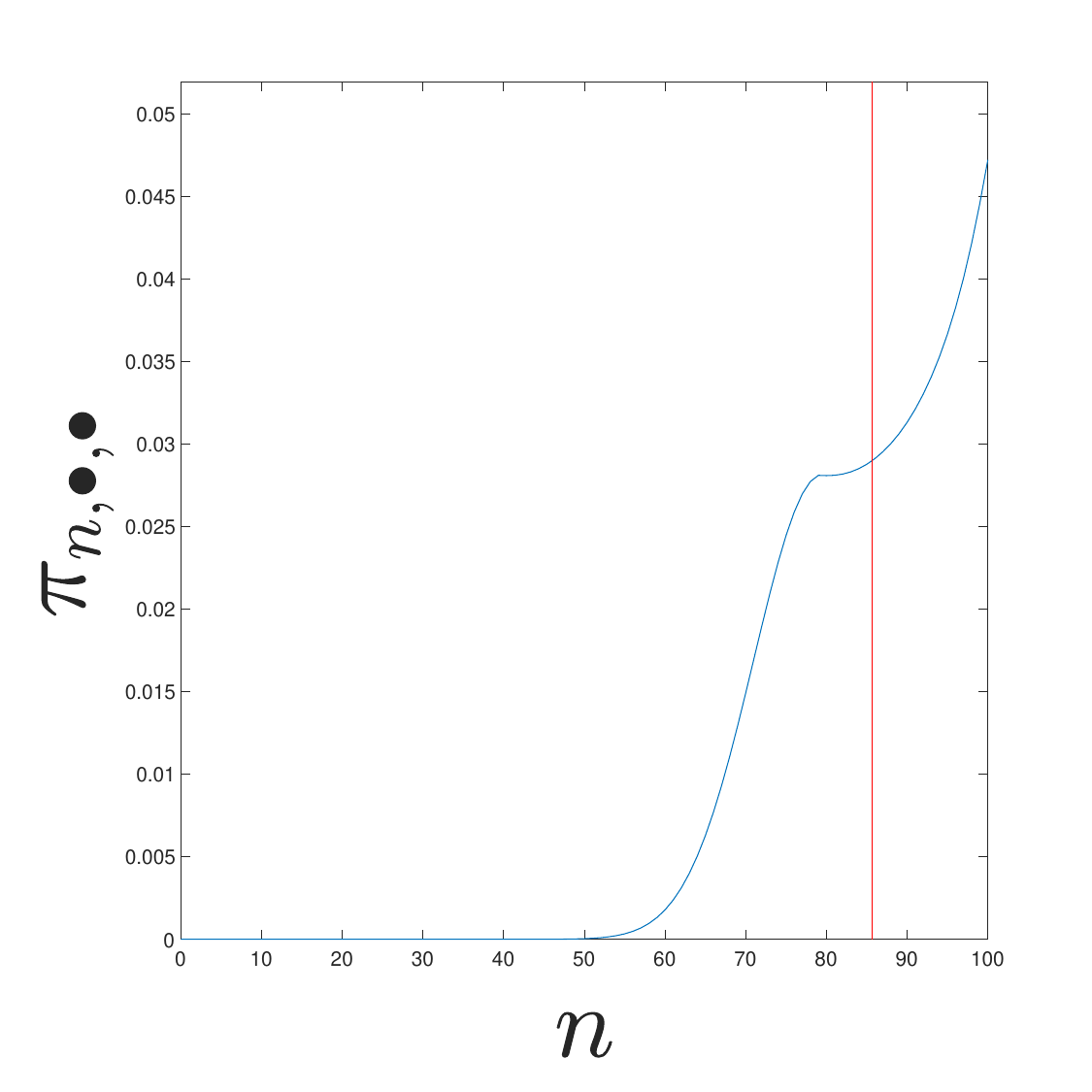}}
    \quad
	{\includegraphics[width=5.1 cm,height=5.1 cm]{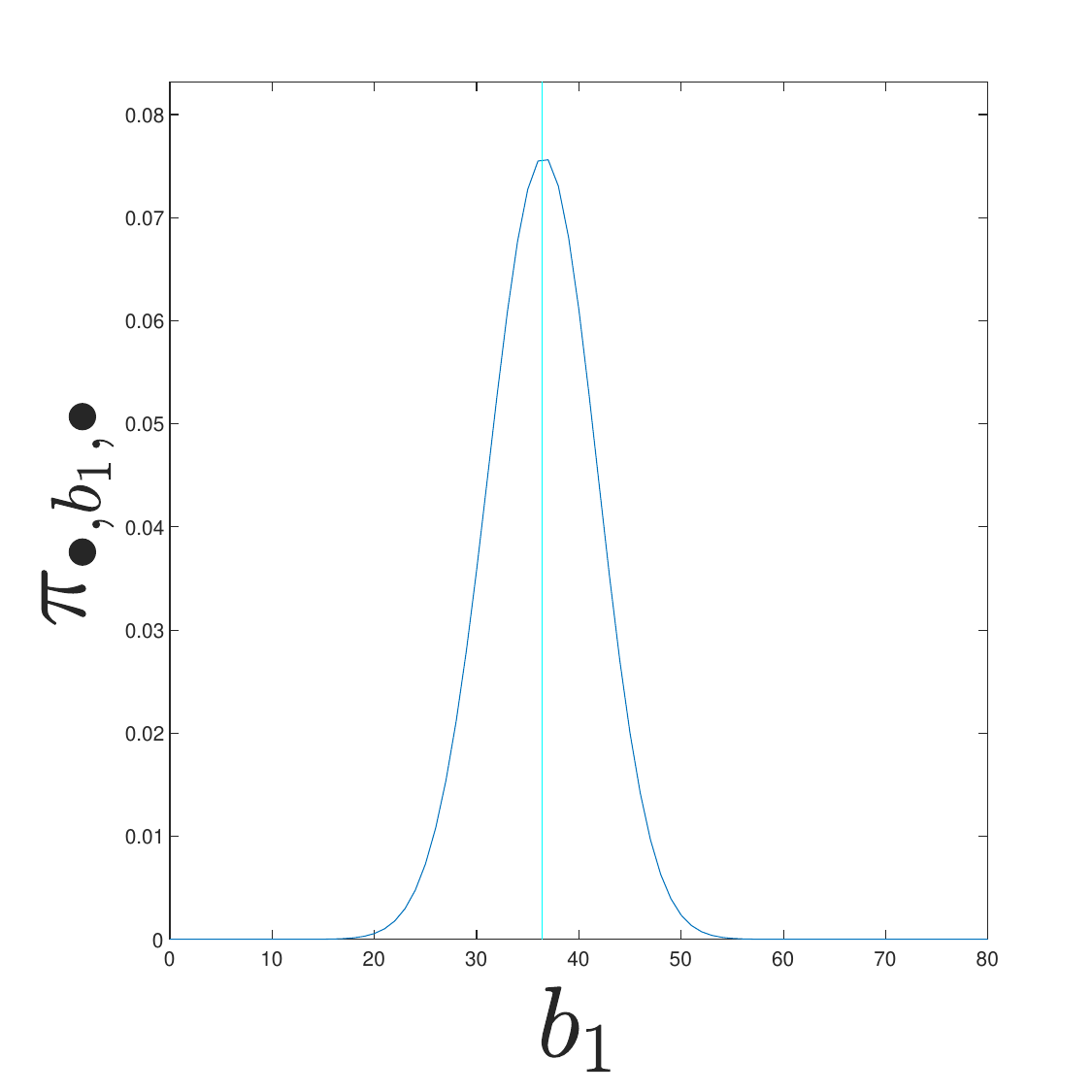}}    
    \quad
	{\includegraphics[width=5.1 cm,height=5.1 cm]{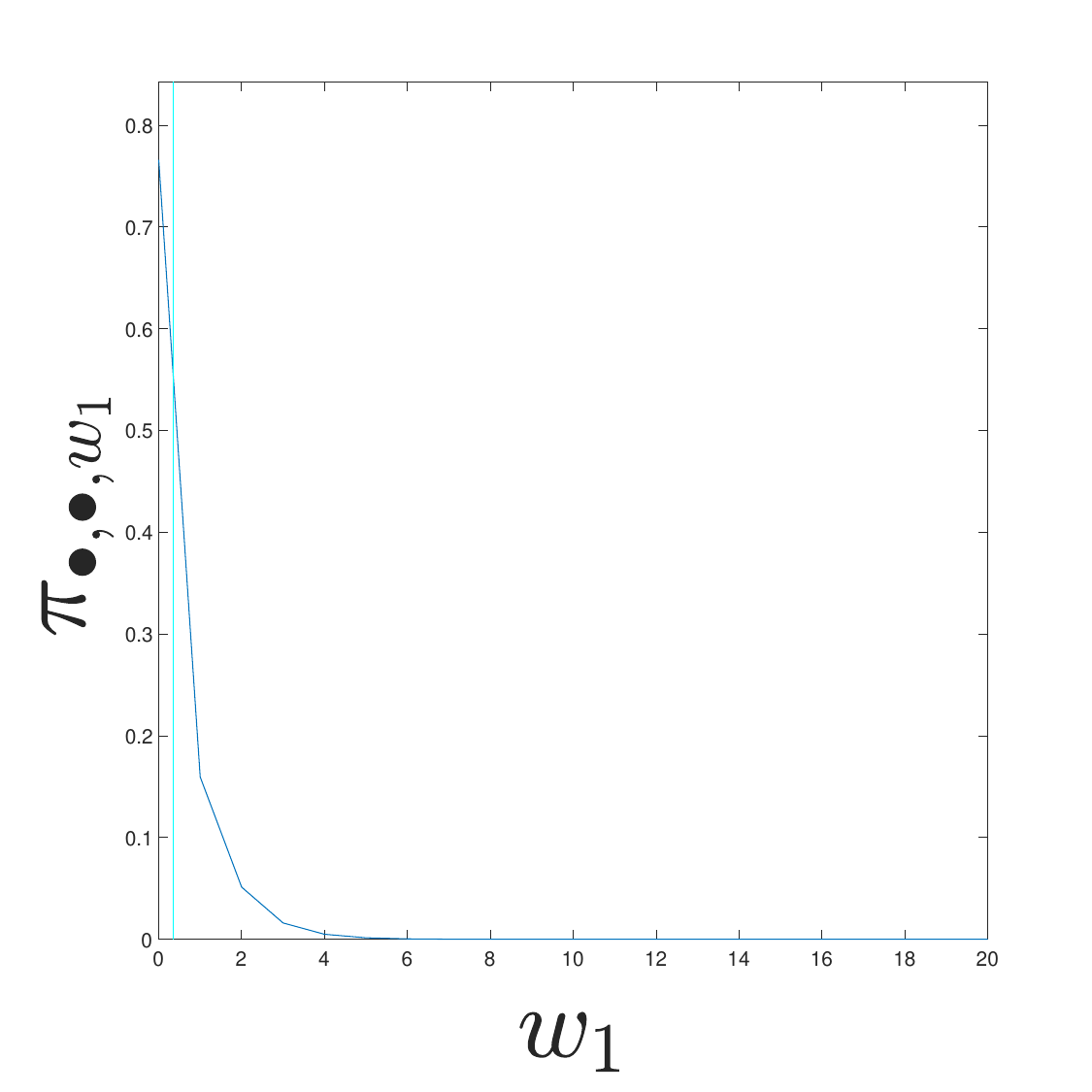}}    
	\caption{Stationary distribution of the QBD model assuming $N=100$, $B=80$ and parameters described in Table~\ref{tab:QBDparameters}. We assume that the patients have priority ($r_1=0.7$, $r_2=0.3$).}
	\label{stationaryQBD_r1_07_r2_03}
\end{figure}

\begin{figure}[H]
	\centering
	{\includegraphics[width=5.1 cm,height=5.1 cm]{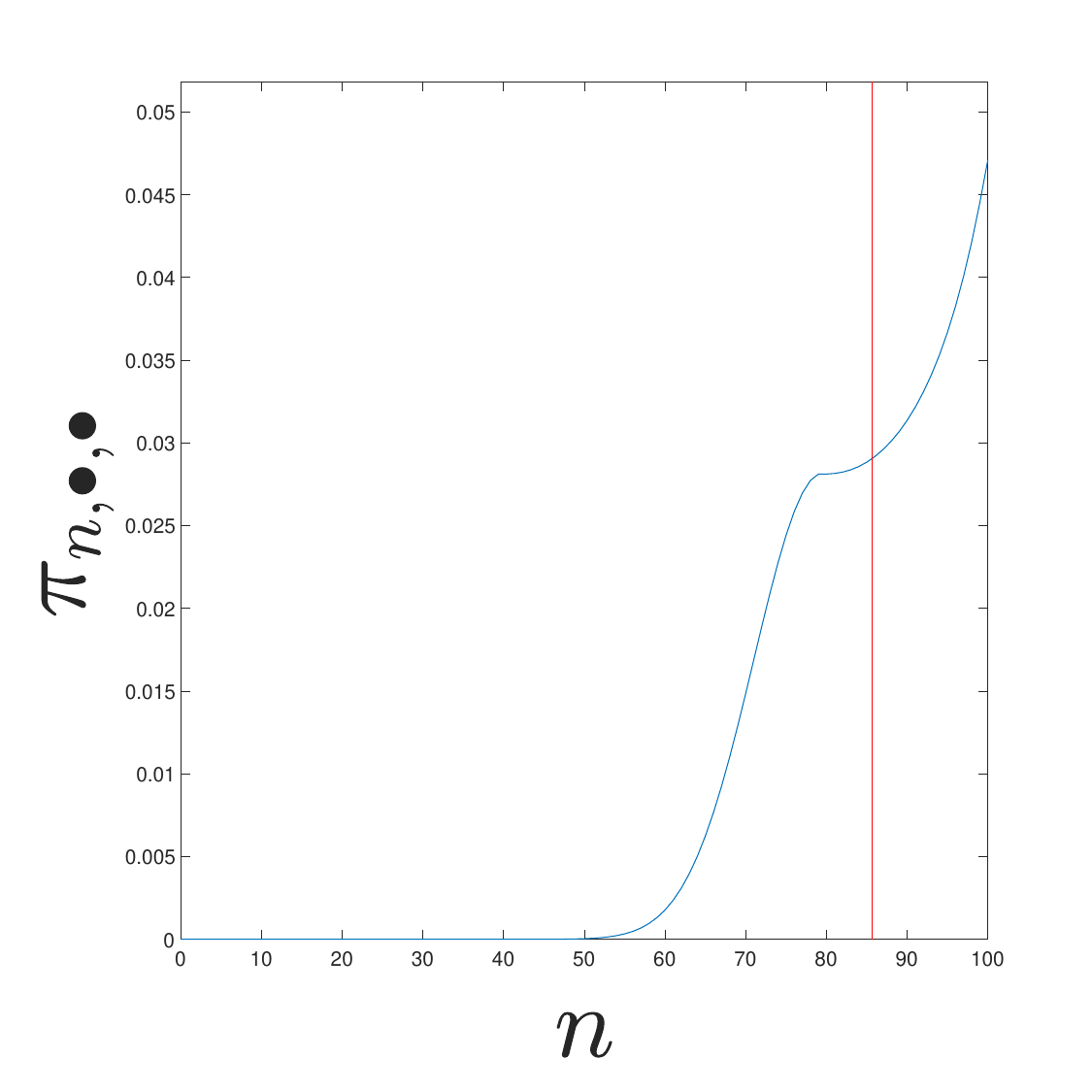}}
    \quad
	{\includegraphics[width=5.1 cm,height=5.1 cm] {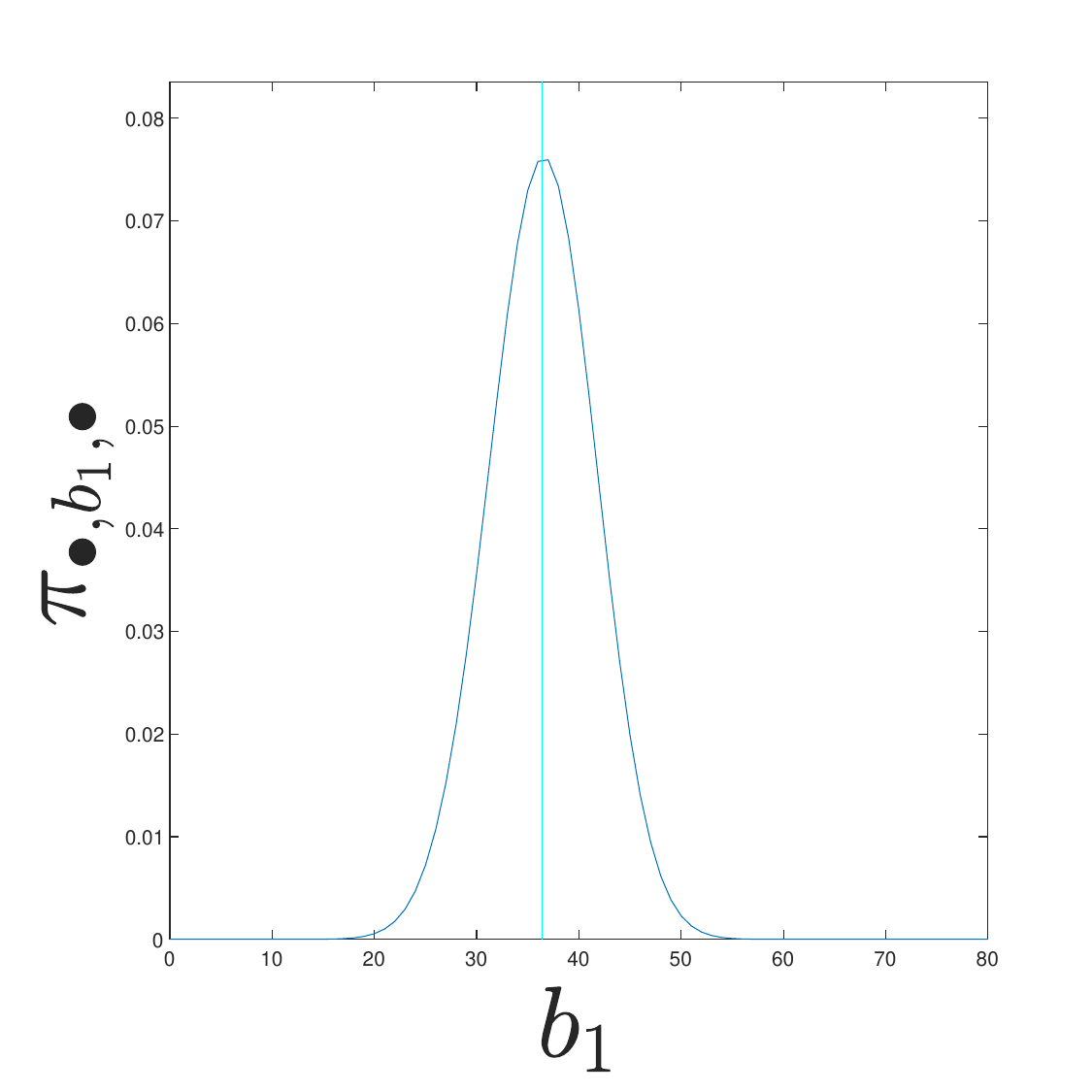}}    
    \quad
	{\includegraphics[width=5.1 cm,height=5.1 cm]{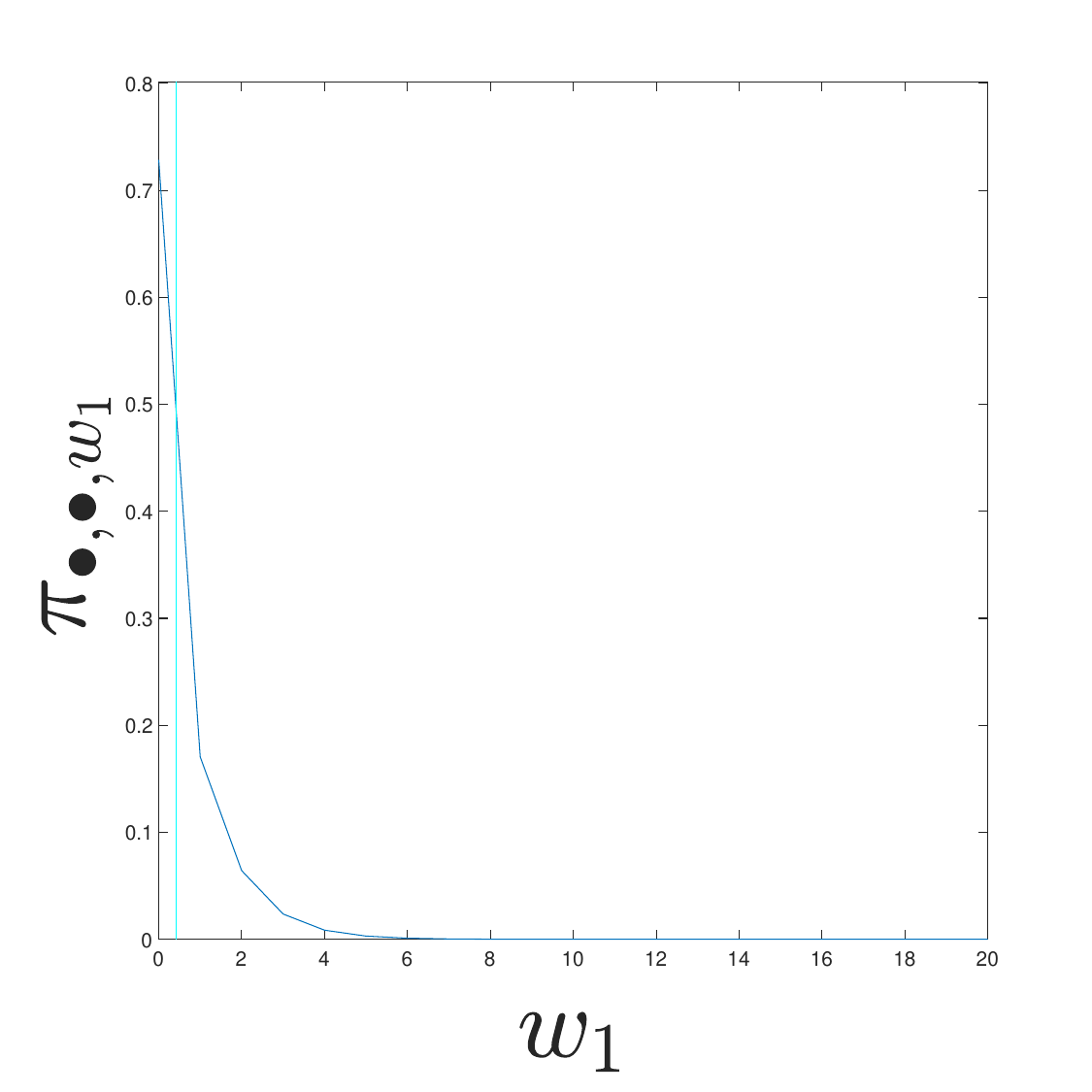}}    
	\caption{Stationary distribution of the QBD model assuming $N=100$, $B=80$ and parameters described in Table~\ref{tab:QBDparameters}. We assume that the patients have priority ($r_1=0.6$, $r_2=0.4$).}
	\label{stationaryQBD_r1_06_r2_04}
\end{figure}

\begin{figure}[H]
	\centering
	{\includegraphics[width=5.1 cm,height=5.1 cm]{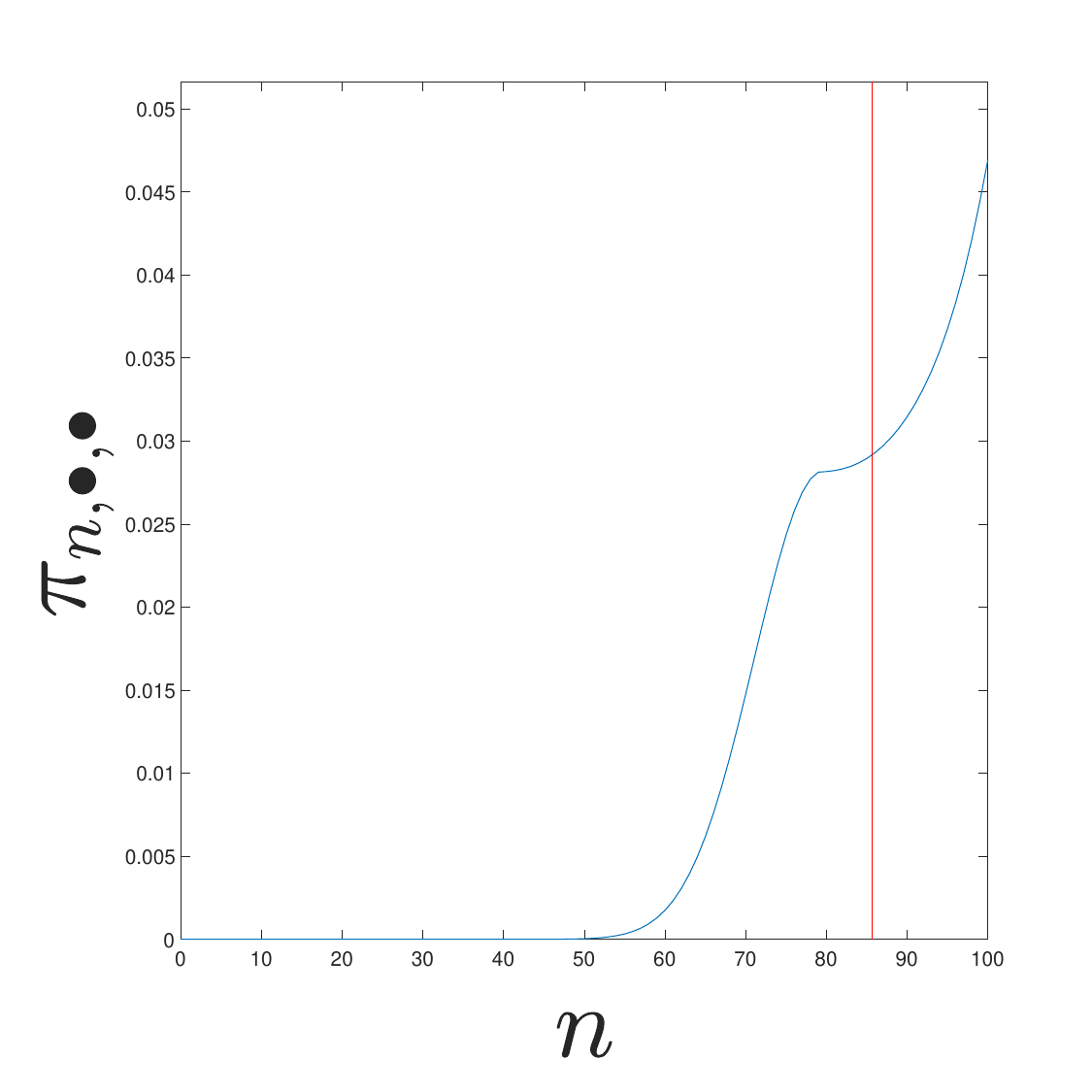}}
    \quad
	{\includegraphics[width=5.1 cm,height=5.1 cm]{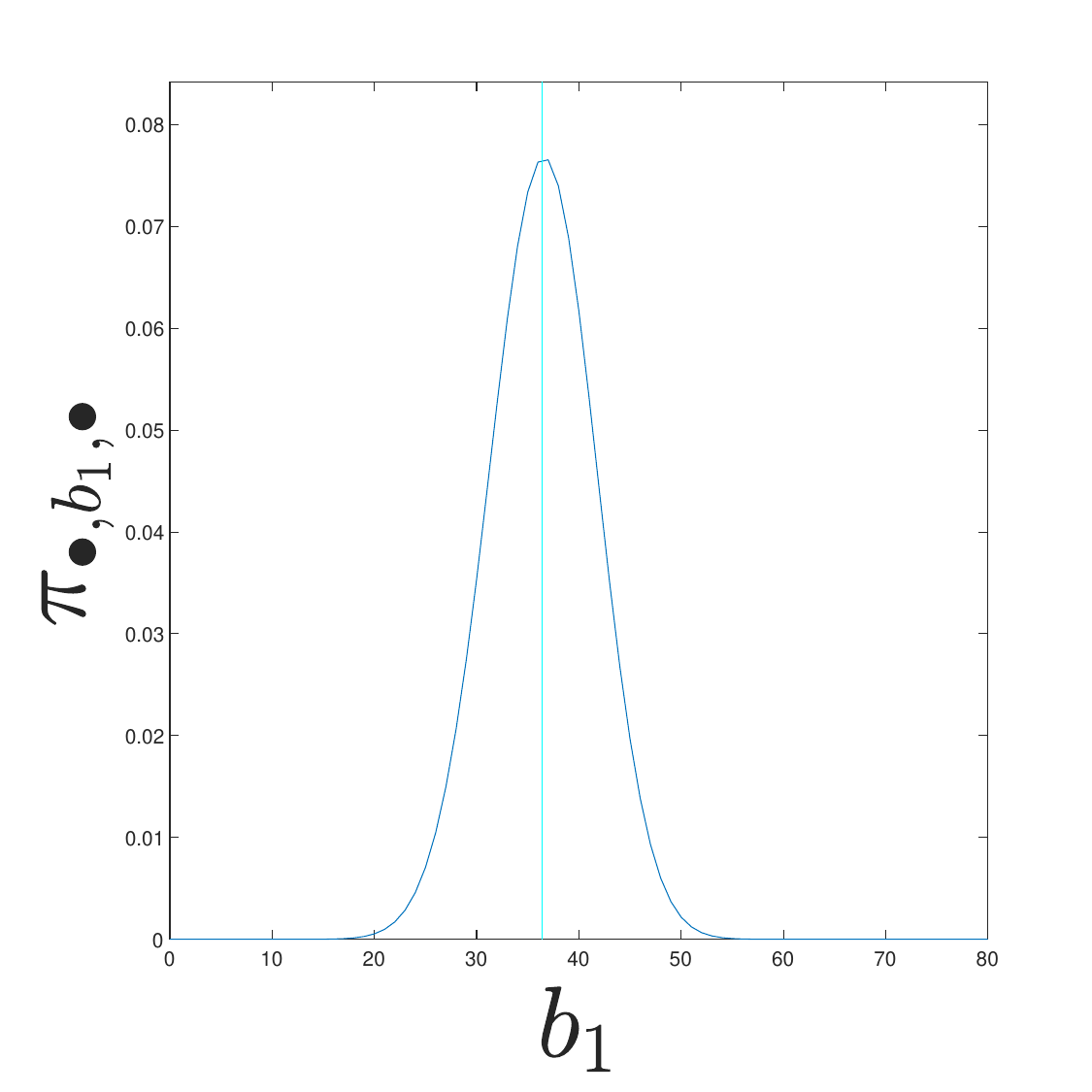}}    
    \quad
	{\includegraphics[width=5.1 cm,height=5.1 cm]{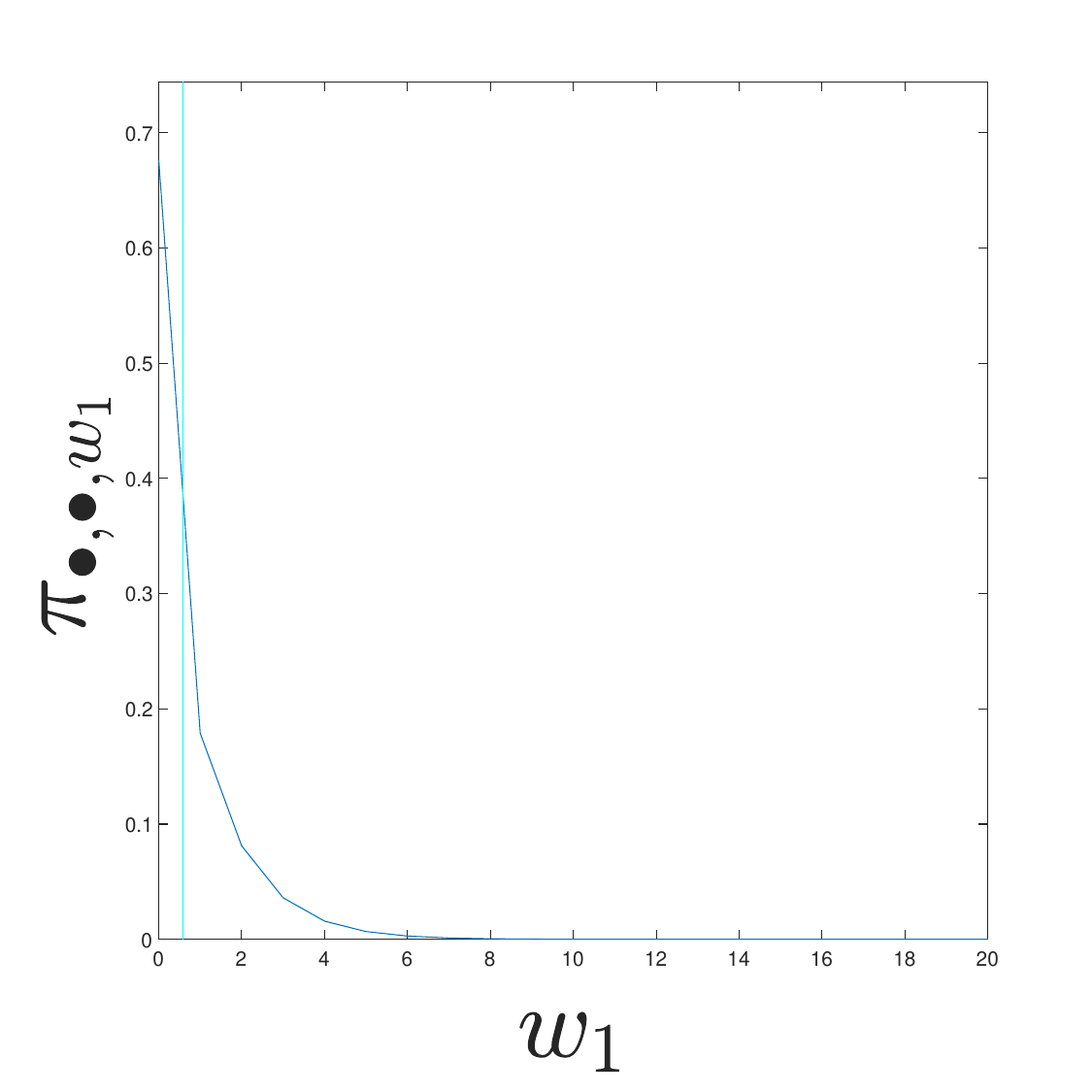}}    
	\caption{Stationary distribution of the QBD model assuming $N=100$, $B=80$ and parameters described in Table~\ref{tab:QBDparameters}. We assume that the patients have priority ($r_1=0.5$, $r_2=0.5$).}
	\label{stationaryQBD_r1_05_r2_05}
\end{figure}

\begin{figure}[H]
	\centering
	{\includegraphics[width=5.1 cm,height=5.1 cm]{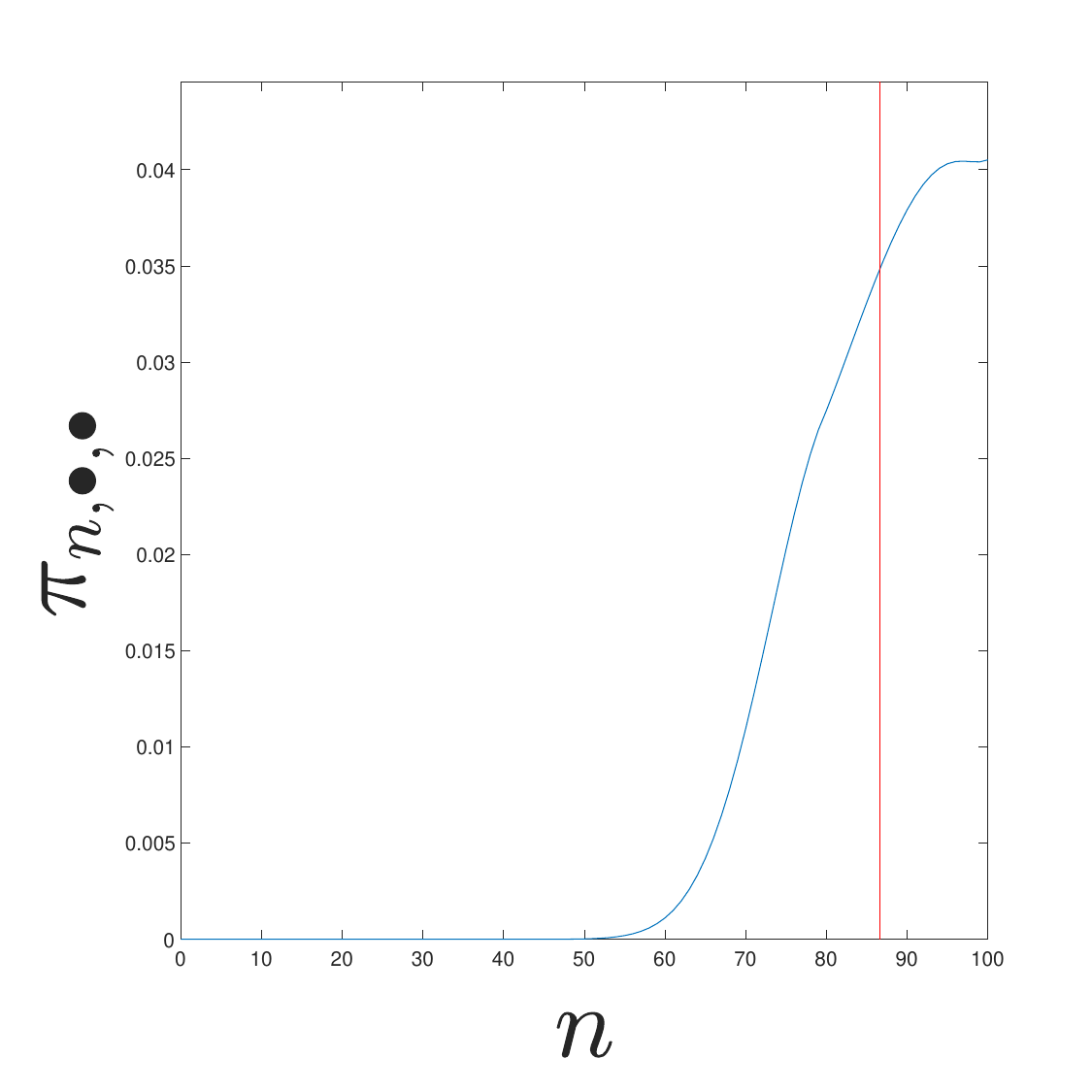}}
    \quad
	{\includegraphics[width=5.1 cm,height=5.1 cm]{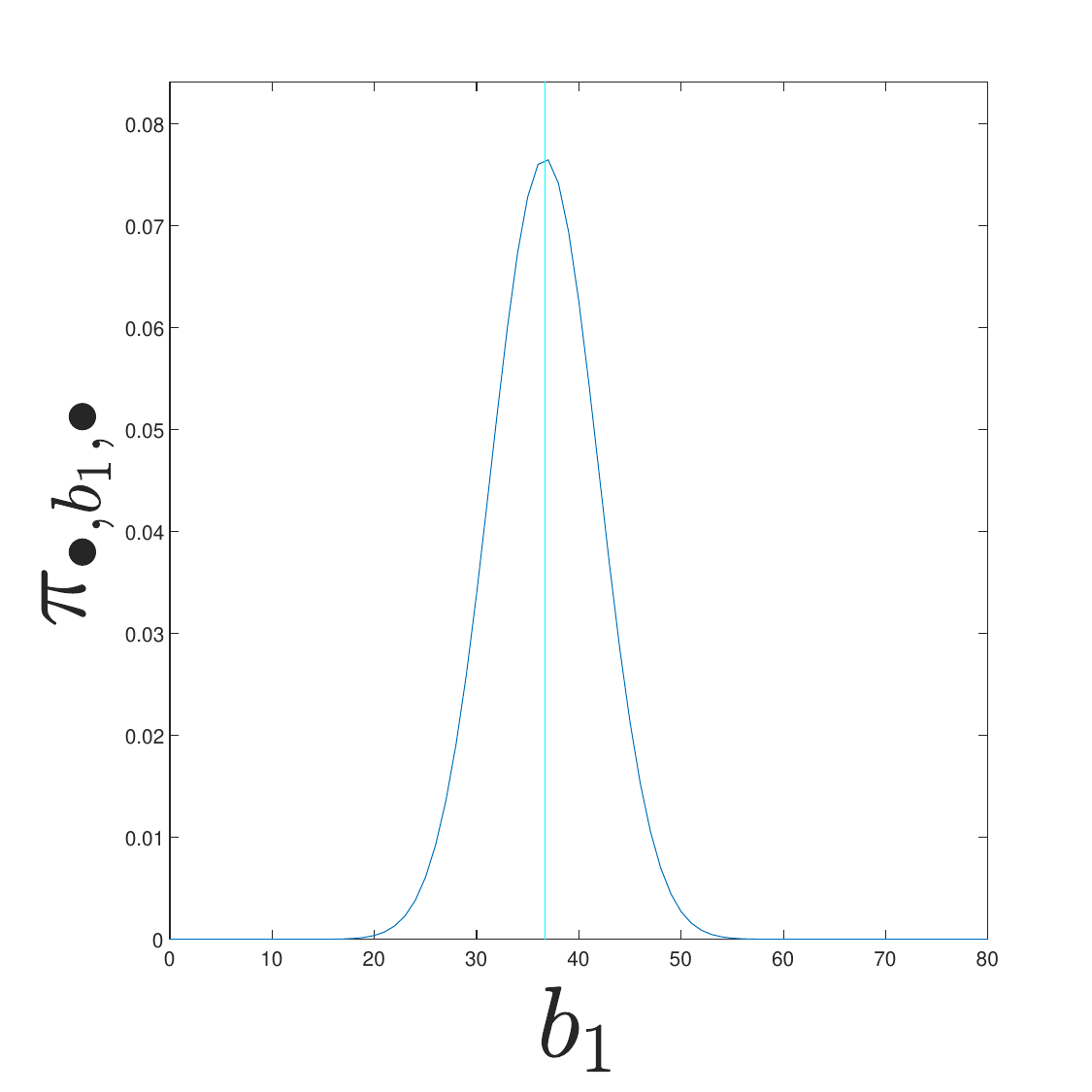}}   \quad
	{\includegraphics[width=5.1 cm,height=5.1 cm]{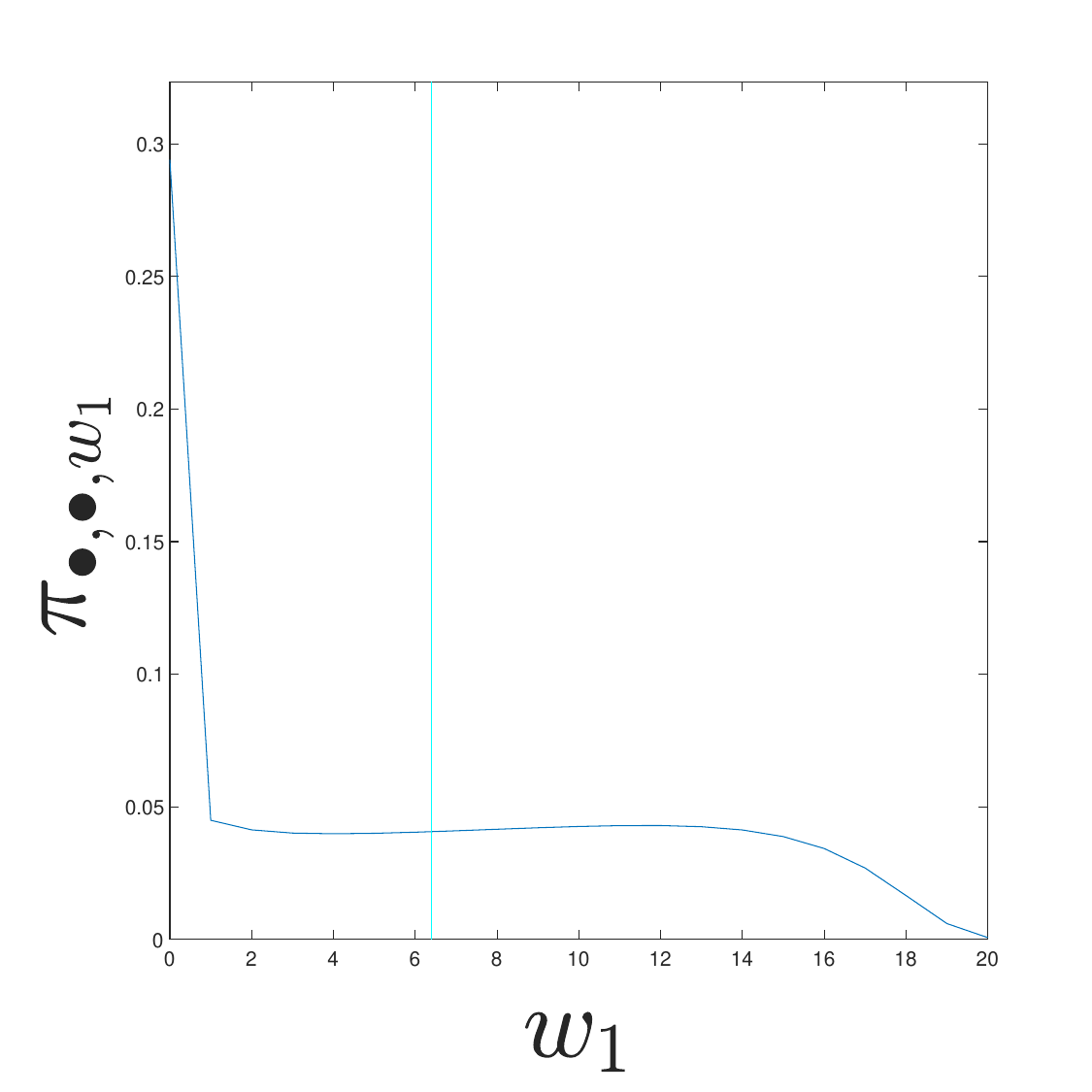}}
    
	\caption{Stationary distribution of the QBD model assuming $N=100$, $B=80$ and parameters described in Table~\ref{tab:QBDparameters}. We assume that Type~$2$ patients have priority ($r_1=0$, $r_2=1$).}
	\label{stationaryQBD_r1_0_r2_1}
\end{figure}

\begin{table}[H]
		\centering
        \setlength{\tabcolsep}{1pt}	
		\begin{tabular} {l@{\hspace{0.5cm}} l}
  \hline
   		\text{Meaning} &	\text{Formula}  \\
\hline
Percentage of time the system is busy &	$\pi_{(N,\bullet,\bullet)} = \sum_{b_1,w_1} \pi_{(N,b_1,w_1)}$ \\

Mean rate of patients to be redirected &	$\lambda_{redirect} = \lambda  \times \pi_{(N,\bullet,\bullet)}$ \\

Mean number of patients in the system &	$L=\sum_n n\pi_{(n,\bullet,\bullet)} $ \\

Mean number of Type~$1$ in the system &	$L_1=\sum_{n,b_1,w_1} (b_1+w_1)\pi_{(n,b_1,w_1)} $ \\

Mean number of Type~$2$ in the system &	$L_2=\sum_{n,b_1,w_1} (b_2+w_2)\pi_{(n,b_1,w_1)} $ \\

Mean number of complex inpatients &	$B_1= \sum_{b_1} b_1\pi_{(\bullet,b_1,\bullet)} $ \\

Mean number of complex patients  &	$W_1= \sum_{w_1} w_1\pi_{(\bullet,\bullet,w_1)} $ \\
in the waiting area  \\

Mean occupancy of the system &	$O_N=100(L/N)$ \\
   \hline
  \end{tabular}
		\caption{
Notations and definitions of the key metrics.}
\label{tab:keymetricsdefinition}
\end{table}

\begin{table}[h!]
\centering
\begin{tabular}{l|ccccccc} 
\toprule
Policy & $\pi_{(N,\bullet,\bullet)}$  & $\lambda_{redirect}$ & $L$ & $L_1$ & $B_1$ & $W_1$ & $O_N$ \\ 
\midrule
$r_1=1$, $r_2=0$    & $4.7353\times10^{-2}$    & $1.1223$  & $85.5883$    & $36.6165$   & $36.4046$  & $0.2119$ & $85.59\%$\\
$r_1=0.8$, $r_2=0.2$    & $4.7276\times10^{-2}$    & $1.1204$  & $85.5960$    & $36.6912$   & $36.4076$  & $0.2836$ & $85.60\%$\\
$r_1=0.7$, $r_2=0.3$    & $4.7212\times10^{-2}$    & $1.1189$  & $85.6027$    & $36.7516$   & $36.4100$  & $0.3416$ & $85.60\%$\\
$r_1=0.6$, $r_2=0.4$    & $4.7113\times10^{-2}$    & $1.1166$  & $85.6138$    & $36.8432$   & $36.4138$  & $0.4294$ & $85.61\%$\\

$r_1=0.5$, $r_2=0.5$    & $4.6942\times10^{-2}$    & $1.1125$  & $85.6342$    & $36.9984$   & $36.4203$   & $0.5781$ & $85.63\%$\\

$r_1=0$, $r_2=1$    & $4.0520\times10^{-2}$    & $0.9603$  & $86.5796$    & $43.0631$   & $36.6657$   & $6.3973$ & $86.58\%$\\
\bottomrule
\end{tabular}
\caption{The values of the key metrics given in Table~\ref{tab:keymetricsdefinition} corresponding to different priority policies as defined by the pair ($r_1,r_2$). Note that $L_2=L-L_1$, $B_2=B-B_1$, $W_2=L-B-W_1$.}
\label{tab:keymetricestimations}
\end{table}

\section{Conditional waiting times ($\beta_{2\to1} \geq 0$, $r_1=1$, $r_2=0$)}
\label{Conditional_waiting_times_beta_0_r1_1_r2_0}

Here, we present the analysis of the conditional waiting times for both Type~$1$ and Type~$2$ patients in the waiting area. We  evaluate their waiting time distributions under different modeling approaches. Specifically, we compare the results obtained from the QBD model and the CTMC model, under the assumptions $\beta_{2\to1} \geq 0$, $r_1=1$, $r_2=0$, to assess the accuracy of the CTMC-based approximation. Following this comparison, we further compute the conditional waiting time distributions for both patient types under a fixed type-change rate of $\beta_{2 \to 1} = 3$ and explore the impact of varying priority policies corresponding to the choices of $(r_1,r_2)$.

\subsection{Conditional waiting times of Type~$1$ patients ($\beta_{2\to1}=0$, $r_1=1$, $r_2=0$)}
\label{conditional_waittime_T1_beta0_r1_1_r2_0}

We compute the waiting time distributions for Type~$1$ patients, assuming $\beta_{2\to1}=0$, $r_1=1$, $r_2=0$ and a system capacity of $N = 100$ and $B = 80$. We note that the waiting times of Type~$1$ patients when $\beta_{2\to 1}>0$ will be the same, since Type~$2$ to Type~$1$ changes do not affect the waiting time of a Type~$1$ patient already in Queue~$1$. We compute the conditional waiting time distributions using the CTMC model described in section~\ref{sec:CTMCoriginal} and the QBD model defined in Section~\ref{sec:WT_Type1nochange}.

The conditional probability densities of waiting times of $5^{th}$ (top left), $10^{th}$ (top right), $15^{th}$ (bottom left), and $20^{th}$ (bottom right) Type~$1$ patients in the waiting area, assuming that there are $b_1=30$ and $b_2=50$ Type~$1$ and Type~$2$ patients in the beds, respectively, and $w_2=0$ Type~$2$ patients in the waiting area, at time $t=0$, are presented in Figure~\ref{wait_time_QBD_CTMC_51220_30_5_0_1}. We note that as the position of the Type~$1$ patient in the queue increases, the mean waiting time of such patients increases. This reflects the natural behavior of a queuing system where patients further back in the queue experience longer delays.

Furthermore, as the position of the Type~$1$ patient in the queue increases, the corresponding probability density function shifts to the right and becomes more symmetrical and bell-shaped. This increasing symmetry occurs because patients later in the queue must wait through many service completions, and the aggregation of these independent events naturally smooths out variability. For patients earlier in the queue (e.g., the $5^{th}$), the distribution is slightly right-skewed, indicating a higher probability of shorter waiting times but with a long tail representing occasional long waiting times. This skewness arises because, when only a few patients are ahead, the tagged patient's waiting time may strongly be affected by the possibility of a single long service time at the front of the queue. In contrast, for later patients (e.g., the $20^{th}$), the distribution becomes less skewed and more sharply peaked around the mean, reflecting a higher concentration of waiting times near the average. These evolving patterns in the distribution shape can provide valuable insights for resource allocation and priority policies in managing patient queues.

\begin{figure}[!htbp]
    \centering

    \mbox{\includegraphics[scale=0.4]{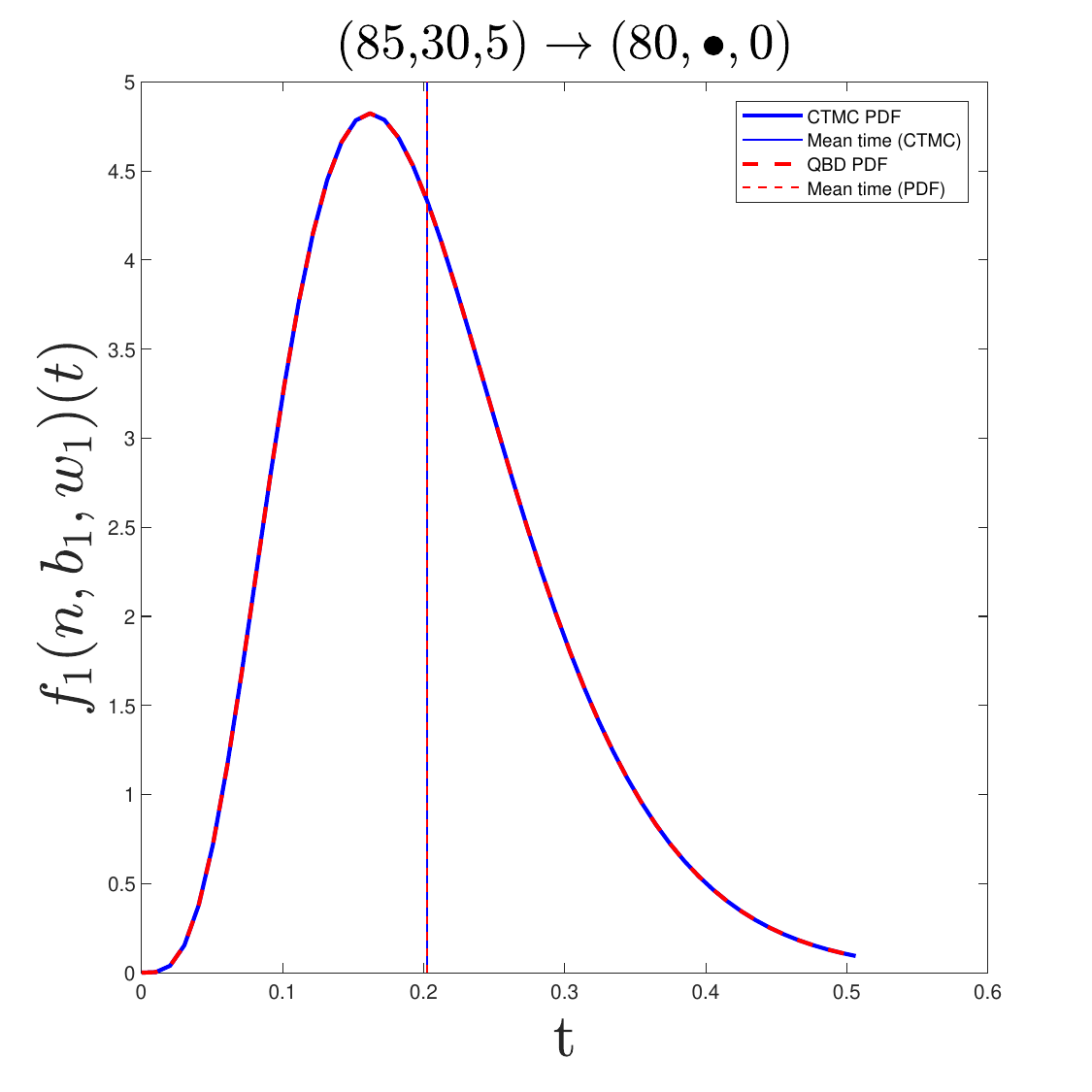}}
    \quad
    \mbox{\includegraphics[scale=0.4]{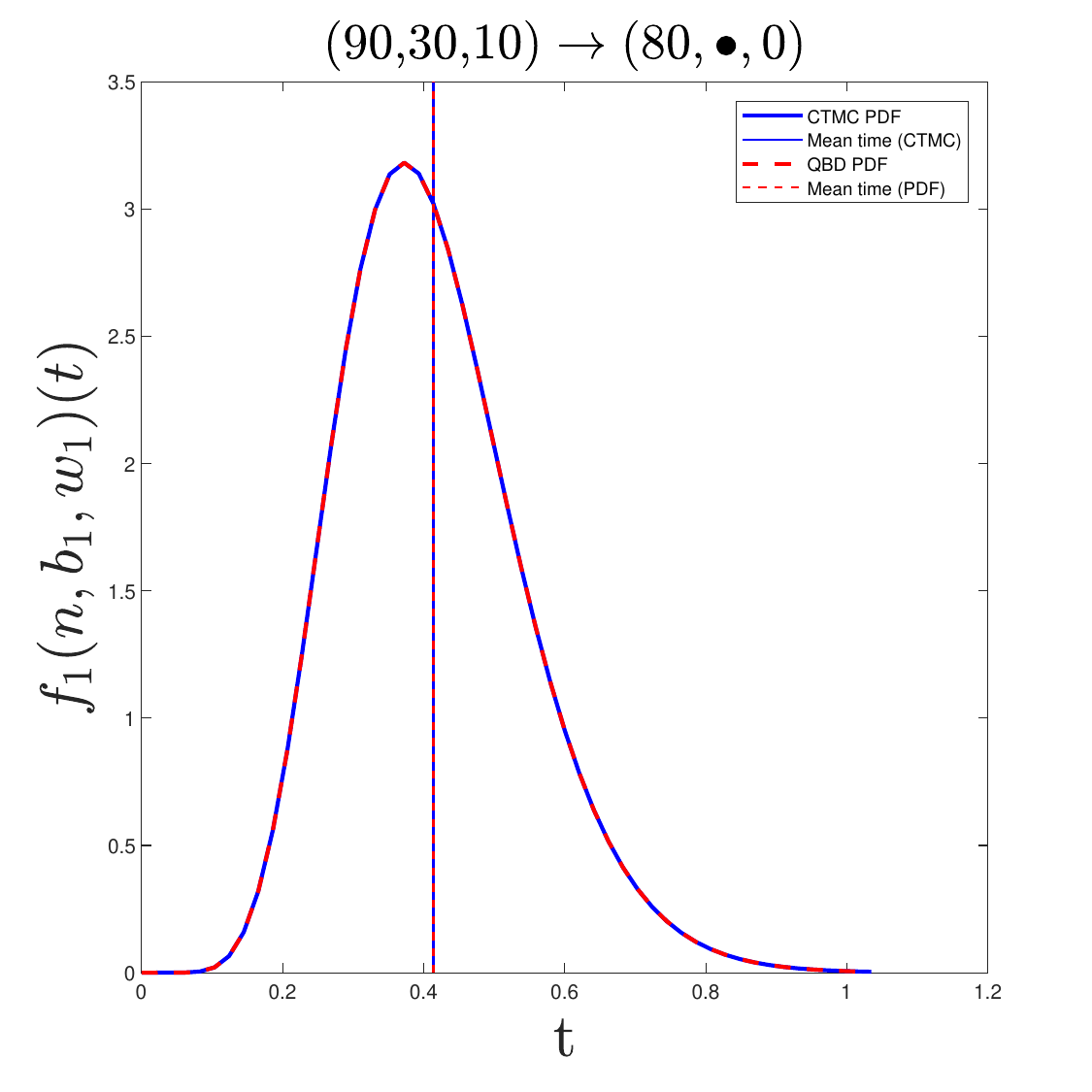}}

    \mbox{\includegraphics[scale=0.4]{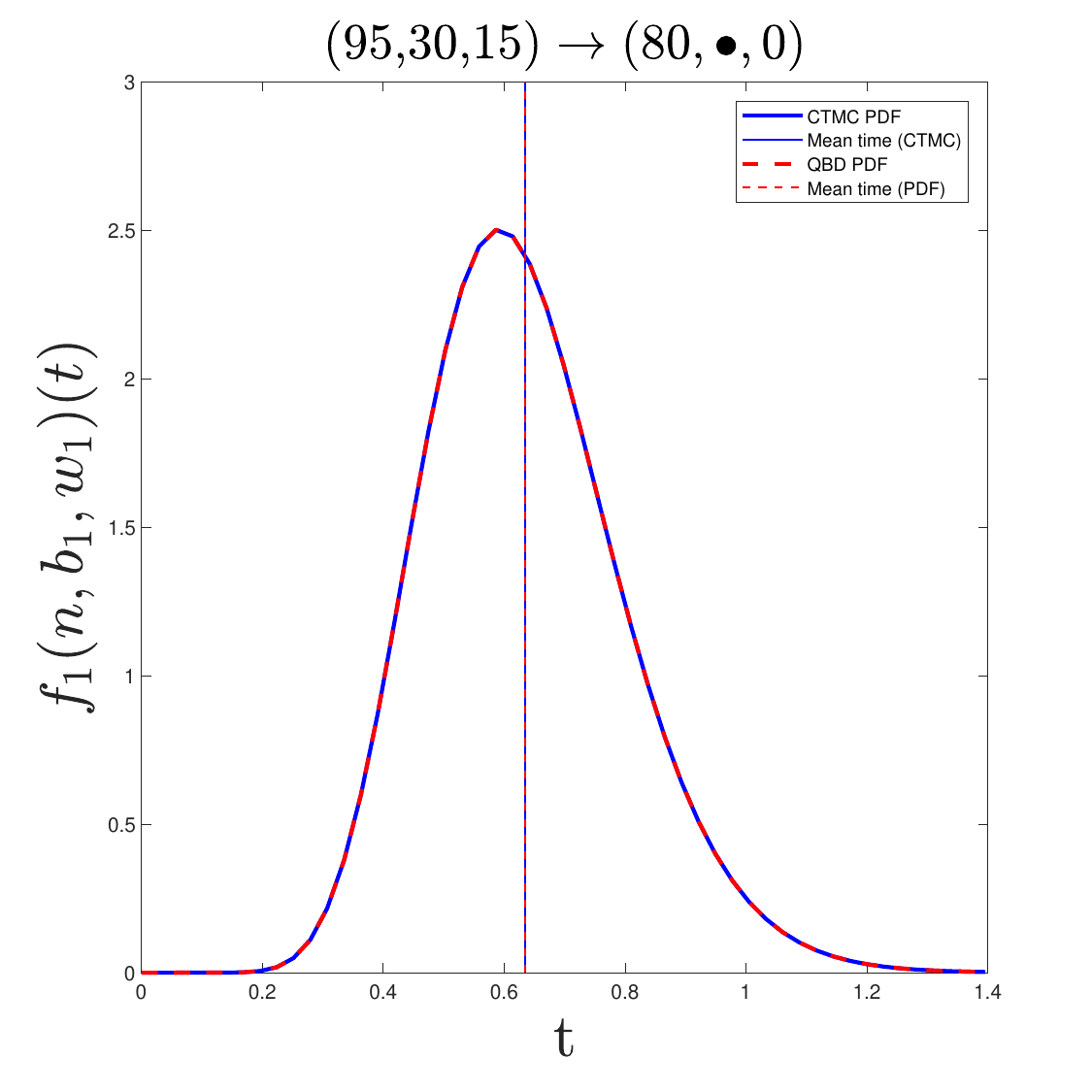}}
    \quad
    \mbox{\includegraphics[scale=0.4]{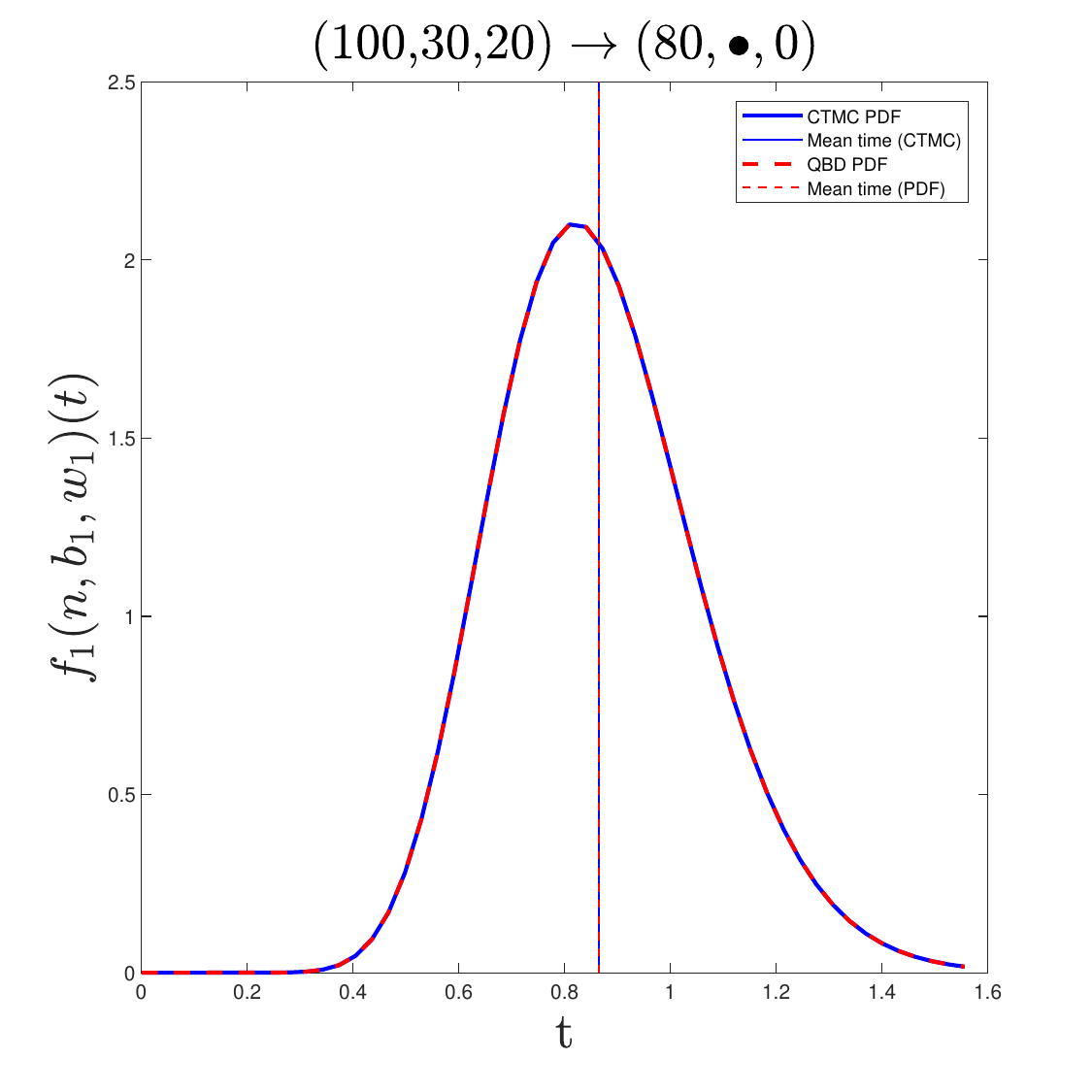}}

    \caption{Probability densities $f_1(n,b_1,w_1)(t)$, with $n=B+w_1$, $w_1=5,10,15,20$, of the waiting times of Type~$1$ patient from the moment they arrive in the system, and so they are $5^{th}$ (top left figure), $10^{th}$ (top right figure), $15^{th}$ (bottom left figure), and $20^{th}$ tagged patient at the end of Queue~$1$ in the waiting area, respectively, given $b_1=30$, $b_2=50$, $w_2=0$, and assuming parameters $\lambda_1=0$, $\lambda_2=0$, $\mu_1=0.1517$, $\mu_2=0.4113$ (see Table~\ref{tab:QBDparameters}), and $\beta_{2\to1}=0$, in a system of capacity $N=100$, $B=80$. Mean waiting times $ \mathbb{E}(T_1(n, b_1, w_1))$ for these Type~$1$ patients, evaluated using both the CTMC and the QBD model, are $0.2025$, $0.4139$, $0.6344$, and $0.8644$ days, respectively.}
	\label{wait_time_QBD_CTMC_51220_30_5_0_1}
\end{figure}

\subsection{Conditional waiting times of Type~$2$ patients ($\beta_{2\to1}=0$, $r_1=1$, $r_2=0$)} \label{conditional_waittime_T2_beta0_r1_1_r2_0}

Next, we give the conditional waiting time distributions for Type~$2$ patients, assuming $\beta_{2\to1}=0$, $r_1=1$, $r_2=0$ and a system capacity of $N = 100$ and $B = 80$. The probability densities of the waiting times of $5^{th}$ (top left), $10^{th}$ (top right), $15^{th}$ (bottom left), and $20^{th}$ (bottom right) Type~$2$ patients in the waiting area, assuming that there are $b_1=30$ and $b_2=50$ Type~$1$ and Type~$2$ patients in the beds, respectively, and $w_1=2$ Type~$1$ patients in the waiting area, at time $t=0$, are given in Figure~\ref{Conditional_QBD_CTMC_T2_beta_0}, under both the QBD and the CTMC models. We note that as the position of the Type~$2$ patient in the queue increases, the mean waiting time for the Type~$2$ patient increases, reflecting the natural behavior of a queuing system.

As the position of the Type~$2$ patient increases in the queue, the probability density function of their waiting time gradually shifts rightward and adopts a more symmetrical, bell-shaped form. This suggests a stabilisation in waiting time variability for patients down the queue, resulting in more consistent and predictable delays. For patients positioned earlier in the queue (e.g., the $2^{nd}$), the distribution exhibits slight right skewness, indicating a higher probability of shorter waits, yet with a long tail that captures occasional delays. This heavier tail arises because, when only a few patients are ahead, the tagged patient’s wait is strongly influenced by the potentially long service time of a single individual at the front of the queue. For those later in the queue (e.g., the $20^{th}$), the distribution becomes more concentrated around the mean, with reduced skewness and a sharper peak, reflecting a tighter clustering of waiting times near the average. This tightening occurs because patients later in the queue must wait through many independent service completions, and the accumulation of these independent events smooths out variability, producing a more regular and symmetric distribution.

\begin{figure}[!htbp]
    \centering

    \mbox{\includegraphics[scale=0.4]{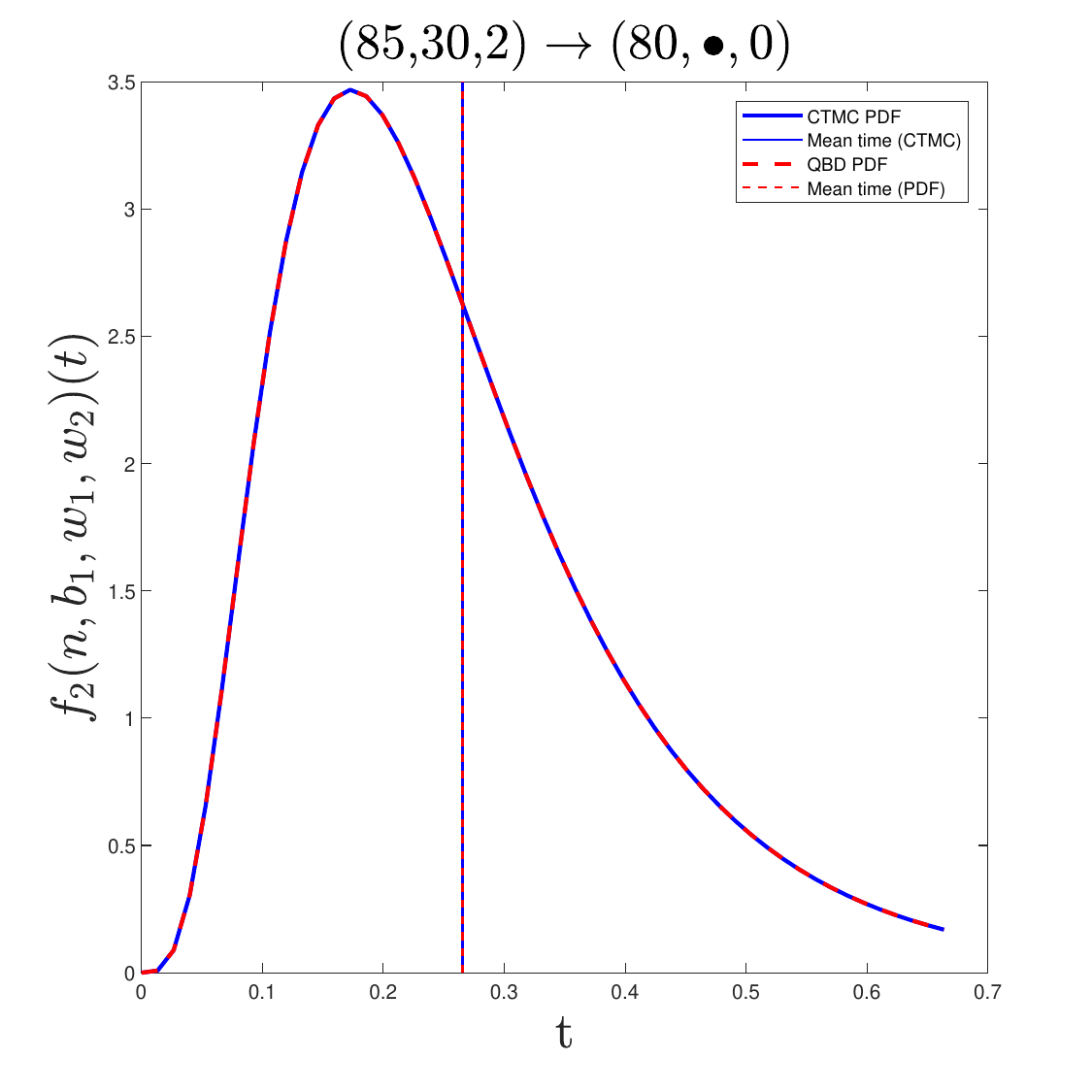}}
    \quad
    \mbox{\includegraphics[scale=0.4]{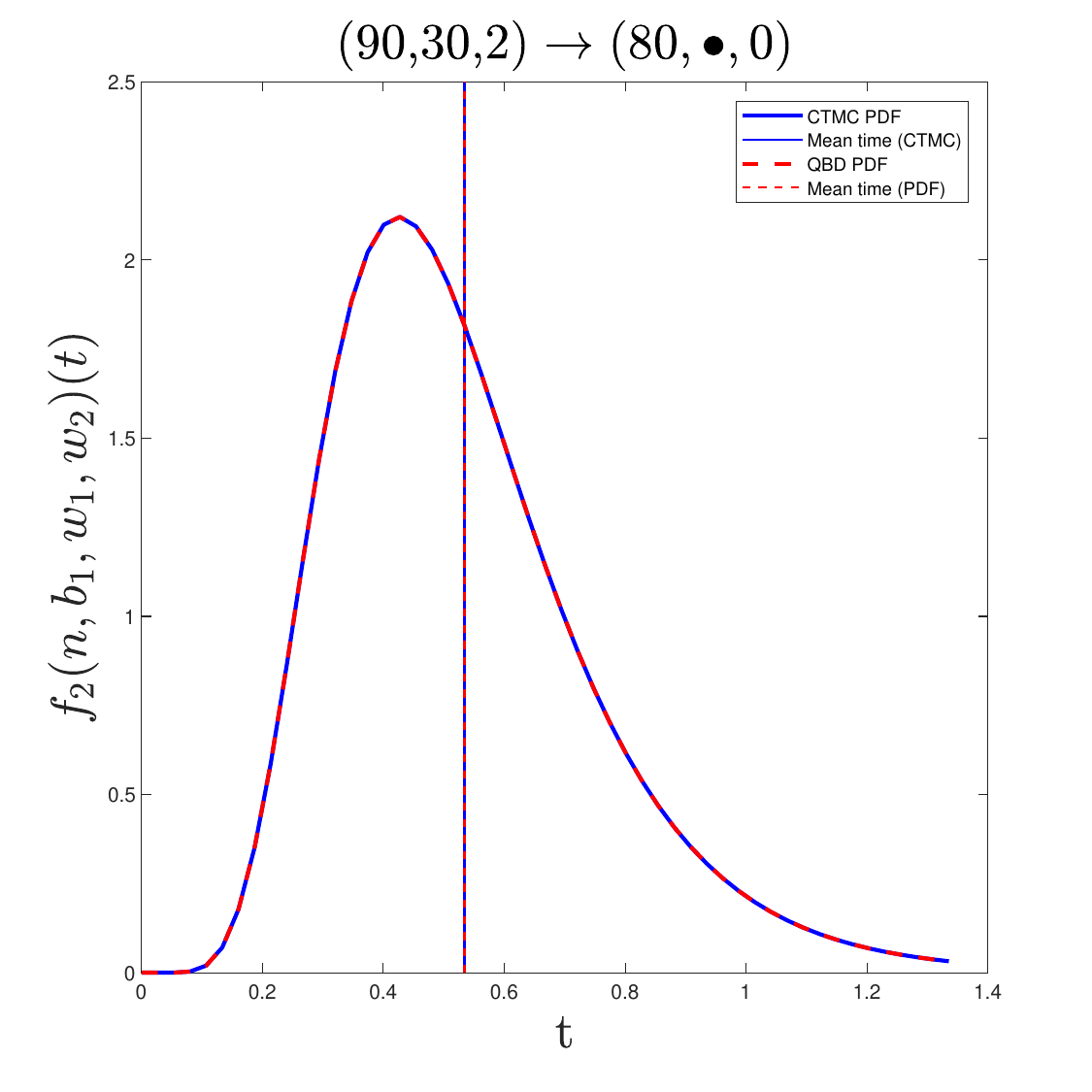}}

    \mbox{\includegraphics[scale=0.4]{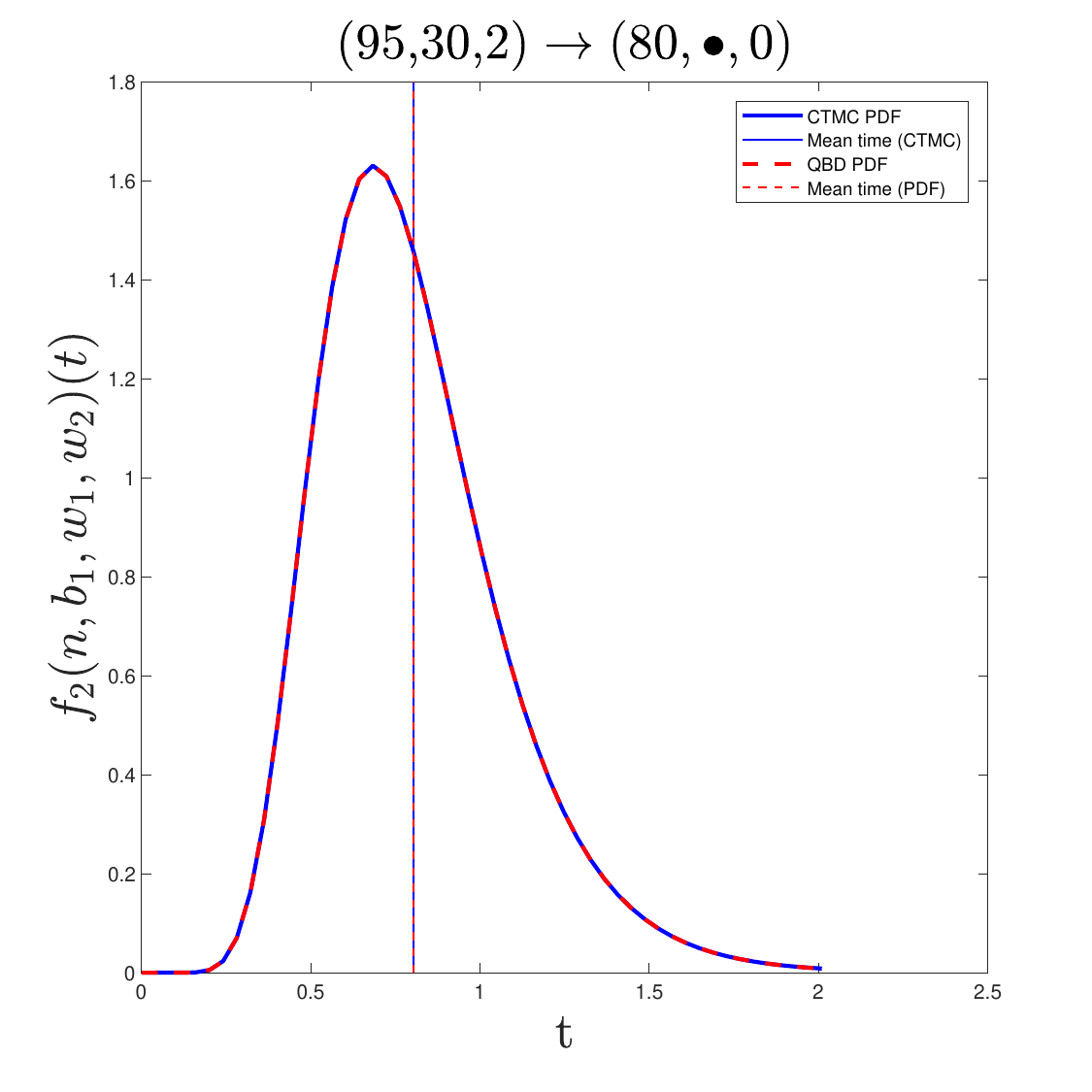}}
    \quad
    \mbox{\includegraphics[scale=0.4]{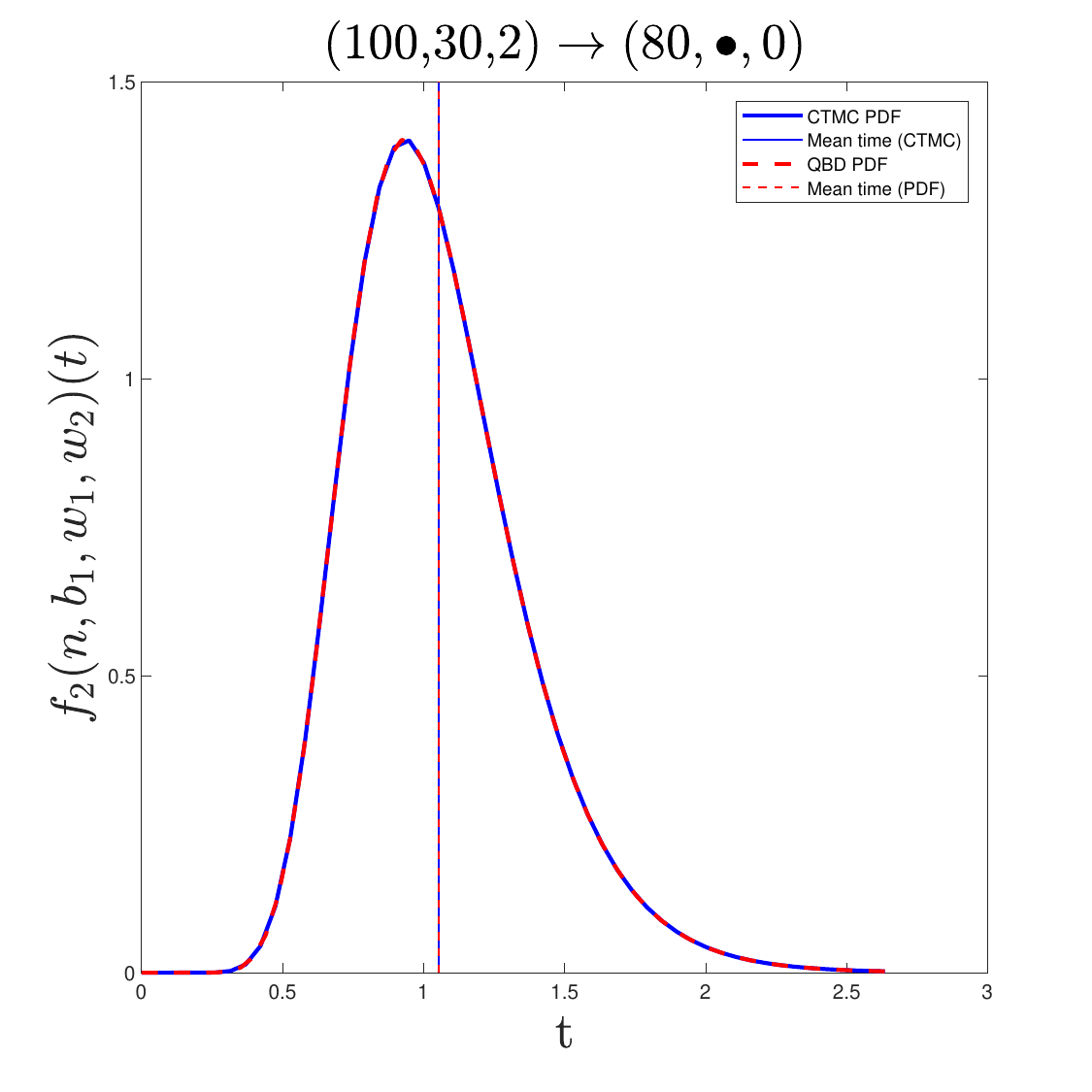}}

    \caption{Probability densities $f_2(n,b_1,w_1,w_2)(t)$, with $n=B+w_1+w_2$, $w_1=2$, $w_2=3,8,13,18$, of the waiting times of Type~$2$ patient from the moment they arrive in the system, and so they are $5^{th}$ (top left figure), $10^{th}$ (top right figure), $15^{th}$ (bottom left figure), and $20^{th}$ tagged patient in the waiting area, respectively, given $b_1=30$, $b_2=50$, and assuming parameters $\lambda_1=5.7961$, $\lambda_2=0$, $\mu_1=0.1517$, $\mu_2=0.4113$ (see Table~\ref{tab:QBDparameters}), and $\beta_{2\to1}=0$, in a system of capacity $N=100$, $B=80$. Mean waiting times $ \mathbb{E}(T_1(n, b_1, w_1, w_2))$ for these Type~$2$ patients, evaluated using both CTMC and QBD model, are $0.2655$, $0.5342$, $0.8038$, and $1.0541$ days, respectively.}
	\label{Conditional_QBD_CTMC_T2_beta_0}
\end{figure}

\subsection{Comparative analysis of QBD and CTMC models}

In Sections~\ref{conditional_waittime_T1_beta0_r1_1_r2_0}~$\&$~\ref{conditional_waittime_T2_beta0_r1_1_r2_0}, we presented the conditional waiting time distributions for Type~$1$ and Type~$2$ patients, assuming $\beta_{2\to1}=0$, $r_1=1$, $r_2=0$, under both the QBD model and the CTMC model. The results, illustrated in Figures~\ref{wait_time_QBD_CTMC_51220_30_5_0_1}~$\&$~\ref{Conditional_QBD_CTMC_T2_beta_0}, demonstrate a remarkable agreement between the two modeling approaches. This shows that the CTMC model accurately approximates the QBD models. Also, the mean waiting times computed from both models are similar for each patient type and their position in the queue. These results validate the CTMC model as a reliable alternative to the QBD model for estimating conditional waiting times.

\subsection{Conditional waiting times of Type~$2$ patients ($\beta_{2\to1}=3$, $r_1=1$, $r_2=0$)}

We now compute the conditional waiting time distributions for Type~$2$ patients using the CTMC model, assuming $\beta_{2 \to 1}=3$, $r_1=1$, $r_2 = 0$, $N=100$, and $B=80$. We note that under these conditions, the waiting time distributions for Type~$1$ patients remain identical to those presented in Figure~\ref{wait_time_QBD_CTMC_51220_30_5_0_1}. Figure~\ref{Conditional_QBD_CTMC_T2_beta_3} shows the probability density functions of the waiting times for the $5^{th}$ (top left), $10^{th}$ (top right), $15^{th}$ (bottom left), and $20^{th}$ (bottom right) Type~$2$ patients in the waiting area. These results are obtained under the assumption that, at time $t = 0$, there are $b_1 = 30$ and $b_2 = 50$ Type~$1$ and Type~$2$ patients occupying beds, respectively, and $w_1 = 2$ Type~$1$ patients in the waiting area.

Further, in comparison with Figure~\ref{Conditional_QBD_CTMC_T2_beta_0}, the mean waiting times for Type~$2$ patients are slightly higher when they are closer to the head of Queue~$2$, which makes sense as Type~$2$ patients behind them in the queue are likely to change from Type~$2$ to Type~$1$ and move ahead of them in the queue. On the other hand, the mean waiting times for Type~$2$ patients are reduced when they are closer to the end of Queue~$2$. This is because these patients would be likely to change from Type~$2$ to Type~$1$, and effectively bypass those ahead of them in Queue~$2$, gaining priority access to a bed, before they are allocated to a bed.

\begin{figure}[!htbp]
    \centering

    \mbox{\includegraphics[scale=0.4]{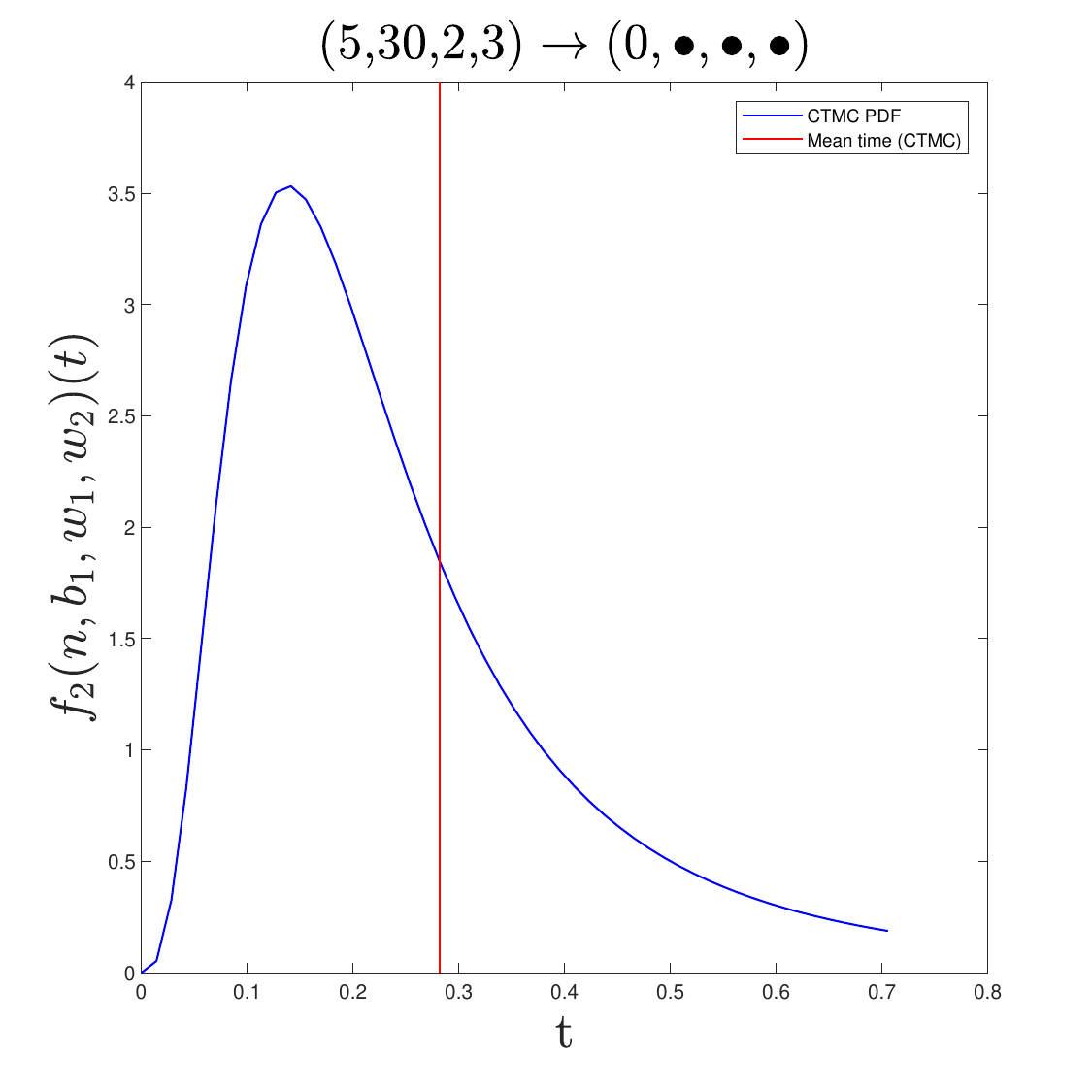}}
    \quad
    \mbox{\includegraphics[scale=0.4]{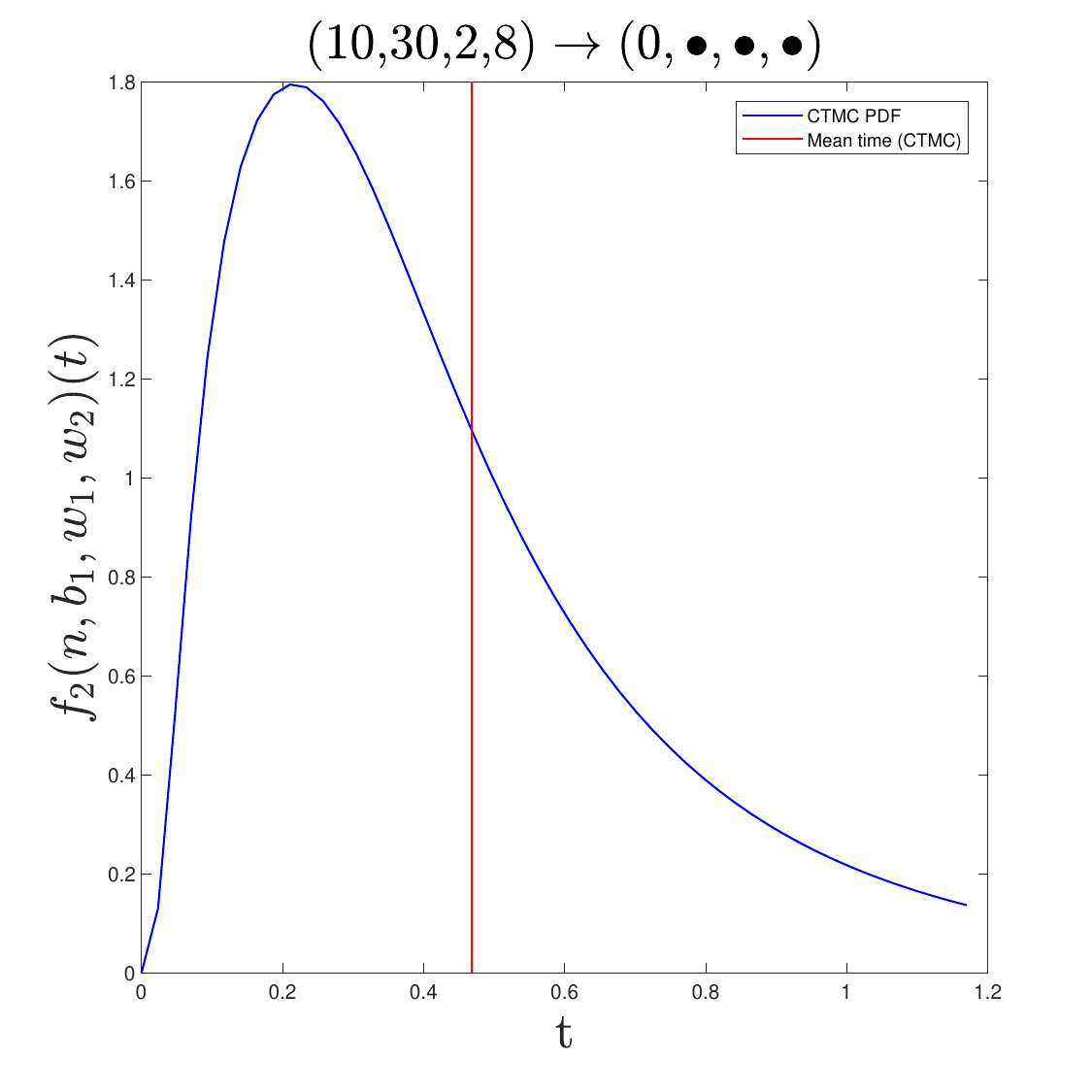}}

    \mbox{\includegraphics[scale=0.4]{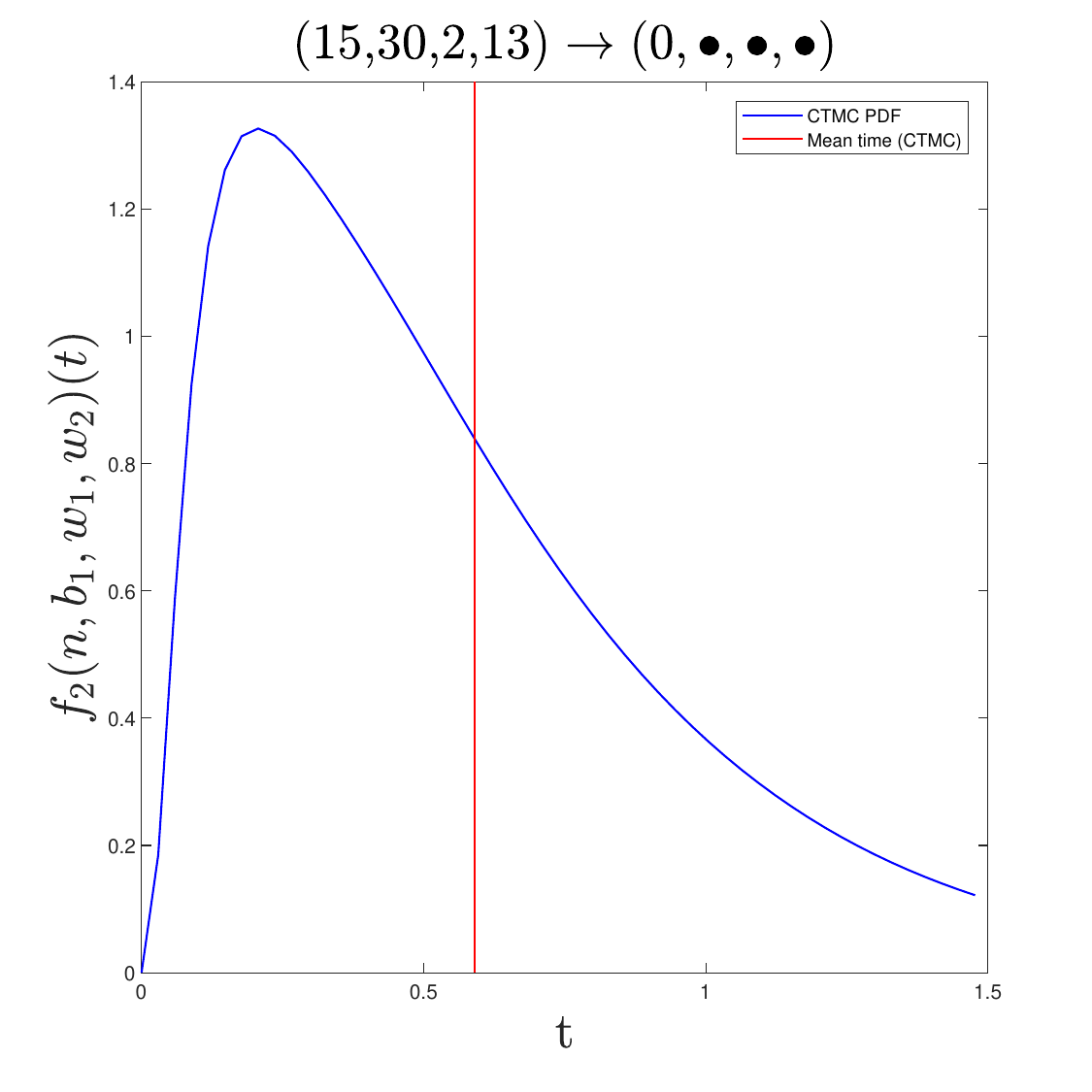}}
    \quad
    \mbox{\includegraphics[scale=0.4]{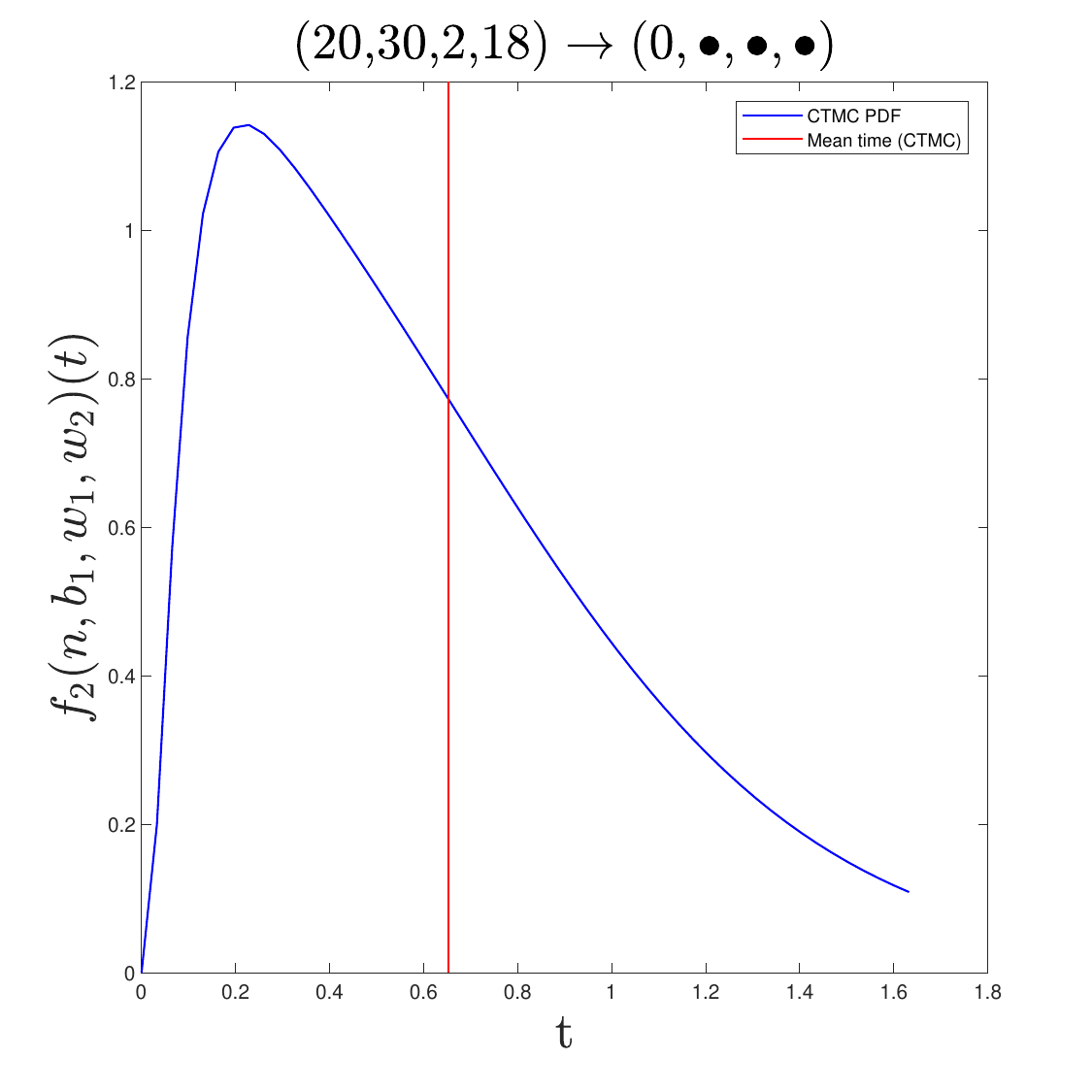}}

    \caption{Probability densities $f_2(n,b_1,w_1,w_2)(t)$, with $n=B+w_1+w_2$, $w_1=2$, $w_2=3,8,13,18$, of the waiting times of Type~$2$ patient from the moment they arrive in the system, and so they are $5^{th}$ (top left figure), $10^{th}$ (top right figure), $15^{th}$ (bottom left figure), and $20^{th}$ tagged patient in the waiting area, respectively, given $b_1=30$, $b_2=50$, and assuming parameters $\lambda_1=5.7961$, $\lambda_2=17.9039$, $\mu_1=0.1517$, $\mu_2=0.4113$ (see Table~\ref{tab:QBDparameters}), and $\beta_{2\to1}=3$, in a system of capacity $N=100$, $B=80$. Mean waiting times $ \mathbb{E}(T_2(n, b_1, w_1, w_2))$ for these Type~$2$ patients, evaluated using both CTMC and QBD model, are $0.2822$, $0.4681$, $0.5909$, and $0.6530$ days, respectively.}
	\label{Conditional_QBD_CTMC_T2_beta_3}
\end{figure}

\subsection{Conditional waiting times of Type~$1$ and Type~$2$ patients in a congested system ($\beta_{2\to1}=3$, $r_1, r_2 \in [0,1]$)}

We evaluate the conditional waiting time distributions of a Type~$1$ patient who is at the $20^{th}$ position in the waiting area with no Type~$2$ patients present, and for a Type~$2$ patient who is at the $20^{th}$ position in the waiting area with two Type~$1$ patients in front of them. We evaluate the distributions under priority policies $(r_1, r_2) = (0.8, 0.2)$, $(0.6, 0.4)$, and $(0.5, 0.5)$ using the CTMC with $\beta_{2\to1}=3$, and analyse the impact of the priority policies on the waiting times. The analysis we present here is useful for hospital administrators to understand how policy decisions such as adjusting priority weights can influence patient flow and resource allocation.

The outputs are shown in Figures~\ref{conditional_densities_T1_diff_r1_r2} and~\ref{conditional_densities_T2_diff_r1_r2}. We recall that the Figure~\ref{wait_time_QBD_CTMC_51220_30_5_0_1} (bottom right) shows the waiting time distribution of the $20^{th}$ Type~$1$ patient under the priority policy $(r_1, r_2) = (1, 0)$. We note that decreasing the priority for Type~$1$ patients increases the mean waiting times, which is consistent with the system's natural behaviour. The pdfs remain symmetrical because of the aggregation of many independent service completions before the $20^{th}$ Type~$1$ patient goes to a bed. The observed flattening of the density curve as $r_1$ decreases reflects greater competition for beds (more instances in which Queue~$2$ is favoured).

Moreover, as the priority of Type~$1$ patients decreases, the mean waiting time of the tagged Type~$2$ patient ($20^{th}$) increases. This is because the mean waiting time of the tagged Type~$2$ patient is greater than the mean time ($8$ hours) until a Type~$2$ patient changes to Type~$1$, and so it is likely that the tagged Type~$2$ patient changes to Type~$1$ before getting to a bed. As a result, the priority of the tagged Type~$2$ patient ultimately decreases with $r_1$ and so the mean waiting time increases.

We also note that the mean waiting time of the tagged $20^{th}$ Type~$2$ patient remains below that of the $20^{th}$ Type~$1$ patient (Fig.~\ref{conditional_densities_T1_diff_r1_r2}). This is because a Type~$2$ patient always has the possibility of changing type and gaining priority, providing a potential shortcut to service. The longer right tails observed for higher $r_2$ reflect situations in which many patients behind the tagged individual change to Type~$1$ early, repeatedly overtaking the tagged patient and prolonging their wait.

\begin{figure}[!htbp]
	\centering
	{\includegraphics[width=5.1 cm,height=5.1 cm]{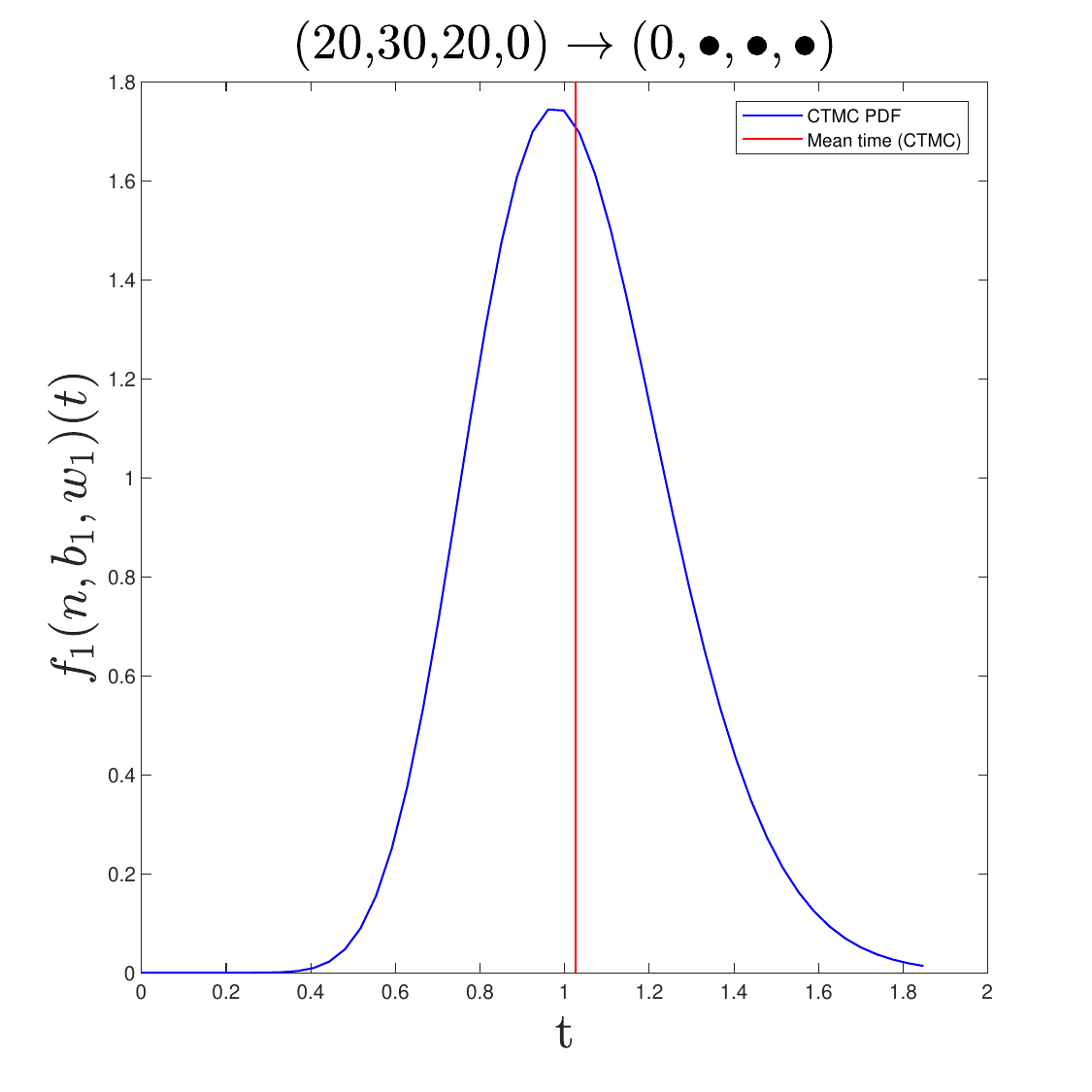}}
    \quad
	{\includegraphics[width=5.1 cm,height=5.1 cm]{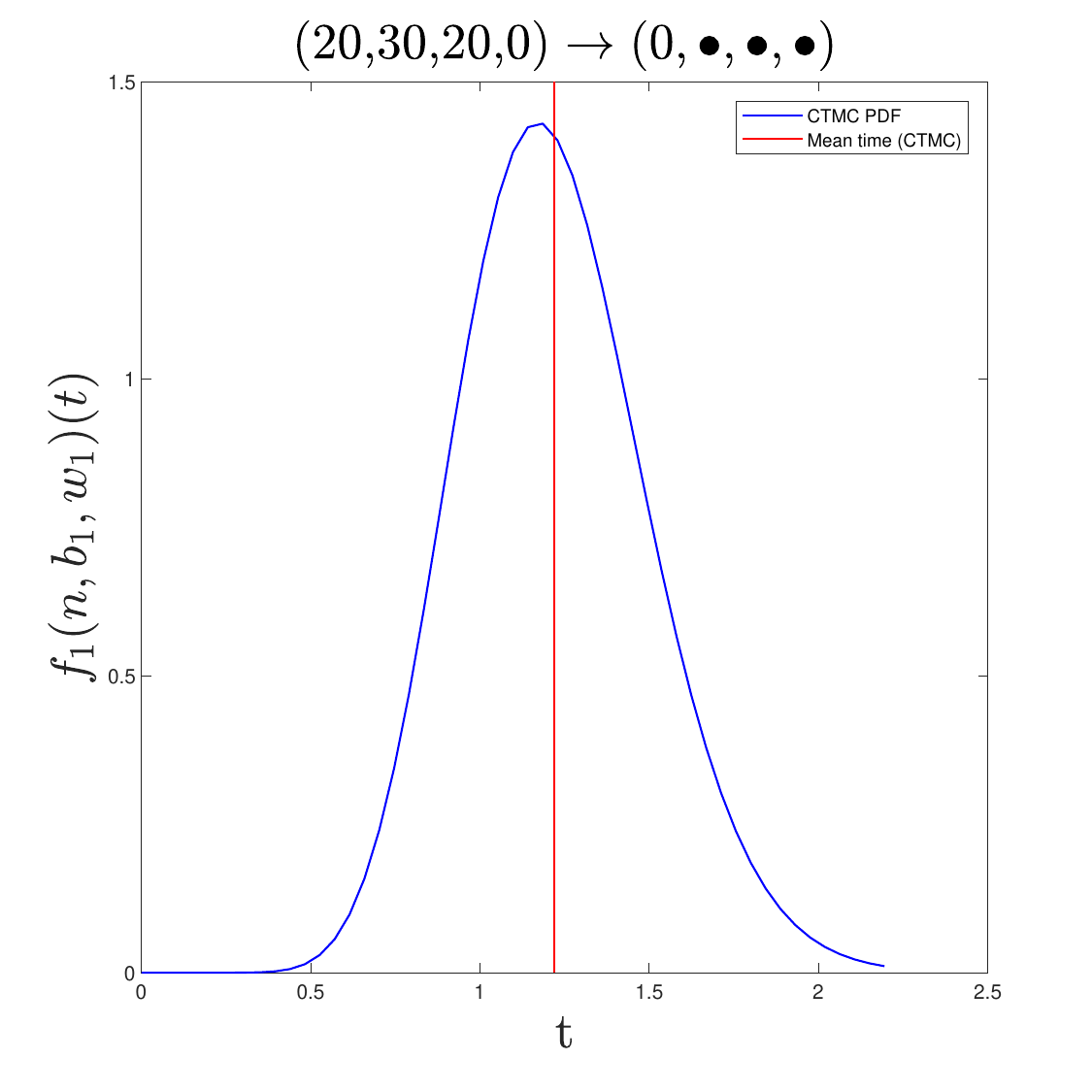}}    
    \quad
	{\includegraphics[width=5.1 cm,height=5.1 cm]{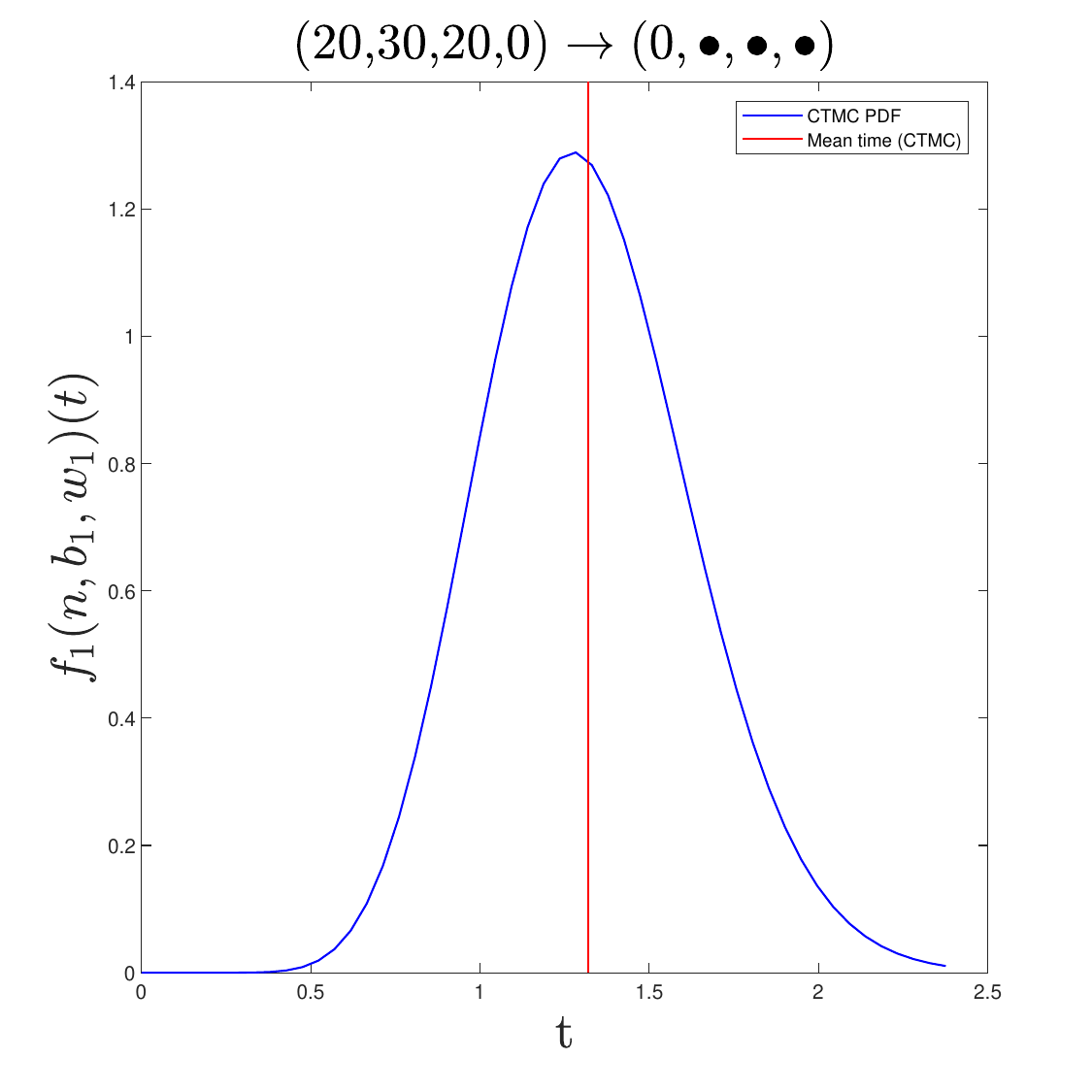}}    
	\caption{Probability densities $f_1(n, b_1, w_1)(t)$, with $n = B + w_1$, $b_1 = 30$, $w_1=20$, and $w_2 = 0$, of the waiting times for a Type~$1$ patient arriving in a congested system with $N = 100$, $B = 80$, so that they are the $20^{th}$ tagged Type~$1$ patient in the waiting area. These densities are computed assuming the parameters $\lambda_1 = 5.7961$, $\lambda_2 = 17.9039$, $\mu_1 = 0.1517$, $\mu_2 = 0.4113$, and $\beta_{2 \to 1} = 3$ (see Table~\ref{tab:QBDparameters}). The left, middle, and right figures show the density curves under priority policies $(r_1, r_2) = (0.8, 0.2)$, $(0.6, 0.4)$, and $(0.5, 0.5)$, respectively. Corresponding mean waiting times $\mathbb{E}(T_1(n, b_1, w_1))$, evaluated using the CTMC model, are $1.0267$, $1.2192$, and $1.3198$ days, respectively.}	\label{conditional_densities_T1_diff_r1_r2}
\end{figure}

\begin{figure}[!htbp]
	\centering
	{\includegraphics[width=5.1 cm,height=5.1 cm]{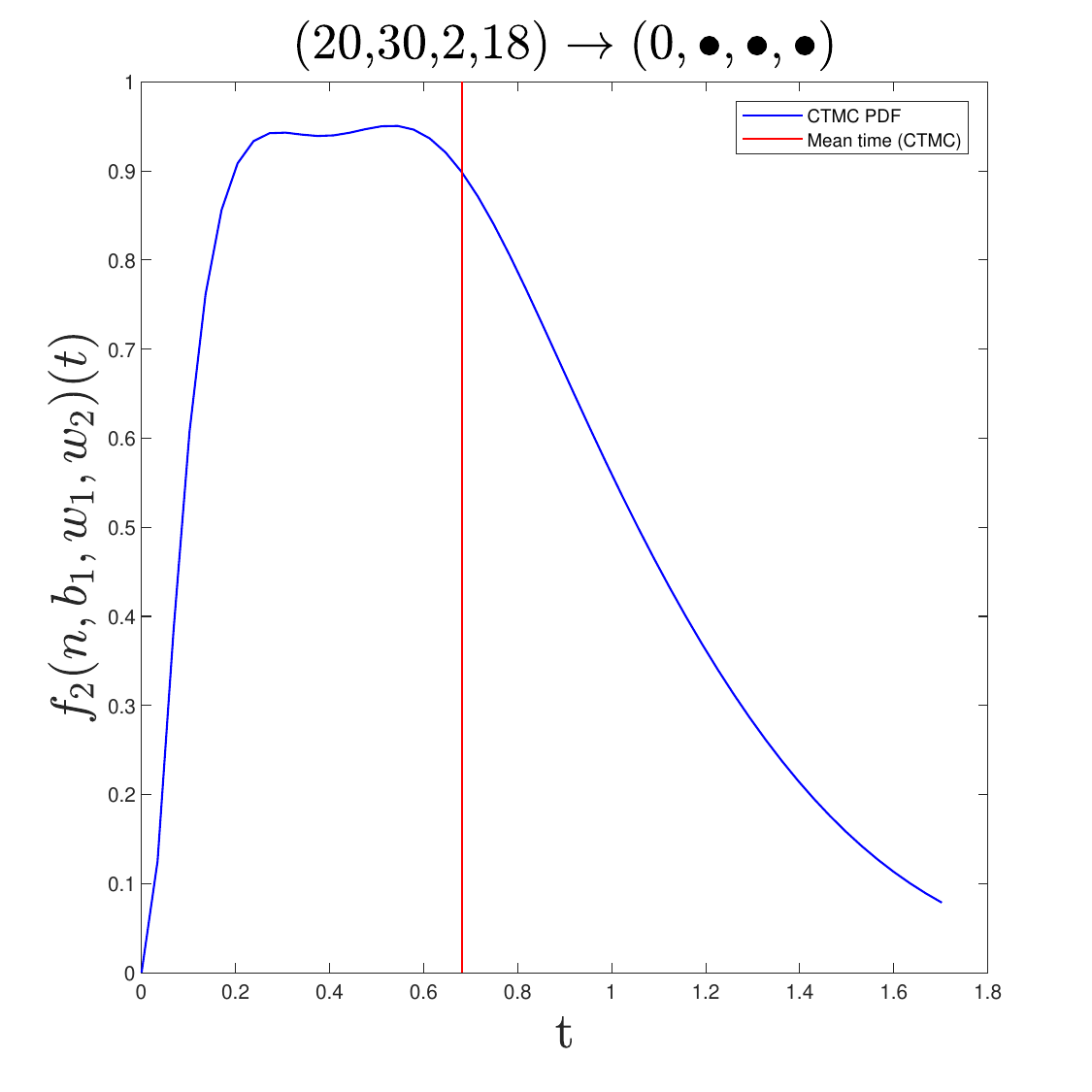}}
    \quad
	{\includegraphics[width=5.1 cm,height=5.1 cm]{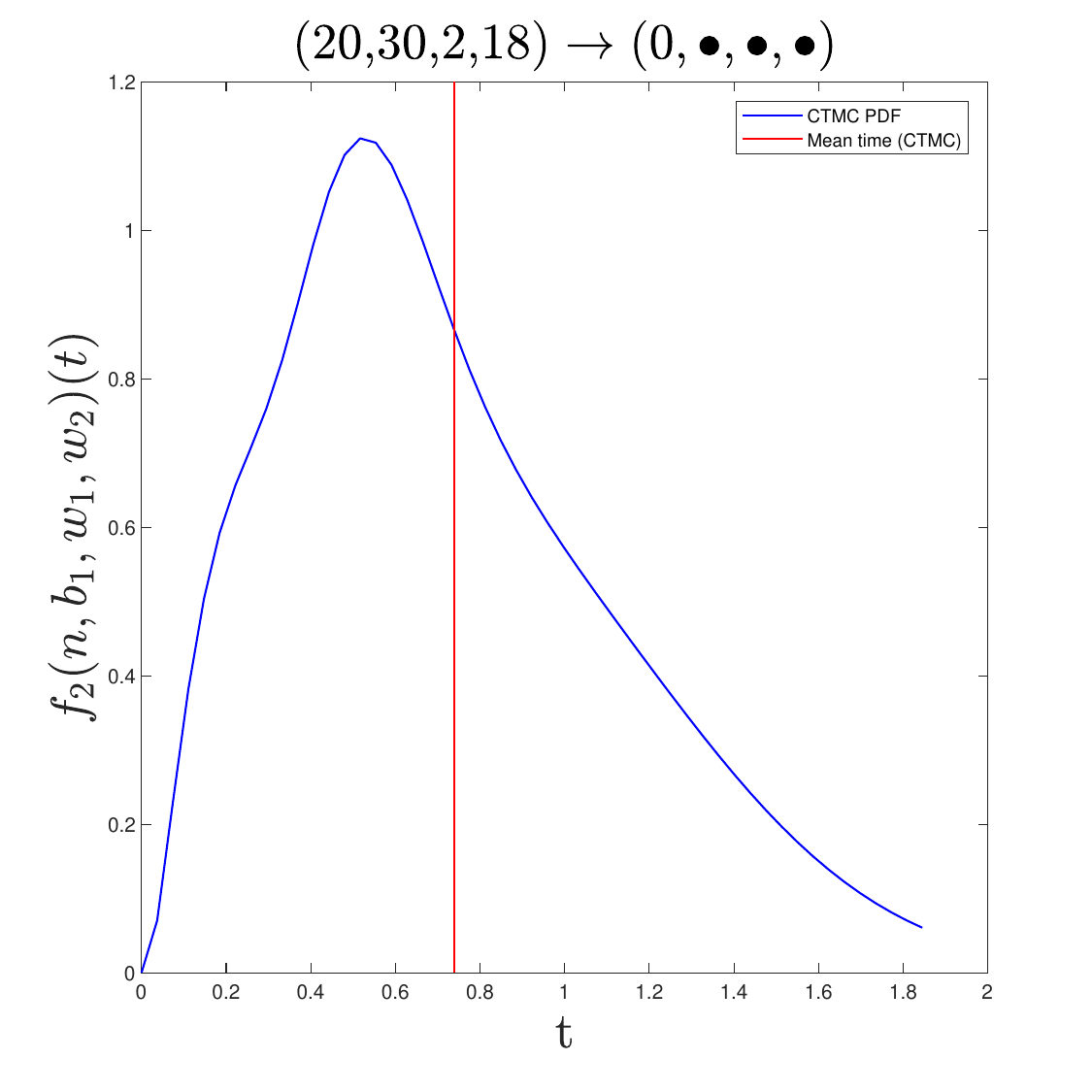}}    
    \quad
	{\includegraphics[width=5.1 cm,height=5.1 cm]{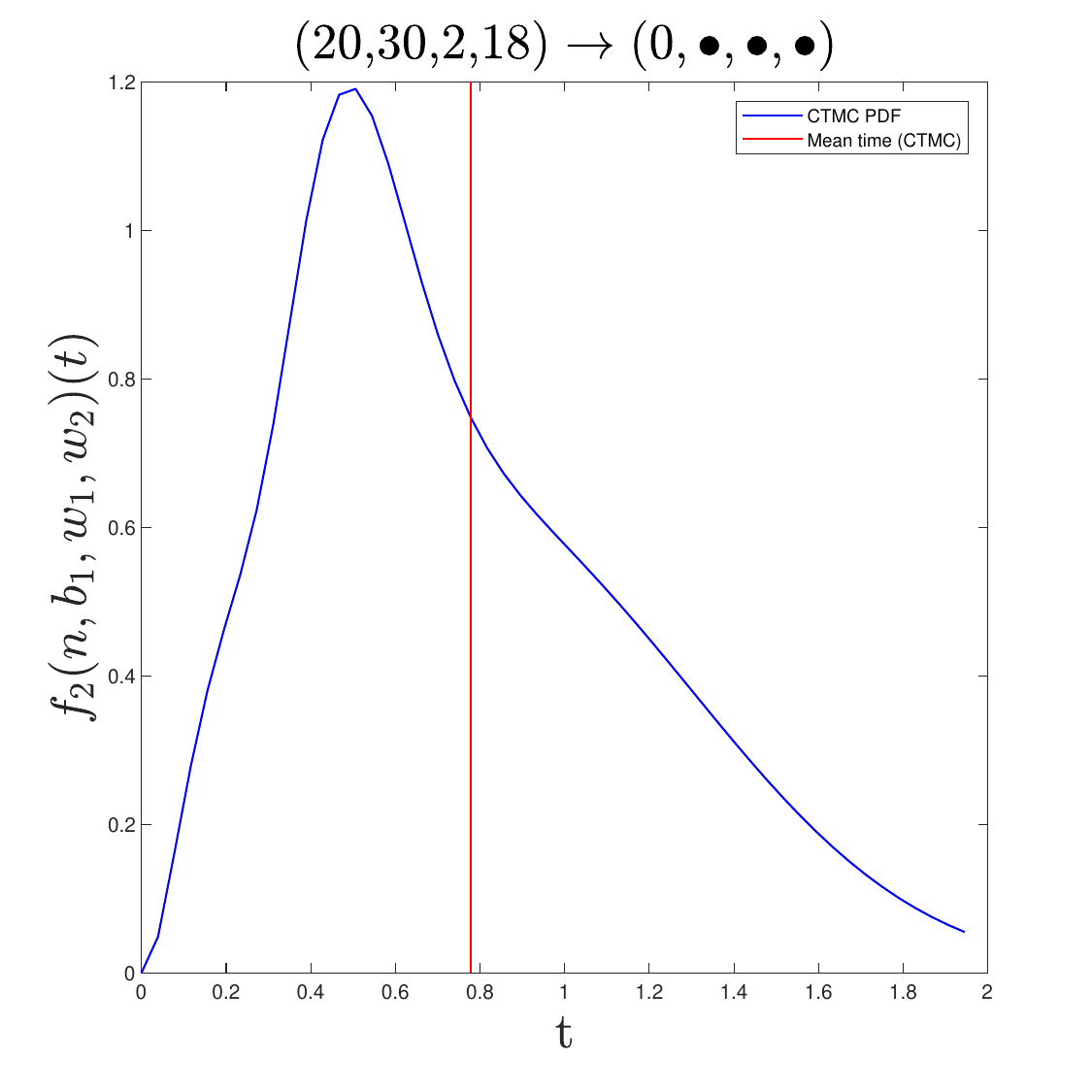}}    
	\caption{Probability densities $f_2(n, b_1, w_1, w_2)(t)$, with $n = B+w_1+w_2$, $b_1 = 30$, $w_1=2$, and $w_2 = 18$, of the waiting times for a Type~$2$ patient arriving in a congested system with $N = 100$, $B = 80$, so that they are the $20^{th}$ tagged Type~$2$ patient in the waiting area. These densities are computed assuming the parameters $\lambda_1 = 5.7961$, $\lambda_2 = 17.9039$, $\mu_1 = 0.1517$, $\mu_2 = 0.4113$, and $\beta_{2 \to 1} = 3$ (see Table~\ref{tab:QBDparameters}). The left, middle, and right figures show the density curves under priority policies $(r_1, r_2) = (0.8, 0.2)$, $(0.6, 0.4)$, and $(0.5, 0.5)$, respectively. Corresponding mean waiting times $\mathbb{E}(T_2(n,b_1,w_1,w_2))$, evaluated using the CTMC model, are $0.6808$, $0.7380$, and $0.7782$ days, respectively.}	\label{conditional_densities_T2_diff_r1_r2}
\end{figure}

\section{Long-run waiting times for Type~$1$ and Type~$2$ patients ($\beta_{2\to1}\geq 0$, $r_1,r_2\in[0,1]$)}
\label{longrun_waiting_times_beta_r1_r2}

We evaluate the distribution of the long-run waiting times for Type~$1$ (complex) and Type~$2$ (other) patients. Table~\ref{tab:long_run_waiting_times} provides long-run mean waiting times $\mathbb{E}(T_i)$ for Type~$i$, $i=1,2$, patients under various priority policies $(r_1,r_2)$ and type-change $\beta_{2\to 1}$, highlighting the sensitivity of waiting times to priority weights. As expected, the mean waiting time $\mathbb{E}(T_1)$ for Type~$1$ patients increases as their priority ($r_1$) decreases, for $\beta_{2 \to 1} = 0, 3$. On the other hand, the mean waiting time $\mathbb{E}(T_2)$ for Type~$2$ patients decreases as their priority ($r_2$) increases, for $\beta_{2 \to 1} = 0$. This is because the access to beds for Type~$2$ patients is minimised as their priority decreases. The mean waiting time $\mathbb{E}(T_2)$ for Type~$2$ patients increases as their priority ($r_2$) increases, for $\beta_{2 \to 1} = 3$. This is because the Type~$2$ patients may change to Type~$1$, and so the priority of the Type~$2$ patients ultimately decreases and the mean waiting time increases. These findings underscore the importance of carefully balancing priority weights to manage patient flow effectively and equitably.

The plots on the left in Figures~\ref{longrun_waittime_QBDCTMC_r1_1_r2_0_beta_0}~$\&$~\ref{longrun_waittime_QBDCTMC_r1_1_r2_0_beta_3} are identical reflecting that under the priority $r_1=1$, Type~$2$$\to$Type~$1$ change does not affect waiting times of Type~$1$ patients. From Figure~\ref{longrun_waittime_QBDCTMC_r1_1_r2_0_beta_0} (right) and Figure~\ref{longrun_waittime_QBDCTMC_r1_1_r2_0_beta_3} (right), we observe that the waiting times of Type~$2$ patients may be longer than the mean under the policy $(r_1,r_2)=(1,0)$ and $\beta_{2\to1}=0$, whereas the waiting times of Type~$2$ patients may be less than the mean under the policy $(r_1,r_2)=(1,0)$ and $\beta_{2\to1}=3$. This highlights the positive impact of considering type-change, which indicates a reduction in the waiting time for Type~$2$ patients.

\begin{table}[H]
\centering
\begin{tabular}{l|llll|llll} 
\toprule
Policy & $r_1=1$  & $r_1=0.8$ & $r_1=0.6$ & $r_1=0.5$ & $r_1=1$ & $r_1=0.8$ & $r_1=0.6$ & $r_1=0.5$ \\

& $\beta_{2\to1}=0$  & $\beta_{2\to1}=0$ &  $\beta_{2\to1}=0$ & $\beta_{2\to1}=0$ &  $\beta_{2\to1}=3$  &  $\beta_{2\to1}=3$ &  $\beta_{2\to1}=3$ & $\beta_{2\to1}=3$ \\ 

\midrule
$\mathbb{E}(T_1)$    & $0.0366$    & $0.0489$  & $0.0741$    & $0.0997$   & $0.0366$  & $0.0489$ & $0.0737$ & $0.0983$ \\

$\mathbb{E}(T_2)$    & $0.4191$    & $0.4151$  & $0.4072$    & $0.3993$   & $0.3152$  & $0.3125$ & $0.3273$ & $0.3419$ \\

\bottomrule
\end{tabular}
\caption{Long-run mean waiting times $\mathbb{E}(T_i)$ for Type~$i$, $i=1,2$, patients corresponding to different priority policies as defined by the pair ($r_1,r_2$). Note that $r_2=1-r_1$, and the parameters used are given by $\lambda_1=5.7961$, $\lambda_2=17.9039$, $\mu_1=0.1517$, $\mu_2=0.4113$.}
\label{tab:long_run_waiting_times}
\end{table}

\begin{figure}[!htbp]
    \centering

    \mbox{\includegraphics[scale=0.4]{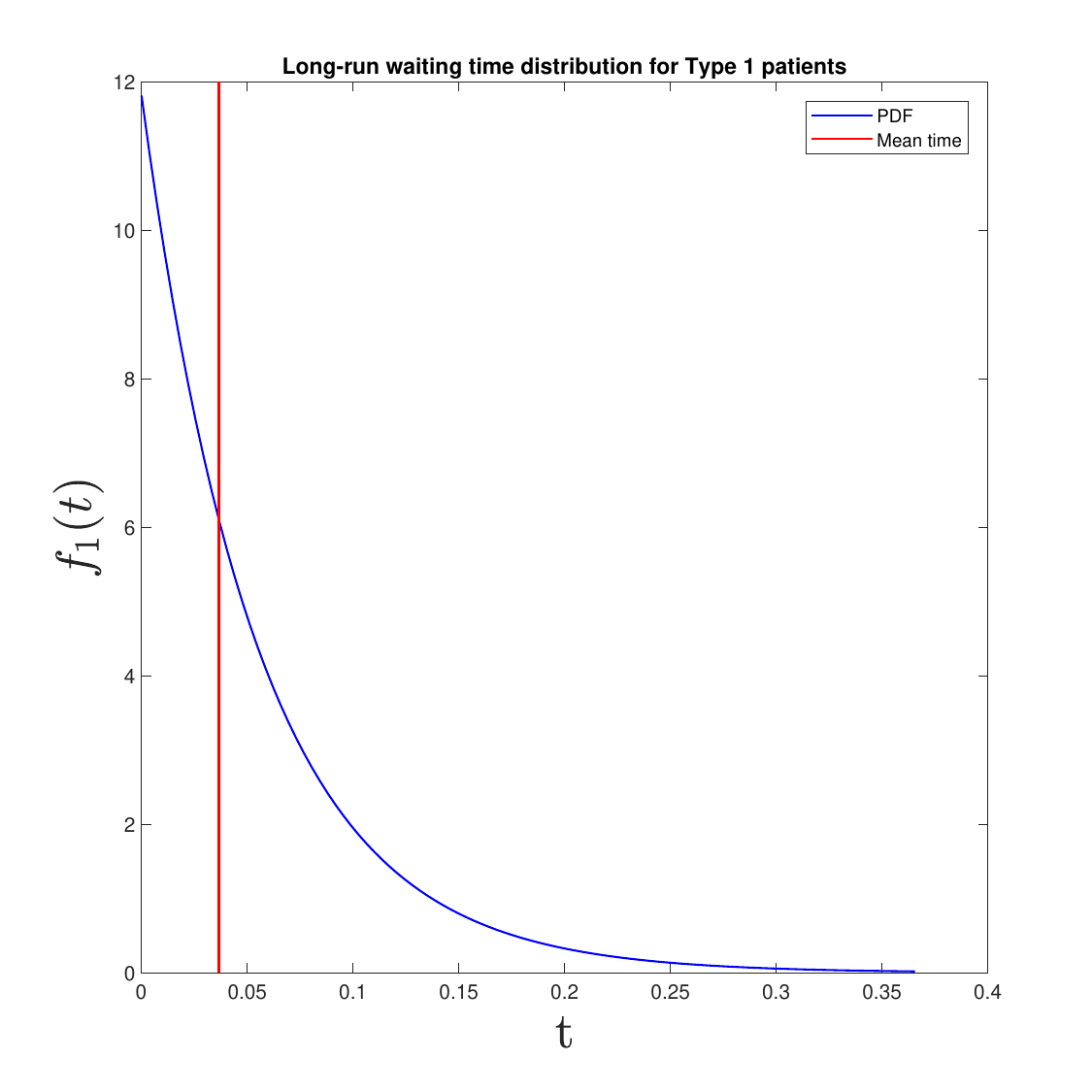}}
    \quad
    \mbox{\includegraphics[scale=0.4]{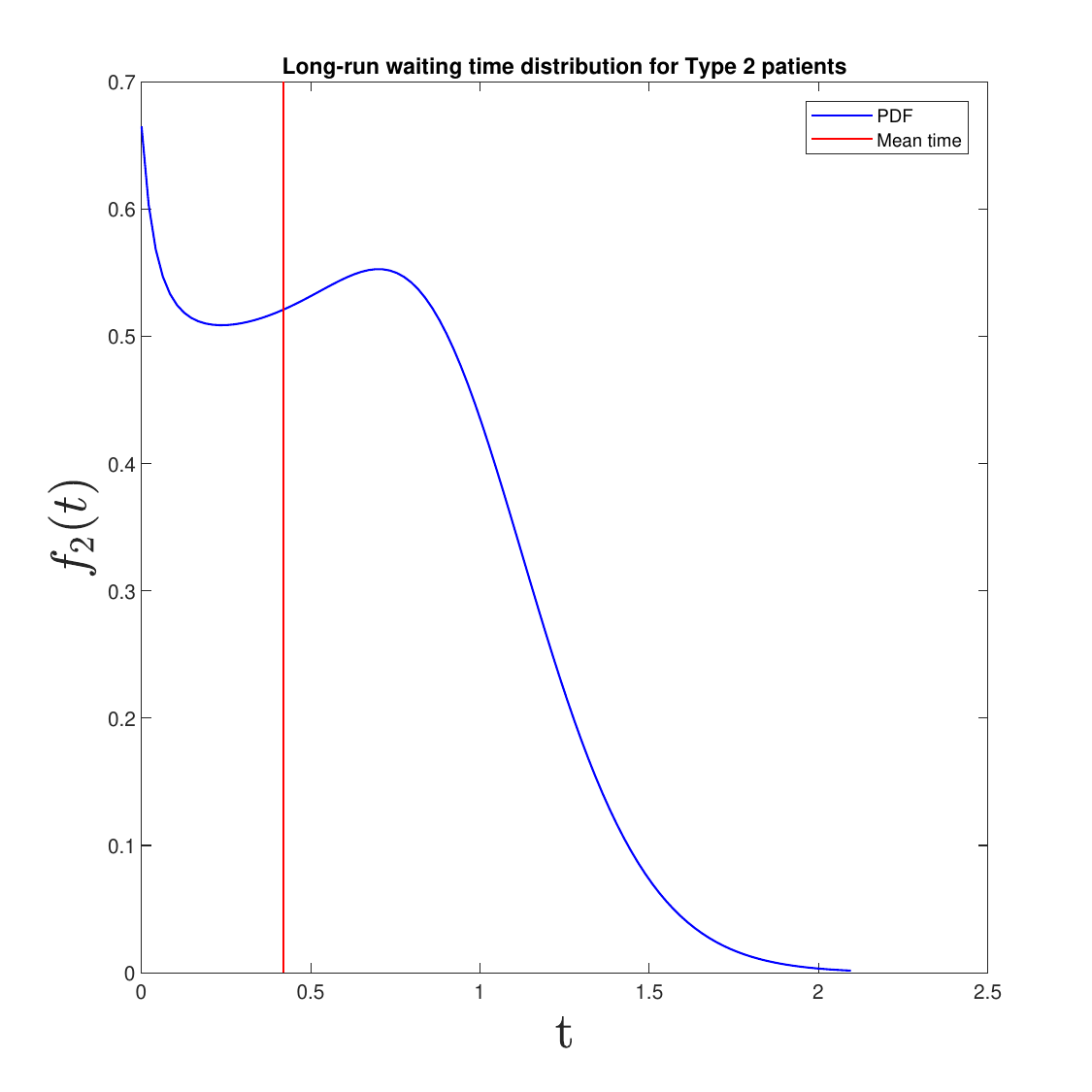}}

    \caption{Long-run probability densities $f_i(t)$ of the waiting time for Type~$1$ (left figure) and Type~$2$ (right figure) patients, evaluated using the CTMC, at the moment when they arrive in the system of capacity of $N=100$ and $B=80$, under the priority policy $r_1=1$, $r_2=0$. The parameters used are given by $\lambda_1=5.7961$, $\lambda_2=17.9039$, $\mu_1=0.1517$, $\mu_2=0.4113$, and $\beta_{2\to 1}=0$. The corresponding long-run mean waiting times for Type~$1$ and Type~$2$ patients are $\mathbb{E}(T_1)=0.0366$ and $\mathbb{E}(T_2)=0.4191$ days, respectively.}  
	\label{longrun_waittime_QBDCTMC_r1_1_r2_0_beta_0}
\end{figure}

\begin{figure}[H]
    \centering

    \mbox{\includegraphics[scale=0.4]{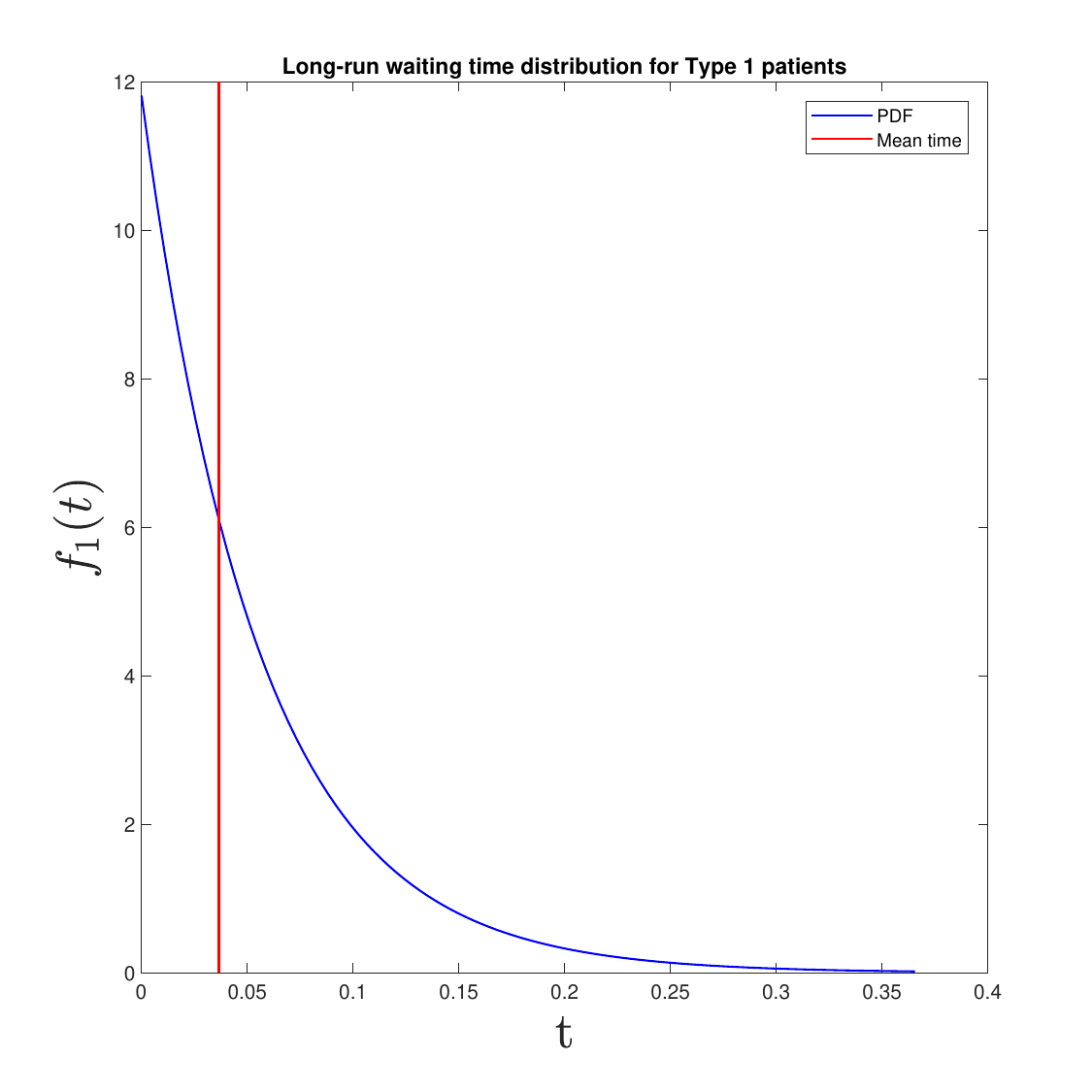}}
    \quad
    \mbox{\includegraphics[scale=0.4]{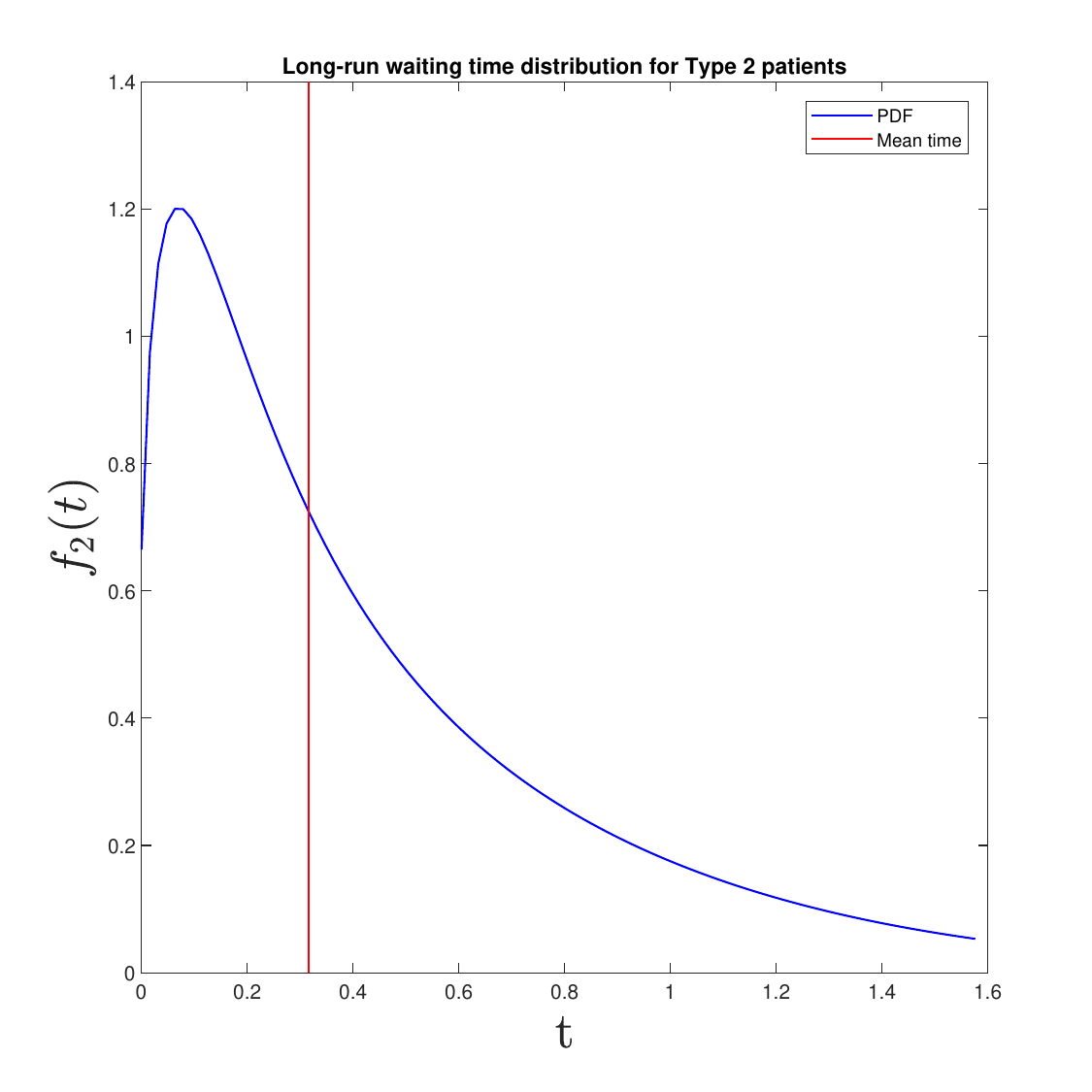}}

    \caption{Long-run probability densities $f_i(t)$ of the waiting time for Type~$1$ (left figure) and Type~$2$ (right figure) patients, evaluated using the CTMC, at the moment when they arrive in the system of capacity of $N=100$ and $B=80$, under the priority policy $r_1=1$, $r_2=0$. The parameters used are given by $\lambda_1=5.7961$, $\lambda_2=17.9039$, $\mu_1=0.1517$, $\mu_2=0.4113$, and $\beta_{2\to 1}=3$. The corresponding long-run mean waiting times for Type~$1$ and Type~$2$ patients are $\mathbb{E}(T_1)=0.0366$ and $\mathbb{E}(T_2)=0.3152$ days, respectively.}
    \label{longrun_waittime_QBDCTMC_r1_1_r2_0_beta_3}
\end{figure}

\section{Discussion}\label{discussion}

In the case when Type~$2$ patients change to Type~$1$ due to health deterioration, we found that increasing priority for Type~$1$ patients in the waiting area can reduce waiting times for both types of patients, because Type~$2$ patients eventually benefit from accelerated access to service due to type change. At the system level, however, increasing Type~$1$ priority has adverse congestion effects, that is, increasing the priority for Type~$1$ patients in the waiting area increases the proportion of time the system operates at full capacity, leading to higher redirection rates $\lambda_{\text{redirect}}$ for the new incoming patients. This gives decision makers important insights that while prioritising Type~$1$ patients improves the early access of both types of patients to the service, it does so at the cost of increased congestion and redirection.

Moreover, as the priority for Type~$1$ patients increases, beds are allocated more frequently to Type~$1$ patients, which leads to fewer Type~$1$ patients in the waiting area. Conversely, Type~$2$ patients are selected less frequently for service and therefore experience longer waiting times prior to service as compared to Type~$1$, resulting in an increased number of Type~$2$ patients in the waiting area, even though a subset of them may eventually convert to Type~$1$ and benefit from priority access to service. Consequently, simply increasing the priority for Type~$1$ patients is unlikely to be an appropriate policy choice in congested systems. Decision-makers must explicitly balance priority weights against the acceptability of operating the system at full capacity and redirecting incoming patients.

We further analyse long‑run waiting times of new arriving patients. When Type~$2$ patients' health deteriorate and they change to Type~$1$ ($\beta_{2\to1}\neq 0$), we observe that increasing the priority for Type~$1$ patients reduces the waiting times of new incoming patients of both types in the long‑run. This occurs because type-change enables deteriorating Type~$2$ patients to access priority service, thereby preventing prolonged waiting times. In contrast, when type-change is not considered, we found that increasing the priority for Type~$1$ patients reduces the waiting time of new incoming Type~$1$ patients but increases the waiting time of new incoming Type~$2$ patients in the long‑run. In this case, priority adjustments induce a direct trade‑off between the two patient classes: improving access for one type necessarily worsens access for the other. As a result, efforts to reduce waiting times for incoming Type~$1$ patients come at the expense of longer waiting times for incoming Type~$2$ patients, and vice versa.

These findings highlight a useful clinically relevant operational insight. If hospital managers do not update the type and so the priority of Type~$2$ patients after their health has sufficiently deteriorated, priority policies alone may not ensure timely access to beds for such patients. Conversely, updating the type of deteriorating Type~$2$ patients to Type~$1$ transforms priority adjustments from a zero‑sum trade‑off into a mechanism that can simultaneously improve waiting times across patient types. This highlights the importance of adjusting the priority weights with evolving patient types, rather than relying on static priority labels in congested healthcare systems.

\section{Conclusion}\label{sec:Conclusion}

In this paper, we studied customer's waiting times in two-class finite-capacity multi-server systems with dynamic admission priorities, dynamically evolving customer type, and abandonment. We constructed a continuous-time Markov chain (CTMC) model for this system and developed a methodology based on Krylov subspace approximation methods to evaluate conditional distributions of the waiting times. Next, we developed a quasi-birth-and-death (QBD) model for this system and derived theoretical expressions based on matrix-analytic methods to evaluate the conditional distributions of the waiting times. Further, we developed expressions for the long-run distributions of the waiting times, based on both models.

We illustrated the application potential of our methodology in multi-server systems, with numerical examples based on a large dataset obtained from a tertiary referral hospital in Australia. We considered two types of patients, complex and other, and compared the conditional waiting times obtained from CTMC and QBD models in the cases when complex patients have absolute admission priority. The results showed that the conditional distributions of the waiting times from the CTMC model approximates those obtained from the QBD model with an error smaller than $10^{-7}$. Further, we evaluated the long-run distributions of the waiting times as well. 

The insights from our analysis demonstrated how admission priorities and customer type evolution can influence both conditional and long-run waiting time distributions. The methodologies developed in this paper provide a practical framework for analysing waiting times and supporting operational decision-making in complex service systems such as hospitals.

\section{Limitations and future directions}\label{limitations}

We acknowledge that empirical validation of the model outputs can be carried out using the existing statistical validation techniques. Such validation would enhance the credibility and reliability of the proposed modeling approach. However, a direct validation of the waiting-time distributions against the empirical data was not feasible in this study, as the available hospital dataset did not include queue-position-dependent waiting-time information.

Several avenues for future research emerge from this study. One natural extension is to generalise the model to accommodate non-exponential service time and abandonment distributions, which would enhance its applicability in healthcare settings where  these distributions often deviate from the memoryless assumption. Expanding the framework to include more than two patient types or to model multi-stage treatment processes such as triage, diagnostics, and treatment would better reflect the complexity of real-world hospital workflows. Future work could also incorporate state-dependent and time-varying arrival patterns, such as arrivals driven by seasonal or epidemic trends, which would further enhance the robustness and realism of the model in practical applications.

\section{Statements and declarations}\label{Sec:Declarations}

\noindent{\bf Data}

Data used in the paper was obtained following ethical approval from the  University of Tasmania Human Research (approval number H23633) and site-specific approval from the Research Governance Office of the Tasmanian Health Service.

\noindent{\bf Authorship contribution statement}

This paper contributes to a chapter in the PhD project by~\textcite{KhokharPhD}. The following are the contributions of the authors, Muhammad Abdullah Khokhar (MAK), Ma\l gorzata M. O'Reilly (MMO), and Richard Turner (RT).
\begin{itemize}
\item Problem formulation and methodology development: MAK, MMO, and RT;
\item Derivation of the theoretical expressions and algorithms: MAK, MMO;
\item  Coding: MAK, MMO;
\item Data analysis: MAK;
\item Numerical analysis: MAK;
\item Conceptualisation, Practical background: MAK, RT;
\item Conceptualisation, Mathematical background: MAK, MMO;
\item Write-up and edits: MAK, MMO, and RT.
\end{itemize}

\bigskip\noindent{\bf Declaration of competing interests}\\
\noindent The authors have no competing interests to declare that are relevant to the content of this article.

\bigskip\noindent{\bf Funding}\\
\noindent This research did not receive any specific grant from funding agencies in the public, commercial, or not-for-profit sectors.

\printbibliography
\end{document}